\documentclass[11pt]{article}
\usepackage[letterpaper, margin=1.0in]{geometry}
\usepackage[utf8]{inputenc}
\usepackage{tabularx}
\usepackage{microtype}
\usepackage{hyperref}
\usepackage{amsmath}

\usepackage{amssymb,amsfonts,amsthm}
\usepackage{mathtools,thmtools}
\usepackage{enumerate}
\usepackage{color}
\usepackage{xspace}
\usepackage{characters}

\usepackage{tikz}
\usetikzlibrary{shapes,shapes.callouts,shadows,arrows,backgrounds,%
matrix,patterns,arrows,decorations.pathmorphing,decorations.pathreplacing,%
positioning,fit,calc,decorations.text,arrows.meta,quotes,%
decorations.markings%
}
\usepackage{todonotes}
\usepackage{comment}
\usepackage{cleveref}
\usepackage[font=small,labelfont=bf]{caption}

\theoremstyle{plain}
  \newtheorem{theorem}{Theorem}
  \newtheorem{corollary}[theorem]{Corollary}
  
  \newtheorem{lemma}[theorem]{Lemma}
  \newtheorem{proposition}[theorem]{Proposition}
  
  \newcounter{theoremspecial}
  \newtheorem{claimspecial}{Claim}[theoremspecial]
  \newtheorem*{theorem*}{Theorem}
  \newtheorem*{corollary*}{Corollary}
  \newtheorem*{lemma*}{Lemma}
  \newtheorem*{proposition*}{Proposition}
  \newtheorem*{claim*}{Claim}
  
\theoremstyle{definition}
  \newtheorem{definition}[theorem]{Definition}
  \newtheorem{example}[theorem]{Example}

  \newtheorem*{definition*}{Definition}
  \newtheorem*{example*}{Example}
  \newtheorem*{question*}{Question}
  \newtheorem*{conjecture*}{Conjecture}
  \newtheorem*{remark*}{Remark}

\AddToHook{env/lemma/begin}{\crefalias{theorem}{lemma}}

\newenvironment{claimproof}
{\noindent {\em Proof of claim:} }
{\hfill $\diamond$ \smallskip}

\DeclarePairedDelimiter\norm{\lVert}{\rVert}

\mathchardef\mhyphen="2D

\newcommand*{\eps}{\varepsilon}
\renewcommand{\epsilon}{\varepsilon}

\newcommand{\bigoh}{O}

\newcommand{\cc}[1]{{\mbox{\textnormal{\textsf{#1}}}}\xspace}  %
\newcommand{\SB}{\{\,}
\newcommand{\SM}{\;{|}\;}
\newcommand{\SE}{\,\}}

\newcommand{\NP}{\cc{NP}}
\newcommand{\FPT}{FPT\xspace}

\newcommand{\Weft}{{\cc{W}}}
\newcommand{\W}[1]{{\Weft}{{[#1]}}}

\newcommand*{\ETH}{\textsc{ETH}\xspace}

\DeclareMathOperator{\cost}{cost}

\usepackage{boxedminipage}

\newcommand{\pbDefP}[4]{%
  \noindent
  \begin{center}
  \begin{boxedminipage}{0.98 \columnwidth}
  {\sc #1}\\[5pt]
  \begin{tabular}{l p{0.70 \columnwidth}}
  {\sc Instance}: & #2\\
  {\sc Parameter:} & #3\\
  {\sc Question}: & #4
  \end{tabular}
  \end{boxedminipage}
  \end{center}
}

\newcommand{\pbDefGap}[5]{%
  \noindent
  \begin{center}
  \begin{boxedminipage}{0.98 \columnwidth}
  {\sc #1}\\[5pt]
  \begin{tabular}{l p{0.80 \columnwidth}}
    {\sc Instance}: & #2\\
    {\sc Parameter:} & #3\\
    {\sc Goal}: & Distinguish between the following cases: \\
     & {(\textsc{Yes})} #4 \\
     & {(\textsc{No})} #5
  \end{tabular}
  \end{boxedminipage}
  \end{center}
}

\newcommand{\pbDefGapnonparam}[4]{%
  \noindent
  \begin{center}
  \begin{boxedminipage}{0.98 \columnwidth}
  {\sc #1}\\[5pt]
  \begin{tabular}{l p{0.80 \columnwidth}}
    {\sc Instance}: & #2\\
    {\sc Goal}: & Distinguish between the following cases: \\
     & {(\textsc{Yes})} #3 \\
     & {(\textsc{No})} #4
  \end{tabular}
  \end{boxedminipage}
  \end{center}
}

\newcommand{\CSP}{\textsc{CSP}}

\newcommand{\MinCSP}{\textsc{MinCSP}}
\newcommand{\MinSat}{\textsc{MinSat}}
\newcommand{\Sat}{\textsc{Sat}}

\newcommand*{\GAP}[1]{\ensuremath{\textsc{Gap}_{#1}}}

\newcommand{\Htkp}{$H_{t}$}
\newcommand{\Htk}{H_{t,k}}

\newcommand{\svens}{\textsc{\Htkp-MSI}\xspace}

\newcommand{\unit}{\ensuremath{\overline{13}\xspace}}

\newcommand{\gapsven}{\textsc{$\GAP{\eps}$
    \Htkp-Multicoloured Subgraph
    Isomorphism}\xspace}

\newcommand{\gapsvens}{\textsc{$\GAP{\eps}$ \Htkp-MSI}\xspace}

\newcommand{\pSplit}{\textsc{$\GAP{p-\eps}$ $p$-Split Min Cut}\xspace}
\newcommand{\pSplitng}{\textsc{$p$-Split Min Cut}\xspace}

\DeclareMathOperator{\mincost}{mincost}

\mathchardef\mhyphen="2D

\newcommand*{\range}[2]{\ensuremath{\{#1,\dots,#2\}}}

\newcommand{\eqn}{\textsf{eqn}}

\newcommand{\ed}{\textsf{ed}}

\usepackage[noend]{algpseudocode}
\usepackage{algorithm,algorithmicx}

\algnewcommand\algorithmicinput{\textbf{Input:}}
\algnewcommand\INPUT{\item[\algorithmicinput]}

\algnewcommand\algorithmicoutput{\textbf{Output:}}
\algnewcommand\OUTPUT{\item[\algorithmicoutput]}

\newcommand{\linz}[2]{\ensuremath{\textsc{{#1}-Lin}(\ZZ_{#2})}}
\newcommand{\minlinz}[2]{\ensuremath{\textsc{Min-{#1}-Lin}(\ZZ_{#2})}}
\newcommand*{\minlin}[2]{\textsc{Min-\ensuremath{#1}-Lin(\ensuremath{#2})}}

\def \RandomCover {\textsf{RandomCover}\xspace}

\def \SolveMinLin {\textsf{SolveMinLin}\xspace}

\newcommand{\bool}[1]{\ensuremath{{#1}_{\textsf{B}}}}
\newcommand{\enhance}[1]{\ensuremath{{#1}^{+}}}

\newcommand{\sind}[2]{\mathrm{s}_{#2}(#1)}
\DeclareMathOperator{\ord}{ord}

\DeclareMathOperator{\suff}{suff}
\DeclareMathOperator{\wgt}{w}

\newcommand{\bvar}[4]{\ensuremath{
  \mathbf{x}_{#1}
  [{#2}, {#3}: p^{#4}]
  }\xspace}
\newcommand{\coset}[3]{\ensuremath{\textsf{coset}({#1}, {#2}: p^{#3})}\xspace}
\newcommand{\bvarsingle}[2]{\ensuremath{
  \mathbf{x}_{#1}
  [{#2}]
  }\xspace}

\title{Optimal FPT-Approximability for Modular Linear Equations}

 \author{
   Konrad K. Dabrowski\thanks{Newcastle University, UK, \texttt{konrad.dabrowski@newcastle.ac.uk}} \and
   Peter Jonsson\thanks{Link{\"o}ping University, Sweden, \texttt{peter.jonsson@liu.se}} \and
   Sebastian Ordyniak\thanks{University of Leeds, UK, \texttt{sordyniak@gmail.com}} \and
   George Osipov\thanks{Royal Holloway, University of London, UK, \texttt{george.osipov@pm.me}} \and
   Magnus Wahlstr{\"o}m\thanks{Royal Holloway, University of London, UK, \texttt{Magnus.Wahlstrom@rhul.ac.uk}}
 }

\date{}

\begin{document}

\maketitle

\begin{abstract}
  We show optimal FPT-approximability results
  for solving almost satisfiable systems of modular linear equations,
  completing the picture of the parameterized complexity and FPT-approximability landscape for the \textsc{Min-{$r$}-Lin}$(\mathbb{Z}_m)$ problem for every $r$ and $m$.
  
  In \textsc{Min-{$r$}-Lin}$(\mathbb{Z}_m)$,
  we are given a system $S$ of linear equations modulo $m$, 
  each on at most $r$ variables,
  and the goal is to find a subset $Z \subseteq S$ of minimum cardinality
  such that $S - Z$ is satisfiable.
  The problem is NP-hard, and UGC-hard to approximate within
  any constant factor for every $r \geq 2$ and $m \geq 2$,
  which motivates studying it through the lens of parameterized complexity
  with solution size as the parameter.
  From previous work (Dabrowski~et~al. SODA'23/TALG and ESA'25) we know that
  \textsc{Min-{$r$}-Lin}$(\mathbb{Z}_m)$ is W[1]-hard to FPT-approximate
  within any constant factor when $r \geq 3$, and that
  \textsc{Min-{$2$}-Lin}$(\mathbb{Z}_m)$ is in \FPT\ when $m$ is prime
  and W[1]-hard when $m$ has at least two distinct prime factors.
  The case when $m = p^d$ for some prime $p$ and $d \geq 2$
  has remained an open problem.
  We resolve this problem in this paper and prove the following:
  \begin{itemize}
      \item \textbf{FPT algorithm for prime powers}: 
        we prove that \textsc{Min-{$2$}-Lin}$(\mathbb{Z}_{p^d})$ is in \FPT for every prime $p$ and $d \geq 1$. This implies that
      \textsc{Min-{$2$}-Lin}$(\mathbb{Z}_{m})$ can be FPT-approximated 
        within a factor of $\omega(m)$,
        where $\omega$ is the number of distinct prime factors of $m$.
      \item \textbf{Matching lower bound}: we show that, under the ETH,
        \textsc{Min-{$2$}-Lin}$(\mathbb{Z}_m)$ cannot be FPT-approximated 
        within $\omega(m) - \varepsilon$ for any $\varepsilon > 0$.
  \end{itemize}

  Our main algorithmic contribution is a new technique coined
  \emph{balanced subgraph covering}, which generalizes
  important balanced subgraphs of Dabrowski~et~al.\ (SODA'23/TALG) and
  shadow removal of Marx~and~Razgon\ (STOC'11/SICOMP).
  In a nutshell, we relax \textsc{Min-{$2$}-Lin}$(\mathbb{Z}_{p^d})$ by allowing
  ambiguous assignments, which may map variables to subsets of values 
  (as opposed to only individual values).
  We then apply balanced subgraph covering to impose additional constraints on the relaxation,
  solve it by a reduction to \textsc{MinSat},
  and use the additional constraints to convert
  any ambiguous assignment into a proper value assignment of the same cost.
  For the lower bounds, we develop a framework for proving
  optimality of FPT-approximation factors under the ETH.
  We consider a variant of $st$-\textsc{MinCut} where the edges
  of the input graph are partitioned into bundles of size $\ell$ and each bundle
  can be removed at cost $1$.
  Starting from a recent inapproximability result of Bafna~et~al.\ (STOC'25), we prove that $st$-\textsc{MinCut} with edge bundles of size $\ell \geq 2$
  cannot be FPT-approximated within $\ell - \varepsilon$ under the ETH,
  and then reduce the latter problem into 
  \textsc{Min-{$2$}-Lin}$(\mathbb{Z}_{m})$ with $\omega(m) = \ell$.
\end{abstract}
\thispagestyle{empty}

\newpage

\tableofcontents
\thispagestyle{empty}

\newpage

\setcounter{page}{1}

\section{Introduction} \label{sec:intro} 

Systems of linear equations are omnipresent in computer science and 
mathematics~\cite{grcar2011ordinary}.
In particular, equations and congruences over 
the finite ring of integers modulo $m$ ($\ZZ_m$)
are of central importance in number theory, but also have many
applications in computer science, see e.g.~\cite{Blake:ic72,Dawar:etal:lmcs2013,Ding:etal:CRT,Yan:numbertheory}.
An example is the ongoing effort to understand the approximability of
finite-domain constraint satisfaction problems, where the approximability of equations
over $\ZZ_m$ plays a key role. For instance, the {\em unique games conjecture}
(UGC)~\cite{khot2002power}
and its underlying theory is intrinsically based on
linear equations over $\ZZ_m$ (see the book by Widgerson~\cite[Section~4.3.2]{Wigderson:math}). 
Linear equations over a modular ring $\ZZ_m$ can be solved in polynomial time
by methods such as Gaussian elimination, but such methods are not suited for handling 
inconsistent systems of equations.
We consider the  $\minlin{r}{R}$ problem, where one wants to find an assignment to a system of linear equations over the ring $R$
that violates the minimum number of equations, and where each equation contains at most $r$
distinct variables. 
We note that $\minlin{r}{R}$ for $r \in \NN$ and finite ring $R$ is a special case of
\textsc{MinCSP$(\Gamma)$} for a finite constraint language~$\Gamma$.
This, and the more general problem \textsc{Valued CSP}, have been widely studied
from various perspectives~\cite{Bonnet:etal:esa2016,Dabrowski:etal:ipec2023,Kim:etal:sicomp2025,Kolmogorov:etal:sicomp2017,Kratsch:etal:compscirev2026,Osipov:Wahlstrom:esa2023,raghavendra2008optimal}.

The problem $\minlin{r}{R}$
 is \NP-hard even when $r = 2$ and $R$ is the simplest
 nontrivial ring $\ZZ_2$~\cite{Kolmogorov:etal:sicomp2017}.
Two ways of coping with \NP-hardness 
are approximation and parameterized algorithms,
but neither of them seem sufficient for $\minlin{r}{\ZZ_m}$.
Even $\minlin{2}{\cdot}$ over finite fields such as $\ZZ_2$
is conjectured to be \NP-hard to approximate within any constant factor under
the UGC:
see Definition~3 in~\cite{khot2016candidate} and the discussion that follows.
The natural parameter for  $\minlin{r}{R}$ is the cost of the optimal solution, i.e.\ the number of violated equations.
Under this parameterization, the $\minlin{3}{R}$ problem is \W{1}-hard
for {\em every} nontrivial ring~\cite[Theorem 44]{minlin25talg}.
The situation is slightly better for 
$\minlin{2}{\ZZ_m}$: it is fixed-parameter tractable 
when $m$ is a prime~\cite[Theorem 33]{minlin25talg} (and $\ZZ_m$ is a field), but \W{1}-hard when $m$ is not a prime power~\cite[Section~6.2]{minlin25talg}. 
However, the parameterized complexity of seemingly simple
cases such as $\minlin{2}{\ZZ_4}$ is not settled by previous work,
and this is a recurring open question~\cite{minlin25talg,Kratsch:etal:compscirev2026,PACS24open}.

This has motivated the use of \emph{parameterized approximation}
for studying $\minlin{2}{\ZZ_m}$~\cite{Dabrowski:etal:arxiv2024,Dabrowski:etal:esa2025}.
Parameterized approximation is an approach that has received rapidly increasing interest (see, for instance,~\cite{EibenRW22ipec,Guruswami:etal:stoc2024,GuruswamiRS24ccc,Lokshtanov:etal:soda2021,lokshtanov2020parameterized,OsipovPW24pointalgebra} and the surveys~\cite{Feldmann:etal:algorithms2020,Marx:tcj2008}).
Informally speaking, \FPT-approximability gives more
time to compute the output (compared to polynomial-time
approximation) and the output 
may be an approximate solution (unlike exact FPT algorithms).
Let $c \geq 1$ be a constant and assume, 
as usual, that the notation $O^*(\cdot)$ hides polynomial factors in the input size.
A factor-$c$ {\em FPT-approximation
algorithm} takes an instance $(I, k)$, runs in
$O^*(f(k))$ time for some computable function $f$, and
either returns that there is no solution of size at most $k$ or returns that there is 
a solution of size at most $c \cdot k$. 
  A decision $c$-approximation procedure for 
  $\minlin{2}{\ZZ_m}$ can be turned into an algorithm
  that returns a $c$-approximate solution using
  self-reducibility; simply note that we have access to
  equations $x=a$ for all elements $a$ in $\ZZ_m$.
  Clearly, a problem is in FPT if and only if it is
  FPT-approximable within a factor of $c=1$.

   The currently best known bounds
for FPT-approximability of $\minlin{2}{\ZZ_m}$
are scattered in the literature and we give a summary below.
Let $m \geq 2$ be a natural number and let $\omega(m)$ denote the number of 
distinct prime factors of $m$.

  \begin{enumerate}
  \item
  If $r \geq 3$, then
  $\minlin{r}{\ZZ_m}$
  is \W{1}-hard to \FPT-approximable within any constant factor~\cite[Theorem 16]{Dabrowski:etal:esa2025}.

  \item
  $\minlin{2}{\ZZ_m}$ is \FPT-approximable within a factor of $2\omega(m)$~\cite[Theorem 1]{Dabrowski:etal:esa2025}.

  \item
  If $m$ is prime, then
  $\minlin{2}{\ZZ_m}$ is in \FPT~\cite[Theorem 33]{minlin25talg}.

  \item
  If $m$ is not a prime power, then $\minlin{2}{\ZZ_m}$ is not FPT-approximable within a factor of $2-\epsilon$ for any $\epsilon > 0$, unless the \ETH is false~\cite[Theorem 45]{Dabrowski:etal:arxiv2024}.
\end{enumerate}

\noindent
We obtain a full picture of constant-factor FPT-approximability of
$\minlin{r}{\ZZ_m}$ (under the ETH) by proving the following.

\begin{theorem} \label{thm:main}
Let $m \geq 2$ be a natural number.
$\minlin{2}{\ZZ_m}$ is  FPT-approximable within a factor of $\omega(m)$ and it is not FPT-approximable within $\omega(m)-\epsilon$ for any $\epsilon > 0$ unless the \ETH is false.
\end{theorem}

\noindent
We are not aware of
any other example in the literature where tight lower and
upper bounds for FPT-approximability have been obtained for such a
wide range of natural problems.
In particular, our result settles the \FPT-status of
the aforementioned $\minlin{2}{\ZZ_4}$ problem.

\paragraph*{Technical overview.} 
Let us give a brief overview of our tools and techniques; more
details are found in
Sections~\ref{sec:intro-coverings}--\ref{sec:intro-hardness}.
We begin with the positive result. Let $m \geq 2$ be a constant. An
instance of $\minlin{2}{\ZZ_m}$ consists (somewhat informally) of a
set of variables $V$, a set $S$ of binary equations over $\ZZ_m$, and
an integer $k$, and the question is whether there is an assignment
$\varphi \colon V \to \ZZ_m$ which satisfies all but at most $k$
equations. Here, a generic binary equation over $\ZZ_m$ takes the form
$au+bv=c \pmod {m}$ where $a, b, c \in \ZZ_m$ and $u, v \in V$ are
variables.

We simplify the situation in several ways.
By~\cite[Proposition~4]{Dabrowski:etal:esa2025}, it is sufficient to prove that
$\minlin{2}{\ZZ_{p^d}}$ (for $d \geq 2$ and an arbitrary prime $p$)
is in \FPT. 
Furthermore, it is known from~\cite[Lemma~5]{Dabrowski:etal:esa2025} that we can
restrict our attention to instances with two types of constraints, a
crisp (unbreakable) constraint $s=1$ and constraints $u=av$ for $u, v
\in V$ and $a \in \ZZ_{p^d}$. For simplicity,
we further separate the latter into equations $u=av$
where $a \bmod p \neq 0$ (i.e. $a$ is a \emph{unit} in $\ZZ_{p^d}$) and equations $u=pv$.
We refer to these as \emph{special} instances. 
Of these equation types, $u=av$ induces a permutation of $\ZZ_{p^d}$, while $u=pv$
acts like a ``coarsening'' of the domain
(e.g.\ for $m=9$, $u=3v$ maps $v \in \{1,4,7\}$ to $u=3$, $v \in \{2,5,8\}$
to $u=6$, and $v \in \{0,3,6\}$ to $u=0$). While the FPT literature contains
many methods for dealing with systems of permutation constraints~\cite{chitnis2016designing,iwata2016half,iwata201801all,wahlstrom2017lp},
their interaction with constraints $u=pv$ causes significant complications.

For our solution to $\minlin{2}{\ZZ_{p^d}}$, we first show a support statement,
starting from the formulation of equation systems as \emph{biased graphs} used in previous work~\cite{minlin25talg,wahlstrom2017lp}
and generalize it into a statement we call \emph{balanced subgraph covering};
see Section~\ref{sec:intro-coverings}. Informally, this result takes
an inconsistent system $G$ of permutation constraints, with a hidden almost
consistent induced subsystem $H$, and returns a consistent subsystem
$G'$ of $G$, with some annotations, such that with sufficient probability
$G'$ captures the hidden system $H$. We thus generalize the \emph{shadow covering},
(a.k.a.\ \emph{shadow removal}) technique of Marx and
Razgon~\cite{marx2014fixed} into the setting of biased graphs.

With this tool in hand, our algorithm has two aspects. 
First, we formulate a \emph{Boolean relaxation} of the problem, where
we reduce the instance to a Boolean bijunctive formula, in order to use the powerful
FPT-algorithm for \MinSat\ by Kim et al.~\cite{Kim:etal:sicomp2025} as a backend solver. 
This formulation is not precise, but it is an overapproximation of the
solution space, whose solution gives partial information about a variable value $\varphi(v)$.
Technically, it assigns a \emph{coset} of $\ZZ_{p^d}$ for every value $\varphi(v)$;
e.g.\ in the simplest interesting case of $\ZZ_4$ such a solution may encode
$\varphi(v)=1 \pmod 2$ without settling on either $\varphi(v)=1$ or $\varphi(v)=3$.
This will handle constraints $u=pv$, but fail to resolve permutation constraints
such as $u=3v$. To address this, we apply the balanced subgraph covering tool
to a collection of auxiliary biased graphs that capture the structure
of these permutation constraints. The solutions returned by this tool
can then be used to add additional constraints to the Boolean formula,
in a way that allows a \emph{disambiguation step} -- a guarantee
that every assignment to the Boolean formula can be fully resolved
into an assignment $\varphi \colon V \to \ZZ_{p^d}$ without increasing
the number of violated constraints. Thus the \MinSat\ solver,
in the presence of these additional constraints, acts as a complete FPT
algorithm for $\minlin{2}{\ZZ_{p^d}}$. 

\smallskip

For the negative result, we first note that $\ZZ_m$ is a direct sum of $\omega(m)$ different rings $\ZZ_{p^d}$ by the Chinese Remainder Theorem.
We show that $\minlin{2}{R}$ for a ring $R$
that is the direct sum of $p$ rings $R_i$ does not have an
FPT-approximation with ratio $p-\varepsilon$ for any $\varepsilon > 0$
unless ETH is false. 
In the course of the proof,
we define an auxiliary problem
\pSplitng{} and show that it
cannot be FPT-approximated within $p -\eps$
for any $\eps > 0$.
\pSplitng{} is a natural generalization of 
the problem \textsc{Paired Min Cut} 
which has frequently been used in W[1]-hardness proofs (e.g.~\cite{
Anand:etal:ipec2025,
Kim:etal:sicomp2025,
Osipov:Wahlstrom:esa2023,
Dabrowski:etal:ipec2023,
marx2009constant}).
To show hardness of \pSplitng{}, 
we start from a recent result of 
Bafna~et~al.~\cite{Bafna:etal:stoc2025}
that \textsc{Max-$2$-CSP} parameterized
by number of variables is hard to approximate
within any constant factor under the ETH,
even on complete instances,
i.e.\ where every pair of variables
is bound by a constraint.
To reduce to \pSplitng{},
we construct dense instances of 
\textsc{Max-$2$-CSP} where the constraints
can be partitioned into $p$-cliques
using Tur{\'a}n's theorem,
transfer inapproximability
from Bafna~et~al.\ to this case,
and then reduce the case to \pSplitng{}.

\subsection{Balanced Subgraph Covering}
\label{sec:intro-coverings}

Our algorithm builds on a combinatorial statement about so-called \emph{biased graphs}.
Biased graphs were studied by Zaslavsky~\cite{Zaslavsky:jctb89,Zaslavsky:jctb91}
and are an important topic in matroid theory; see e.g.~\cite{oxley2011matroid,walsh2026gaingraphs,zaslavsky2012mathematical}.
They have also seen applications in parameterized complexity~\cite{minlin25talg,wahlstrom2017lp}.
Before we give the definitions, let us review an example. Let $G$ be a
graph, and consider a simple cycle $C$ in $G$ to be \emph{balanced} if
it has even length. A subgraph $H$ of $G$ is balanced if every cycle
in $H$ is balanced -- i.e.\ if $H$ is bipartite. 
The problem of deleting at most $k$ edges from $G$ so that the remaining graph is balanced is now effectively \textsc{Edge Bipartization}, which has a long history in parameterized complexity (along with its vertex deletion variant, \textsc{Odd Cycle Transversal})~\cite{pilipczuk2019edge,guo2006compression,reed2004finding,KolayMRS20oct}.

Biased graphs provide a rich generalization of these notions: a {\em biased graph} is a pair $(G, \BB)$ where $G$ is a graph and $\BB$ is the set of cycles in $G$ that are considered balanced (provided via oracle access). $\BB$ must satisfy a closure property:  for any cycle $C$ in $G$ that is not balanced, and any chord path $P$ that connects two vertices of $C$, at least one of the two resulting new cycles in $C+P$ is also unbalanced. 
Both the edge- and vertex-deletion versions of the above ``cleaning task'' (i.e.\ remove at most $k$ edges/vertices from $G$ such that the resulting graph is balanced) are FPT, with a running time of $O^*(4^k)$,  even with oracle access to $\BB$~\cite{wahlstrom2017lp,minlin25talg}.
As a prominent application, let $\Gamma$ be a group, and let $G$ be a graph. Refer to an oriented edge in $G$ as a triple $(e,u,v)$ where $e \in E(G)$ and $u$ and $v$ are the endpoints of $e$, representing orienting $e$ from $u$ to $v$. 
Let $\gamma$ be a function labelling the oriented edges of $G$ by $\Gamma$, where for every edge $e=uv \in E(G)$ we have $\gamma(e,u,v)=\gamma(e,v,u)^{-1}$. Let a cycle $C$ be balanced if its edge labels, taken in the order and orientation of $C$, multiply to identity in $\Gamma$. Then this defines a biased graph $(G,\BB)$,
known as a group-labelled graph or \emph{$\Gamma$-gain graph},
and the corresponding cleaning task serves as a ``meta algorithm'' that captures many problems
such as \textsc{Edge Bipartization/OCT}, \textsc{Multiway Cut}, \textsc{Subset Feedback Edge/Vertex Set}
and more~\cite{wahlstrom2017lp}. The setting also provides linear-time FPT algorithms~\cite{iwata201801all} 
and approximation results~\cite{wahlstrom17preprint}.

However, we are looking at a more general task. Let $(G,\BB)$ be a biased graph and let $H$ be a subgraph of $G$. For a set of vertices $S \subseteq V(G)$, let $\delta(S)$ be the set of edges in $G$ with one endpoint in $S$, and define the \emph{cost of $H$} to be 
\[
\cost(H)=|\delta(V(H))|+|E(G[V(H)] \setminus E(H)|
\]
i.e.\ the number of edges that need to be deleted to create and ``cut off'' $H$ from the rest of $G$. We define a \emph{domination order} on balanced subgraphs of $G$: given two balanced subgraphs $H_1$ and $H_2$, we say that $H_2$ \emph{dominates} $H_1$ if $V(H_1) \subseteq V(H_2)$ and $\cost(H_1) \geq \cost(H_2)$, and \emph{strictly dominates} if one of the inequalities is strict. We say that $H$ is an \emph{important balanced subgraph} of $G$ if no balanced subgraph strictly dominates $H$. Note that distinct subgraphs can be \emph{domination equivalent}, i.e.\ mutually dominate each other. Thus, we lift the definition: for a subset $S \subseteq V(G)$, we define $\cost(S)$ as the minimum cost of a balanced subgraph $H$ of $G$ with $V(H)=S$, and say that $S$ is an \emph{important subset} if there is an important balanced subgraph $H$ with $V(H)=S$. It is known that for every vertex $v \in V(G)$ there are at most $4^k$ important \emph{connected} subsets (i.e.\ $G[S]$ is connected) containing $v$ of cost at most $k$, and they can be enumerated in $O^*(4^k)$ time~\cite{minlin25talg}. 
However, if $S$ (respectively $H$) is not required to be connected, then the situation is much more intangible. 

Now the setting is as follows. Given a graph $G$, with an unknown balanced subgraph $H$ of cost at most $k$, not necessarily connected, we wish to \emph{cover} $H$ by another balanced subgraph $H'$, such that $V(H) \subseteq V(H')$ and $H'$ (in some sense) captures the cost of $H$. 
We show the following. 

\begin{restatable}{theorem}{thmBiasSample} \label{ithm:balanced-shadow-sampling}
      There is a randomized algorithm that 
      takes as input a biased graph $(G,\BB)$ 
      and an integer $k$, with oracle access to $\BB$,
      and in $\bigoh^*(4^k)$ time outputs 
      a vertex set $S \subseteq V(G)$ and edge set
      $F \subseteq E(G)$ such that the following hold.
      \begin{enumerate}
        \item $\delta_G(S) \subseteq F$.
        \item $(G-F)[S]$ is balanced.
        \item For every (not necessarily connected) balanced
          subgraph $H$ of $G$ of cost at most $k$,
          with probability at least $1/2^{\bigoh(k^2)}$ we have (1) $V(H) \subseteq S$ and (2) there are at most $k$ edges in $F$ with an endpoint in $V(H)$.
      \end{enumerate}
\end{restatable}

For a set $U \subseteq V$, refer to a \emph{cleaning set} for $U$ as $F_U \subseteq E(G)$ such that $G[U]-F_U$ is balanced. Returning to bipartite graphs, the theorem then informally says: for any bipartite subgraph $H$ of $G$, with $\cost(H) \leq k$, 
with probability $1/2^{O(k^2)}$ the set $F$ that is returned contains an optimal cleaning set for $V(H)$. 
The result can be derandomized by using \emph{cover-free families}~\cite{BshoutyG17cff}.
In the case when $H$ is assumed to be an important balanced subgraph, we can summarize Theorem~\ref{ithm:balanced-shadow-sampling} in a simpler way.

\begin{restatable}{corollary}{corSampleImportant}\label{icor:sample-important}
  There is a randomized algorithm that 
  takes as input a biased graph $(G,\BB)$ and an integer $k$,
  and in $\bigoh^*(4^k)$ time outputs 
  a balanced subgraph $G_B$ of $G$ 
  such that for any important balanced subgraph $H$ of cost at most $k$,
  with probability at least $1/2^{\bigoh(k^2)}$
  every connected component of $H$ is domination equivalent to a
  connected component of $G_B$. In particular, in this case
  we have $V(H) \subseteq S$ and $N_G(V(H)) \cap S =
  \emptyset$.
\end{restatable}

This result is a generalization of \emph{shadow covering} (a.k.a.\ \emph{shadow removal}) for edge cuts in undirected graphs, unifying the concept with the \emph{important connected balanced subgraphs} of Dabrowski et al.~\cite{minlin25talg}. The shadow covering procedure was originally created by Marx and Razgon, under the name \emph{random sampling of important separators}, for an FPT algorithm for \textsc{Multicut}~\cite{marx2014fixed}, and has become an important part of the toolbox for FPT algorithms for graph separation problems (see e.g.~\cite{ChitnisHM13dmwc,chitnis2015directed,HatzelJLMPSS23}).
The shadow covering result was subsequently improved and generalized by Chitnis et al.~\cite{chitnis2015directed}.
Our procedure matches the success probability of~\cite{chitnis2015directed}; the success probability for shadow covering in undirected graphs was very recently improved to $1/2^{\bigoh(k \log k)}$~\cite{chu2026faster}. 

However, for our application the formulation in Theorem~\ref{ithm:balanced-shadow-sampling} is more relevant than Corollary~\ref{icor:sample-important}. We will use the theorem for the setting of group-labelled graphs, where the group is the multiplicative group of $\ZZ_{p^d}$ for various $p^d$; see below.

\subsection{Algorithm for $\minlinz{2}{m}$}
\label{sec:intro-algorithm}

Recall from the technical overview that it is sufficient to
prove the following result in order to prove the positive
part of Theorem~\ref{thm:main}.

\begin{restatable}{theorem}{thmZpd}\label{thm:zpd_is_fpt}
  $\minlin{2}{\ZZ_{p^d}}$ is in \FPT for every prime $p$ and $d \geq 1$.    
\end{restatable}

When $d = 1$, the ring $\ZZ_p$ is a field,
and FPT algorithms for \textsc{Min}-$2$-\textsc{Lin} over a field are known
(see e.g.~\cite{minlin25talg}).
The first interesting case for us is when $d=2$ and $p=2$,
i.e.\
$\minlin{2}{\ZZ_4}$, which has been an open problem 
resisting attacks with the existing FPT techniques.
We concentrate on this example to illustrate some ideas
behind our algorithm.
As pointed out in the technical overview, we can
focus on special instances. Hence,
we only consider 
instances of $\minlin{2}{\ZZ_4}$ where all binary equations are
of the form $u = 2v \bmod 4$ and $u = 3v \bmod 4$,
and all unary equations are crisp and of the form $s = 1$.
It is worth noting that
the equation $u = 3v \bmod 4$ is symmetric, i.e.\ 
$3x = y \bmod 4$ is equivalent to $x + y = 0 \bmod 4$,
which is in turn equivalent to $x = 3y \bmod 4$.

\paragraph{Non-bipartite obstructions.} 

Before diving into the algorithm,
let us give some examples that hint at the connection
between $\minlin{2}{\ZZ_4}$ and
balanced (in fact, bipartite) subgraph covering.
We will make the link more concrete in the next subsection.
It is not hard to see that $\minlin{2}{\ZZ_4}$
encodes \textsc{Edge Bipartization},
i.e.\ the problem asking whether $k$ edges can
be removed from an undirected graph $G$ to make it bipartite.
Given a graph $G$, define a system of equations $I_G$ over $\ZZ_4$ 
with the vertices of $G$ as the variables and 
equations $3u=v \bmod 4$ for every edge $uv$ in $G$.
Force all variables of $I_G$ to take values in $\{1,3\}$,
e.g.\ by adding equations $2v=2 \bmod 4$ for all $v$ in $k+1$ copies
(where $2v=2$ can be implemented as $2v=w$, $w=2s$, $s=1$ in the restricted instances we consider). %
The only allowed assignments to the endpoints of 
an edge are $(1,3)$ and $(3,1)$.
Thus, any assignment $\alpha : V(G) \to \{1,3\}$ of cost $k$ to $I_G$
can be viewed as a $2$-colouring of $G$ (with the two colours 
being $1$ and $3$) that is proper on all but $k$ edges.
However, $\minlin{2}{\ZZ_4}$ encodes more complicated problems.
Start again with a graph $G$ and an instance $I_G$, as described above.
Pick an arbitrary subset $U$ of the vertices in $G$ and
add equations $2u = 2 \bmod 4$ for all $u \in U$ to $I_G$.
Now the algorithm may ``opt out'' of making some parts of $G$
bipartite and use values $0$ and $2$ on those variables.
However, this requires cutting that part of the graph
away from the bipartite region, and also deleting
all unary constraints $2u = 2 \bmod 4$ for $u$ that fall 
outside the bipartite region.

\begin{figure}
  \centering
  \begin{tikzpicture}[
    thick,
    node/.style={circle, draw, minimum size=6pt, inner sep=0pt},
    black node/.style={node, fill=black},
    white node/.style={node, fill=white},
    red node/.style={node, fill=white}
]

    \coordinate (S1) at (-0.5, 0.5);  %
    \coordinate (S2) at (-0.5, -0.5); %
    \coordinate (S3) at (0.5, -0.5);  %
    \coordinate (S4) at (0.5, 0.5);   %

    \path (S2) ++(198:1) coordinate (P3);
    \path (P3) ++(126:1) coordinate (P4);
    \path (P4) ++(54:1)  coordinate (P5);

    \path (S3) ++(-30:1) coordinate (H3);
    \path (H3) ++(30:1)  coordinate (H4);
    \path (H4) ++(90:1)  coordinate (H5);
    \path (H5) ++(150:1) coordinate (H6);

    \coordinate (C1) at (S4);
    \coordinate (C2) at (S3);
    \coordinate (C3) at (H3);
    \coordinate (C4) at (H4);
    \coordinate (C5) at (H5);
    \coordinate (C6) at (H6);

    \foreach \i in {1,...,5} {
        \pgfmathtruncatemacro{\start}{\i+1}
        \foreach \j in {\start,...,6} {
            \draw (C\i) -- (C\j);
        }
    }

    \draw (S2) -- (S1);
    \draw (S4) -- (S3);

    \draw (S1) -- (S4);
    \draw (S3) -- (S2);

    \draw (S1) -- (P5) -- (P4) -- (P3) -- (S2);
    \draw (S1) -- (P3);

    \node[white node] at (S1) {};
    \node[black node] at (S2) {};
    \node[black node] at (P3) {};
    \node[black node] at (P4) {};
    \node[black node] at (P5) {};

    \node[white node] at (S3) {};
    \node[white node] at (S4) {};
    \node[white node] at (H3) {};
    \node[white node] at (H4) {};
    \node[black node] at (H5) {};
    \node[white node] at (H6) {};

\end{tikzpicture}
  \hspace{1cm}
  \begin{tikzpicture}[
    thick,
    node/.style={circle, draw, minimum size=6pt, inner sep=0pt},
    black node/.style={node, fill=black},
    white node/.style={node, fill=white},
    red node/.style={node, draw=red, fill=red}
]

    \coordinate (S1) at (-0.5, 0.5);  %
    \coordinate (S2) at (-0.5, -0.5); %
    \coordinate (S3) at (0.5, -0.5);  %
    \coordinate (S4) at (0.5, 0.5);   %

    \path (S2) ++(198:1) coordinate (P3);
    \path (P3) ++(126:1) coordinate (P4);
    \path (P4) ++(54:1)  coordinate (P5);

    \path (S3) ++(-30:1) coordinate (H3);
    \path (H3) ++(30:1)  coordinate (H4);
    \path (H4) ++(90:1)  coordinate (H5);
    \path (H5) ++(150:1) coordinate (H6);

    \coordinate (C1) at (S4);
    \coordinate (C2) at (S3);
    \coordinate (C3) at (H3);
    \coordinate (C4) at (H4);
    \coordinate (C5) at (H5);
    \coordinate (C6) at (H6);

    \foreach \i in {1,...,5} {
        \pgfmathtruncatemacro{\start}{\i+1}
        \foreach \j in {\start,...,6} {
            \draw (C\i) -- (C\j);
        }
    }

    \draw (S2) -- (S1);
    \draw (S4) -- (S3);

    \draw[red, very thick, dashed] (S1) -- (S4);
    \draw[red, very thick, dashed] (S3) -- (S2);

    \draw (S1) -- (P5) -- (P4) -- (P3) -- (S1); 
    \draw[red, very thick, dashed] (S2) -- (P3);

    \path let \p1 = ($ (S1) + (S2) + (P3) + (P4) + (P5) $), \n1={5} in
          (\x1/\n1, \y1/\n1) coordinate (CPentagon);

    \draw[orange, fill=orange, opacity=0.3] (CPentagon) circle (1.1);

    \node[white node] at (S1) {};
    \node[black node] at (S2) {};
    \node[black node] at (P3) {};
    \node[black node] at (P4) {};
    \node[black node] at (P5) {};

    \node[white node] at (S3) {};
    \node[white node] at (S4) {};
    \node[white node] at (H3) {};
    \node[white node] at (H4) {};
    \node[red   node] at (H5) {};
    \node[white node] at (H6) {};

\end{tikzpicture}
  \caption{An illustration of an instance of $\minlin{2}{\ZZ_4}$ on the left
  and a solution on the right.
  The nodes represent the variables of the instance.
  Every edge $xy$ is an equation of the form $3x = y \bmod 4$
  on its endpoints.
  A node $v$ coloured black is equipped with an equation $2v = 2 \bmod 4$.
  On the right, the red dashed edges and the red node
  show the equations deleted by the solution;
  the orange region encloses the set of variables
  assigned values $1$ and $3$.
  }
  \label{fig:bipartite-cleaning}
\end{figure}
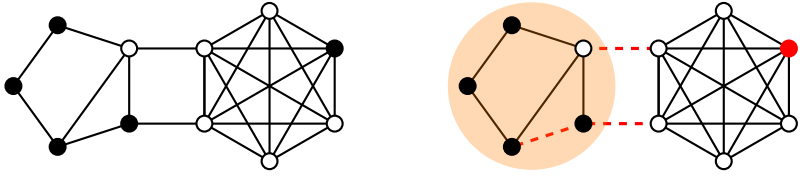

For a concrete example, see Figure~\ref{fig:bipartite-cleaning}.
It depicts an instance $I_G$ for a graph $G$
consisting of a $5$-cycle with a chord joined to 
a $6$-clique by the two horizontal edges.
The set $U$ are the black nodes.
One possible solution to the instance in the illustration
is to delete all unary equations
$2u = 2$ for all $u \in U$, which costs five.
There is a better solution of cost four,
where the vertices of the $5$-cycle receive
values in $\{1,3\}$.
Informally, an optimal solution has to delete
the constraint $2v = 2$ for the vertex $v$ in the $6$-clique because 
carving out a bipartite subgraph from the $6$-clique is too costly.
On the other hand, the $5$-cycle with a chord is almost
bipartite (one edge-deletion away),
and cutting it from the rest of the graph
only requires removing the two horizontal edges.
Intuitively, balanced subgraph covering helps us 
by discovering this information about the graph.
However, there are obstructions to satisfiability of 
$\minlin{2}{\ZZ_4}$ instances beyond non-bipartite
subgraphs, and the technical challenge is in 
integrating balanced subgraph covering
with the rest of the algorithm.

\paragraph{MinCSP relaxation.}

To solve $\minlin{2}{\ZZ_4}$, we
leave the realm of linear equations and consider
a more general problem called \MinCSP{}.
Let $A$ be a set of values called the \emph{domain},
and $\bA$ be a set of relations over $A$.
An instance of $\MinCSP(\bA)$
is a triple $(V, \cC, k)$,
where $V$ is a set of variables,
$\cC$ is a (multi)set of constraints,
and $k$ is an integer.
Each constraint is of the form $R(x_1,\dots,x_r)$,
where $R \in \bA$ is a relation of arity $r$, 
and $x_1,\dots,x_r$ are (not necessarily distinct) 
variables from $V$.
An assignment $\alpha : V \to A$ satisfies 
a constraint $R(x_1,\dots,x_r)$
if $(\alpha(x_1),\dots,\alpha(x_r)) \in R$,
and violates it otherwise.
The question is whether $(V, \cC)$ admits
an assignment that violates at most $k$ constraints.

To cast $\minlin{2}{\ZZ_4}$ as a \MinCSP{},
let $A_4 = \{0,1,2,3\}$ and $\bA_4 = \{R_3, R_2, U_1\}$
where
$R_3 = \{(0,0), (1,3), (2,2), (3,1)\}$,
$R_2 = \{(0,0), (1,2), (2,0), (3,2)\}$, and
$U_1 = \{1\}$.
An equation $3u = v$ is then equivalent
to a constraint $R_3(u,v)$,
an equation $2u = v$ to
$R_2(u,v)$, and
a unary equation $v = 1$
to $U_1(v)$.
The special case of \MinCSP{} with Boolean domain,
i.e.\ $\{0,1\}$, is called \MinSat{}. 
A common approach in FPT algorithms (for 
\MinCSP{}s and other problems) 
is to encode them as \MinSat{}
using Boolean relations
that can be handled by the general algorithm
of Kim~et~al.~\cite{Kim:etal:sicomp2025}.
However, this approach does not seem directly feasible
for solving $\MinCSP(\bA_4)$.
Instead, the same pipeline can be used to handle
a relaxation of the problem in FPT time.

Let $A^*_4 = \{0,1,2,3,\unit\}$, 
i.e.\ $A_4$ extended with a new \emph{ambiguous} value $\unit$, 
which is interpreted as ``1 or 3''.
The relaxation of $R_3$ is then $R^*_3 = R_3 \cup \{(\unit,\unit)\}$, and
the relaxation of $R_2$ is $R^*_2 = R_2 \cup \{(\unit,2)\}$.
We define $\bA_4^* = \{R^*_2, R^*_3, U_1\}$.
The reduction from $\MinCSP(\bA^*_4)$ to \MinSat{} is deferred to 
Section~\ref{ssec:boolean-transformation}:
here we mention without proof that $\MinCSP(\bA^*_4)$ is in \FPT.
Note that $\MinCSP(\bA_4^*)$ is a proper relaxation of $\MinCSP(\bA_4)$:
on the one hand, every yes-instance of $\MinCSP(\bA_4)$
is automatically a yes-instance of $\MinCSP(\bA^*_4)$.
However, $\MinCSP(\bA^*_4)$ accepts strictly more instances:
for example, the set of equations
\begin{equation}
  \label{eq:triangle_on_a_stick}
  s = 1, \ 2s = v, \ 2u_1 = v, \ 3u_1 = u_2, \ 3u_2 = u_3, \ 3u_3 = v_1
\end{equation}
has no zero-cost assignment in $\ZZ_4$ because $u_1, u_2, u_3$ 
are forced to take values in $\{1,3\}$,
and the last three equations form an odd cycle.
On the other hand, the relaxation admits a zero-cost assignment
that maps $s$ to $1$, $v$ to $2$ and 
$u_1,u_2,u_3$ to $\unit$.

Intuitively, the only obstructions that $\MinCSP(\bA_4^*)$ 
``misses'' are the odd cycles of $R_3$-constraints, 
such as the triangle in the example above.
In this case, the relaxed solution is allowed to be ambiguous,
while a proper solution is not.
Our idea is to add more constraints to the relaxed instance
to handle these obstructions, which would allow
\emph{disambiguating} any relaxed solution without increasing cost.
Here, important balanced subgraphs come into play.
Construct the \emph{unit graph} $G$ 
with the variables of $I$ as vertices and 
an edge $uv$ for every equation $3u = v$ in $I$.
Fix a hypothetical optimal assignment $\alpha : V(I) \to \ZZ_4$ 
and let
$V_{1,3}$ be the set of vertices taking values $1$ and $3$
in $\alpha$.
Then, the subgraph $H$ of $G[V_{1,3}]$ obtained 
after removing the edges corresponding to 
the constraints violated by $\alpha$ must be bipartite.
Since at most $k$ edges are removed,
$H$ is a balanced subgraph of cost at most $k$ in
the biased graph $(G, \BB)$, where
$\BB$ is the family of even cycles.

Applying the algorithm from
Theorem~\ref{ithm:balanced-shadow-sampling}
to $(G, \BB)$ with parameter $k$,
we expect to cover $H$ with sufficient probability.
In particular,
we obtain a pair $(S, F)$ such that
$(G - F)[S]$ is bipartite, and
with sufficient probability,
$V_{1,3} \subseteq S$ and $|V(F) \cap V_{1,3}| \leq k$,
i.e.\
all variables taking values $1$ and $3$  are 
covered by $S$ and at most $k$ variables taking
values $1$ and $3$ are endpoints
of the edges in $F$.
Now, we can modify the relaxation to take care of
all non-bipartite obstructions.
First, we remark without proof that the unary relations
$U_{0,2} = \{0,2\}$, $U_{0,2,1} = \{0,2,1\}$ and
$U_{0,2,3} = \{0,2,3\}$ can be added to $\bA^*_4$ 
while maintaining that $\MinCSP(\bA^*_4)$ is in \FPT.
We add the following annotations to the instance:
\begin{itemize}
  \item For every variable $v$ outside $S$,
  add constraint $U_{0,2}(v)$ in $k + 1$ copies,
  forbidding $v$ from taking values $1$, $3$ and $\unit$.
  \item Then, for every variable $v \in V(F)$,
  guess a value $g(v) \in \{1,3\}$ uniformly at random
  and add constraint $U_{0,2,g(v)}(v)$ in $k+1$ copies.
\end{itemize}
With probability $1/2^k$, 
$\alpha$ is an optimal assignment to the relaxation
with the additional constraints:
indeed, no variable outside $S$ takes value in $\{1,3\}$,
and only $\leq k$ variables in $V(F)$ take a value in $\{1,3\}$, so 
only $\leq k$ of our guesses need to be correct.
Thus, if we start with a yes-instance,
with sufficient probability it remains a yes-instance after
the annotations are added.

What we gain from the annotations is that every
assignment $\alpha^* : V(I) \to \{0,1,2,3,\unit\}$
can be disambiguated
into a proper assignment $\alpha : V(I) \to \{0,1,2,3\}$
without increasing the cost.
To be specific,
recall that $(G - F)[S]$ is bipartite
and fix a consistent labelling $\lambda : S \to \{1,3\}$.
Let $\alpha(v) = \alpha^*(v)$ whenever
$\alpha^*(v) \in \{0,1,2,3\}$, and let
$\alpha(v) = \lambda(v)$ if $\alpha^*(v) = \unit$.
The correctness follows by observing
that the variables assigned $\unit$ by $\alpha^*$
must belong to a bipartite subgraph in $(G-F)[S]$.
Indeed, every such variable must be in $S$,
and it cannot be connected to 
an endpoint of an edge in $F$
because the annotations and $R^*_3$-constraints 
would force it to take value $1$ or $3$ but not $\unit$.
Thus, it is safe to use $\lambda$ on all 
variables assigned $\unit$.

\paragraph{Generalizing to $\ZZ_{p^d}$.}

We switch our attention to
the general problem $\minlin{2}{\ZZ_{p^d}}$.
In essence, the algorithm for $\minlin{2}{\ZZ_4}$ 
already contains the main ingredients --
the relaxation (solved by a reduction to \MinSat{}) and
the annotations using balanced subgraph covering.
However, adapting them, especially for the case with $d > 2$
requires further insights.
To generalize the relaxation,
we add an ambiguous value for every subset of $\ZZ_{p^d}$
that excludes zero and
can be defined using a binary linear equation.
These turn out to be \emph{proper cosets},
i.e.\ subsets of values $\{i + r \cdot p^j : r \in \ZZ_{p^d}\}$
defined for every $i \neq 0$ and $0 \leq j \leq d-1$.
For example, for $\ZZ_{9}$ we would
introduce two ambiguous values in the relaxation:
$\overline{147}$ and $\overline{258}$.
For $\ZZ_{8}$, the ambiguous values are
$\overline{1357}$, $\overline{15}$, $\overline{26}$ and $\overline{37}$.

In case $d=2$, the use of balanced subgraphs
generalizes in a straightforward way.
For $\ZZ_4$, we can view bipartite subgraph
covering as a group-labelled graph problem
with the group being $\{1,3\}$ equipped with
multiplication modulo $4$.
By analogy, for $\ZZ_9$ we would use the group
$\{1,2,4,5,7,8\}$ equipped with
multiplication modulo $9$:
these elements are the \emph{units} of $\ZZ_9$,
i.e.\ the elements that admit multiplicative
inverses modulo $9$.
The odd cycle obstructions in $\ZZ_4$ generalize
to non-identity cycles in $\ZZ_9$, i.e.\ cycles 
where the edge labels are units and
their product is not $1$ modulo $9$.

\begin{figure}
    \centering
    \begin{tikzpicture}[
    scale=0.7,
    vertex/.style={circle, draw, minimum size=7mm, inner sep=1pt},
    >={Latex[length=3mm, width=2mm]} 
]
    \node[vertex] (s)  at (0, 0)  {$s$};
    \node[vertex] (w)  at (2.5, 0){$w$};
    
    \node[vertex] (u1) at (5, 0)  {$u_1$};
    \node[vertex] (u2) at (9, 0) {$u_2$}; 

    \node[vertex] (v1) at (6, 1.5)  {$v_1$};
    \node[vertex] (v2) at (8, 1.5) {$v_2$};

    \draw[->] (s)  -- node[above] {$2$} (w);
    \draw[->] (u1) -- node[above] {$2$} (w);
    \draw[-]  (u1) -- node[above] {$\mathbf{5}$} (u2);
    
    \draw[->] (u2) to[bend left=20] node[below] {$2$} (w);

    \draw[->] (u1) -- node[above left]  {$2$} (v1);
    \draw[->] (u2) -- node[above right] {$2$} (v2);
    \draw[-]  (v1) -- node[above]       {$7$} (v2);

    \node[vertex] (ss)  at (11 + 0, 0)  {$s$};
    \node[vertex] (ww)  at (11 + 2.5, 0){$w$};
    
    \node[vertex] (uu1) at (11 + 5, 0)  {$u_1$};
    \node[vertex] (uu2) at (11 + 9, 0) {$u_2$}; 

    \node[vertex] (vv1) at (11 + 6, 1.5)  {$v_1$};
    \node[vertex] (vv2) at (11 + 8, 1.5) {$v_2$};

    \draw[->] (ss)  -- node[above] {$2$} (ww);
    \draw[->] (uu1) -- node[above] {$2$} (ww);
    \draw[-]  (uu1) -- node[above] {$\mathbf{3}$} (uu2);
    
    \draw[->] (uu2) to[bend left=20] node[below] {$2$} (ww);

    \draw[->] (uu1) -- node[above left]  {$2$} (vv1);
    \draw[->] (uu2) -- node[above right] {$2$} (vv2);
    \draw[-]  (vv1) -- node[above]       {$7$} (vv2);

\end{tikzpicture}
    \caption{An illustration of two instance of $\minlin{2}{\ZZ_8}$
    with common equations 
    $s = 1$, $2s = w$, $2u_1 = 2$, $2u_2 = 2$, 
    $v_1 = 2u_1$, $v_2 = 2u_1$ and $v_1 = 7v_2$.
    The one on the left additionally has an equation
    $u_1 = 5u_2$ and the one on the right has
    $u_1 = 3u_2$.
    Note that the equations $u = av$ and $au = v$
    are equivalent for $a \in \{1,3,5,7\}$, so we
    depict them as undirected edges.
    The first instance is unsatisfiable,
    while the second one is satisfiable.
    }
    \label{fig:z8_instance}
\end{figure}
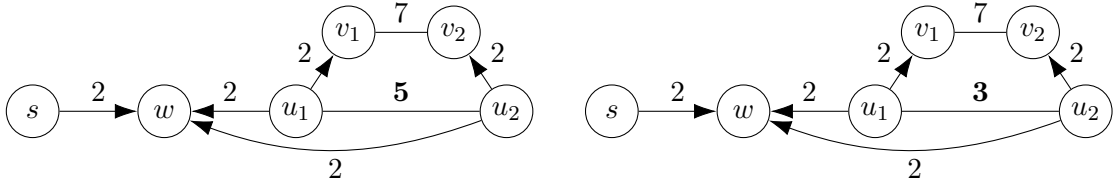

However, for $d > 2$, the situation is more complicated. Consider $\ZZ_8$ and the instances
in Figure~\ref{fig:z8_instance} with the equations
listed in the caption.
To see that the instance on the left has no zero-cost assignments,
note that $2u_1=2$ implies $u_1 \in \{1,5\}$, and by symmetry we can assume
$u_1=1$. Then the two paths from $u_1$ to $v_2$ imply
$v_2=14 \bmod 8=6$ and $v_2=10 \bmod 8=2$, respectively,
and cannot be reconciled in $\ZZ_8$. The cycle on $\{u_1,u_2,v_1,v_2\}$
with $u_1$ being a unit 
forms an obstruction not captured by the unit graph alone
(or indeed by the relaxation).

We handle this via an additional biased graph, which
handles cycles where some variables take values from $\{2,6\}$. Consider the cycle on $\{u_1,u_2,v_1,v_2\}$ in the two instances
in Figure~\ref{fig:z8_instance}. In the first instance, it is unsatisfiable
when $u_1$ is a unit, but is satisfied by $u_1=u_2=2$, $v_1=v_2=4$. 
The second instance is satisfiable (e.g.\ $u_1=1$, $u_2=3$, $v_1=2$, $v_2=6$),
but would be unsatisfiable over e.g.\ $\ZZ_{16}$. The fact that the
cycle involves a factor of $2$ implies that it only needs to be
consistent mod $2^{d-1}$. As a general rule for $\ZZ_{p^d}$,
for a cycle $C$ in the primal graph of the instance (i.e.\ involving both unit edges
and constraints $u=pv$), if a variable of the cycle takes a value
divisible by $p^i$, $i<d$, then the cycle only needs to be balanced
mod $p^{d-i}$; effectively, if we interpret the values as being written in base $p$,
then we ``lose'' $i$ digits of precision due to values being taken mod $p^d$. 
Thus, for every $i<d$ we introduce an auxiliary biased graph
$(G_i,\BB_i)$ which handles cycles on variables taking values divided
by at most $i$ powers of $p$, and with edge labels interpreted mod $p^{d-i}$. 

For the construction of $G_i$, we let every variable $v$ be
represented by $i+1$ vertices $v^0$, \ldots, $v^i$, where the superscripts
trace the number of powers of $p$ that divide a variable value. 
An equation $u=a \cdot v$ becomes edges $u^jv^j$, $0 \leq j \leq i$,
with label $a \bmod p^{d-i}$, and an equation
$u=p \cdot v$ becomes edges $u^{j+1}v^j$, $0 \leq j < i$,  with label $1$.
Assignments in this graph correspond to cosets, e.g.\ 
$v^j=a$ in $G_i$ for a unit $a \in \ZZ_{p^{d-i}}$
represents $v=a \cdot p^j \pmod {p^{d-i}}$.
One can then show that an unbalanced cycle in $G_i$ involving a vertex
$v^j$ cannot be satisfied by any assignment where $v=a \cdot p^j$
for a unit $a$.

Returning to $\ZZ_8$, we use a covering and annotation strategy in $(G_1,\BB_1)$
similarly as in the unit graph, aiming to cover
the subgraph induced by vertices $v^0$ for the variables $v$ assigned $1,3,5,7$,
and vertices $v^1$ for the variables $v$ assigned $2, 6$
in an optimal solution.
The annotations from the unit graph and $(G_1,\BB_1)$ are then combined
to argue that we can turn an assignment
to the relaxation, which is allowed to use ambiguous values
$\overline{1357}$, $\overline{15}$, $\overline{26}$ and $\overline{37}$, 
into a proper assignment of the same cost.
The labelling of the new graph is used to assign
$2$ or $6$ to the variables with ambiguous value $\overline{26}$,
and values
$\overline{15}$ and $\overline{37}$ to the variables assigned $\overline{1357}$.
The labelling of the unit graph is then used to
choose a value in $\{1,5\}$ and a value in $\{3,7\}$
for the variables assigned $\overline{15}$ and $\overline{37}$, respectively.

To summarize, the high-level recipe for 
$\ZZ_{p^d}$ is to introduce ambiguous values
to the relaxation for every proper coset.
Then, construct $d-1$ graphs with appropriately chosen groups
that capture all non-identity cycle
obstructions in $\minlin{2}{\ZZ_{p^d}}$.
Use the outputs from the covering algorithm to annotate the relaxation,
and then use the labellings of the graphs
to disambiguate the relaxed assignment
without increasing its cost.

\subsection{Hardness of FPT Approximation}
\label{sec:intro-hardness}

Let $m \geq 2$ be an integer.
Dabrowski et al.~\cite{Dabrowski:etal:arxiv2024} showed
that if $m$ is not a prime power, then $\minlin{2}{\ZZ_m}$ is not FPT-approximable within $2-\epsilon$ for any $\epsilon > 0$, unless the \ETH is false. In fact, this result holds for all commutative rings
that are not
{\em lineal} rings~\cite{Marks:Mazurek:ijm2016} and
for rings that are the direct sum of two or more non-trivial rings.
The proof is a reduction from \textsc{Paired Min Cut}; this problem
was shown to be W[1]-hard in~\cite[Theorem~7]{marx2009constant},
and variants of it are often used in hardness reductions.
Dabrowski et al. strengthened this result by 
proving that \textsc{Paired Min Cut}
(in a special form equivalent to \textsc{$2$-Split Min Cut} below)
is not \FPT-approximable within $2-\epsilon$ for any 
$\epsilon > 0$ unless the \ETH is false.
We generalize their result by considering the following
problem.

\pbDefP{\pSplitng}
{A graph $G$ with special vertices $s, t$ such that $G-\{s,t\}$ is the disjoint union of 
$p$ graphs $G_1,\dotsc,G_p$, 
a set of pairwise disjoint edge sets (a.k.a. \emph{bundles}) $\cB \subseteq \binom{E(G)}{1}
\cup \binom{E(G)}{p}$, each set containing at most one edge from each
graph $G_i$,
and an integer $k$.}
{$k$.}
{Is there an $st$-cut $Z \subseteq E(G)$ that is the union of $k$ bundles in $\cB$?}

We show that \pSplitng is not FPT-approximable within $p-\epsilon$ under the ETH,
and use this result to improve the $2-\epsilon$ lower bound for
FPT-approximability of $\minlin{2}{\ZZ_m}$ to $\omega(m)-\epsilon$. This
matches the upper bound obtained from our algorithm, and implies the
hardness part of Theorem~\ref{thm:main}. 
Thus, this is the final piece of the puzzle that completes
the picture of FPT-approximability of $\minlin{2}{\ZZ_m}$. 
In the reductions used by Dabrowski et al., 
the step from \textsc{Paired Min Cut} to $\minlin{2}{\ZZ_{m}}$ with $\omega(m) = 2$
generalizes quite naturally to a reduction
from \pSplitng{} to $\minlin{2}{\ZZ_{m}}$ with $\omega(m) = p$.
However, generalizing $(2-\eps)$-inapproximability of \textsc{Paired Min Cut}
to \pSplitng{} with $p > 2$ poses several challenges, which we discuss below.

As in~\cite{Dabrowski:etal:arxiv2024}, we start from 
the \textsc{Max-2-CSP} problem parameterized by the number of variables.
The problem can be equivalently stated as 
\textsc{Multicoloured Subgraph Isomorphism (MSI)}.
In MSI we are given two graphs $G$, $H$
and a partition of $V(G) = \biguplus_{a \in V(H)} V_a(G)$ 
into $|V(H)|$ colour classes, with one class per vertex of $H$.
One should think of $G$ as large and $H$ as small;
the problem is parameterized by $|V(H)|$.
We are looking for a mapping from $H$ to $G$,
and a colour class $V_a(G)$ in the input 
represents the set of possible images of $a$ in $G$.
The goal is to find a mapping $h$
such that $h(a) \in V_a(G)$ for all $a \in V(H)$ 
and $h$ maps as many edges of $H$ to edges of $G$ as possible.

In a recent breakthrough~\cite{Guruswami:etal:jacm2025},
Guruswami~et~al.\ showed that, assuming ETH,
MSI cannot be approximated within any constant factor.
Roughly, for every $0 < \eps < 1$ it is hard
to find a mapping that preserves $\eps \cdot |E(H)|$ edges.
To illustrate the hardness proof, we will describe a reduction from MSI to 
\textsc{Paired Min Cut} (reminiscent of \cite{marx2009constant}), then 
discuss the challenges of extending this reduction to \pSplitng{}
for $p > 2$.
Given a pair of graphs $G$ and $H$ and a partition of 
$V(G)$ into colour classes, we construct 
an instance of \textsc{Paired Min Cut} with a graph $D$
and parameter $k = |E(H)|$.
The edges of $D$ represent the possible mappings of edges 
of $H$ to edges of $G$.
Formally, for every ordered pair of vertices $(a,b)$ in $H$
such that $ab \in E(H)$,
let $E_{a,b}$ be the set of edges $uv$ in $G$ such that $u \in V_a(G)$ and $v \in V_b(G)$.
We create an edge $e_{a,b}^{u,v}$ in $D$ for $uv \in E_{a,b}$.
Note that an edge $uv$ in $G$ corresponds to two edges in $D$: 
$e_{a,b}^{u,v}$ and $e_{b,a}^{v,u}$.
To encourage selecting edges of $G$,
we make a bundle $\{e_{a,b}^{u,v}, e_{b,a}^{v,u}\}$ 
for every edge $uv$ in $G$;
the role of these bundles will be apparent below.

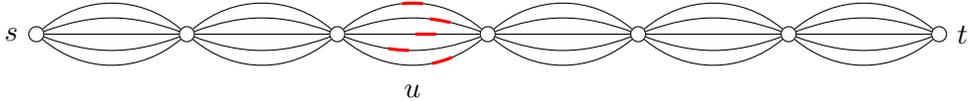
\begin{figure}
    \centering
    \begin{tikzpicture}[
  scale=1,
  vertex/.style={circle, draw, minimum size=2mm, inner sep=1pt}
]
    \node[vertex, label=left:{$s$}] (v0) at (2*0, 0) {};
    \node[vertex] (v1) at (2*1, 0) {};
    \node[vertex] (v2) at (2*2, 0) {};
    \node (v) at (2*2.5, -0.75) {$u$};
    \node[vertex] (v3) at (2*3, 0) {};
    \node[vertex] (v4) at (2*4, 0) {};
    \node[vertex] (v5) at (2*5, 0) {};
    \node[vertex, label=right:{$t$}] (v6) at (2*6, 0) {};

    \draw (v0) to [bend left=40]  (v1);
    \draw (v0) to [bend left=20]  (v1);
    \draw (v0) to                 (v1);
    \draw (v0) to [bend right=20] (v1);
    \draw (v0) to [bend right=40] (v1);

    \draw (v1) to [bend left=40]  (v2);
    \draw (v1) to [bend left=20]  (v2);
    \draw (v1) to                 (v2);
    \draw (v1) to [bend right=20] (v2);
    \draw (v1) to [bend right=40] (v2);

    \draw (v2) to [bend left=40]  (v3);
    \draw (v2) to [bend left=20]  (v3);
    \draw (v2) to                 (v3);
    \draw (v2) to [bend right=20] (v3);
    \draw (v2) to [bend right=40] (v3);

    \path[decorate, decoration={markings, mark=at position 0.5 with {\coordinate (m1); \node[above=-1pt, red] {};}}] (v2) to [bend left=40] (v3);
    \path[decorate, decoration={markings, mark=at position 0.7 with {\coordinate (m2);}}] (v2) to [bend left=20] (v3);
    \path[decorate, decoration={markings, mark=at position 0.6 with {\coordinate (m3);}}] (v2) to                 (v3);
    \path[decorate, decoration={markings, mark=at position 0.4 with {\coordinate (m4);}}] (v2) to [bend right=20] (v3);
    \path[decorate, decoration={markings, mark=at position 0.7 with {\coordinate (m5);}}] (v2) to [bend right=40] (v3);

    \begin{scope}
        \clip (m1) circle (4pt)
              (m2) circle (4pt)
              (m3) circle (4pt)
              (m4) circle (4pt)
              (m5) circle (4pt);
        
        \draw[red, very thick] (v2) to [bend left=40]  (v3);
        \draw[red, very thick] (v2) to [bend left=20]  (v3);
        \draw[red, very thick] (v2) to                 (v3);
        \draw[red, very thick] (v2) to [bend right=20] (v3);
        \draw[red, very thick] (v2) to [bend right=40] (v3);
    \end{scope}

    \draw (v3) to [bend left=40]  (v4);
    \draw (v3) to [bend left=20]  (v4);
    \draw (v3) to                 (v4);
    \draw (v3) to [bend right=20] (v4);
    \draw (v3) to [bend right=40] (v4);

    \draw (v4) to [bend left=40]  (v5);
    \draw (v4) to [bend left=20]  (v5);
    \draw (v4) to                 (v5);
    \draw (v4) to [bend right=20] (v5);
    \draw (v4) to [bend right=40] (v5);

    \draw (v5) to [bend left=40]  (v6);
    \draw (v5) to [bend left=20]  (v6);
    \draw (v5) to                 (v6);
    \draw (v5) to [bend right=20] (v6);
    \draw (v5) to [bend right=40] (v6);

\end{tikzpicture}    
    \caption{
      An illustration for the reduction from
      \textsc{Multicoloured Subgraph Isomorphism} to 
      \textsc{Paired Min Cut}.
      The graph represents a choice gadget
      for a vertex $a$ in $H$.
      The gadget consists of segments, each
      representing a possible image of 
      a vertex from $H$ in the graph $G$.
      Each segment is a union of $|E(H)|$ 
      (which is 5 in the illustration) 
      internally disjoint paths 
      between a pair of vertices.
      A cut corresponding to choosing $u$ in $G$
      as an image of $v$ is depicted using
      red edges.
    }
    \label{fig:msigadget}
\end{figure}

We will construct $D$ in such a way that 
every minimum $st$-cut in $D$ (viewed as a set of deleted edges)
is of size $2|E(H)|$ and 
contains exactly one edge $e_{a,b}^{u,v}$ for every 
ordered pair $(a,b) \in V(H)^2$ with $ab \in E(H)$.
Thus, a minimum $st$-cut may be viewed as a mapping from 
ordered pairs $(a,b)$ in $H$ to their images $(u,v)$ in $G$.
Furthermore, the construction will ensure
the following consistency property:
for every $a,b,c \in V(H)$,
if the cut ``maps'' $(a,b)$ to $(u_1, v_1)$ 
and $(a,c)$ to $(u_2, v_2)$, 
then $u_1 = u_2$, i.e.\ $a$ in the first coordinate
of a pair is always ``mapped'' to the same vertex of $G$.
 
The principal gadget for this construction is illustrated in Figure~\ref{fig:msigadget}, which is 
created for a vertex $a$ in $H$.
It consists of $|V_a(G)|$ segments,
one for each possible image of $a$ in $G$.
Each segment is a pair of vertices
connected by $|\delta_H(a)|$ internally disjoint paths, 
one path for every ordered pair $(a,b)$ with $ab \in E(H)$.
The segments are linked consecutively (in an arbitrary order), 
the first vertex of the first segment is identified with $s$ and the last vertex of the last segment is identified with $t$.
A minimum $st$-cut in this gadget has to cut all paths in one segment, which corresponds to selecting an image in $G$ for the vertex $a$.
We zoom into a segment for a vertex $u \in V_a(G)$.
As mentioned above, each path of the segment
corresponds to an ordered pair $(a,b)$ with $ab \in E(H)$.
The edges on the path are $e_{a,b}^{u,v}$ for all $v \in V_b(G)$
(linked in an arbitrary order).
Note that $a,b$ and $u$ are shared among all the edges of the path.
Now it is easy to verify that constructing $D$ with 
one such gadget for every vertex $a$ in $H$ 
satisfies the specification above.
Moreover, there is a natural way to translate a cut
into a mapping $h : V(H) \to V(G)$:
for every vertex $a$ in $H$, we look at which 
segment is cut in the gadget for $a$, 
and we set $h(a)$ to be the vertex of $G$ corresponding 
to that segment.

Now the role of bundles becomes apparent:
the cost of a minimum $st$-cut in $D$ is 
measured by the number of bundles it intersects,
so it is between $|E(H)|$ and $2|E(H)|$.
Moreover, if the cost is $(2 - \eps) \cdot |E(H)|$,
this means that the cut contains 
at least $\eps \cdot |E(H)|$ bundles.
Each bundle is of the form 
$\{e_{a,b}^{u,v}, e_{b,a}^{v,u}\}$, so if the cut
contains such a bundle, then we know that
$h(a) = u$ and $h(b) = v$, which means that $h$ preserves the edge $ab$.
The total number of edges preserved by $h$ is at least 
$\eps \cdot |E(H)|$, so $(2-\eps)$-approximating \textsc{Paired Min Cut} is at least as hard as $\eps$-approximating MSI.

Both \textsc{Paired Min Cut} and MSI
are concerned with pairs:
in \textsc{Paired Min Cut} the bundles are pairs of edges, and
in MSI we  preserve edges that are pairs of vertices.
In \pSplitng{} with $p > 2$, 
we have bundles of size $p$, 
so a suitable counterpart in MSI are 
$p$-tuples of vertices inducing $p$-cliques.
An important property which is trivially satisfied 
for $p=2$ but stops being true for $p > 2$ 
is that the edge set of any graph can be partitioned into $p$-cliques.
Disjointness of the $p$-cliques is crucial for the correctness argument:
from every bundle intersected by 
the cut in a non-trivial way
(i.e.\ in at least two elements),
we want to deduce that an edge of $H$ preserved by the mapping, 
so different bundles should correspond to different edges of $H$.
A natural idea now is to start with a graph $H$
and drop a constant fraction of edges to make
it $p$-clique-partitionable.
However, this is not possible for arbitrary graphs $H$ so
we cannot continue using the result of
Guruswami~et~al. because their hard instances
are sparse graphs.
Here we depart from the route taken by Dabrowski~et~al.
by starting from a result of
Bafna~et~al.~\cite{Bafna:etal:stoc2025} proving that
MSI is hard even if the graphs $H$ are cliques.
To ensure the $p$-clique-partitionability,
for every $k$  we use Turán's theorem to construct graphs
$\Htk$ with $t \cdot k$ vertices and $\Omega_t(k^2)$ edges 
such that $E(\Htk)$ can be partitioned into $t$-cliques.

The encoding into \pSplitng{} is structurally similar to the one for \textsc{Paired Min Cut}.
We start with a partition $\cP$ of $E(H)$ into $p$-cliques,
and for every $P \in \cP$
and every possible image $I$ of $P$ in $G$,
we create $p$ edges $(e^I_P)_1, \dots, (e^I_P)_p$,
one for each coordinate $i \in [p]$.
In the choice gadget for a vertex $a$ in $H$,
we select an image $u \in V(G)$ for $a$
and images for every $p$-clique of $H$ containing $a$.
The $st$-cuts are thus of size $p \cdot |\cP|$.
The correctness argument is more involved than before.
Informally, if a cut is of cost $(p - \eps) \cdot |\cP|$,
then it intersects at least $\eps \cdot |\cP|$ bundles
in at least two edges;
since the bundles correspond to cliques,
picking up two edges from the bundle 
implies that an edge from the clique is preserved by the mapping.
Thus, the mapping preserves at least 
$(\eps/p) \cdot |\cP|$ edges, which is a constant fraction
of the number of edges in $H$,
and we obtain the following result.

\begin{theorem} \label{thm:intropsplit}
  \pSplitng\ is not FPT-approximable within a factor of $p-\eps$ for
  any $\eps>0$ unless the \ETH is false.
\end{theorem}

\bigskip

\noindent
\large
\textbf{Roadmap.}
\normalsize The structure of the paper closely follows the outline given in the introduction.
We present the necessary preliminaries in Section~\ref{sec:prelims}
and balanced subgraph coverings are introduced in 
Section~\ref{sec:balanced-subgraph-coverings}.
The upper and lower bounds on FPT approximability are collected in
Sections~\ref{sec:upperbounds} and \ref{sec:hardness}, respectively.
We conclude the paper in Section~\ref{sec:discussion} with a discussion of
our results.

\section{Preliminaries}
\label{sec:prelims}

\paragraph*{Groups and Rings.} 
A {\em group} $\Gamma$ is a set 
of elements
together with a binary operation $\circ$
such that (1) $\circ$ is associative, (2) there is a identity element $0$ in $\Gamma$ with respect to $\circ$, and (3) every element $g$ in $\Gamma$ has a unique inverse $g^{-1}$ such that $g \circ g^{-1}=0$.
We say that $\Gamma$ is {\em Abelian} when $\circ$ is commutative.
A {\em ring} is an Abelian group (whose operation $+$ is called {\em addition})
equipped with a second binary operation $\cdot$ (called {\em multiplication}) that is (1) associative, 
(2) is distributive over the addition operation, and (3) has an identity element $1$.
We will exclusively consider {\em commutative rings}, 
where multiplication is a commutative operator. 
Let $R=(R,+,\cdot)$ be such a ring where
$0$ denotes the additive identity element and 
$1$ the multiplicative identity element.
In this paper we will be mostly working with finite rings $\ZZ_m$,
where the elements are integers $\{0,\dots,m-1\}$ and the operations
are defined modulo $m$.
A {\em unit} in $R$ is an invertible element for the multiplication of the ring, i.e. an element $d \in R$ is a unit if there exists $d' \in R$ such that
$d \cdot d'=1$. 
For instance, the units in $\ZZ_4$ are $1$ and $3$.
One may note that the units of $R$ under multiplication form an Abelian group. 

The {\em direct sum} of two rings 
 $R_1=(D_1;+_1,\cdot_1)$ and $R_2=(D_2;+_2,\cdot_2)$ 
 is denoted $R_1 \oplus R_2 = (R; +, \cdot)$.
 Its domain $R$ consists of the ordered pairs 
 $\{(d_1,d_2) \mid d_1 \in D_1, d_2 \in D_2\}$ 
 and the operations are defined coordinate-wise:
 $(d_1,d_2)+(d'_1,d'_2)=(d_1+_1 d'_1,d_2+_2 d'_2)$ and
 $(d_1,d_2) \cdot (d'_1,d'_2)=(d_1 \cdot_1 d'_1,d_2 \cdot_2 d'_2)$.
The following is one (out of many possible) ways
of formulating the well-known Chinese Remainder Theorem.

\begin{proposition}
\label{prop:crt}
Let $m \geq 2$ be an integer and $p_1^{n_1} \cdot \dots \cdot p_\ell^{n_\ell}$ be the prime factorization of $m$. The ring $\ZZ_m$ is isomorphic to the direct sum
$\bigoplus_{i=1}^{\ell} \ZZ_{p_i^{n_i}}$, and the
map $f(x) \mapsto (x \mod p_1,\dots,x \mod p_n)$
defines a ring isomorphism from $\ZZ_m$ to 
$\bigoplus_{i=1}^{\ell} \ZZ_{p_i^{n_i}}$.
\end{proposition}

\paragraph*{Linear equations.}

Let $R$ denote a commutative ring.
An expression $c_1 \cdot x_1+\dots+c_r \cdot x_r=c$ is a {\em (linear) equation over $R$}  
if $c_1,\dots,c_r,c \in R$ and $x_1,\dots,x_r$ are variables with domain $R$.
Let $S$ denote a set (or equivalently a system) of equations over $R$.
We let $V(S)$ denote the variables appearing in $S$, and we
say that $S$ is {\em consistent} if there is an assignment
$\varphi : V(S) \rightarrow R$
that satisfies all equations in $S$. 
An instance of the computational problem \textsc{$r$-Lin$(R)$} %
is a system $S$ of equations on at most $r$ variables
over $R$, and the question is whether $S$ is consistent.
Linear equation systems over $\ZZ_m$ are solvable
in polynomial time
and the well-known procedure is outlined, for instance, in~\cite[p. 473]{Arvind:Vijayaraghavan:stacs2005}.

Let us now consider the case when
we allow some equations in an instance to be
soft (i.e. deletable at unit cost) and
crisp (i.e. undeletable).
We study the following computational problem.

\pbDefP{$\minlin{r}{R}$}
{A (multi)set $S$ of equations over $R$ with at most $r$
variables per equation,
a subset $S^\infty \subseteq S$ of crisp equations
and an integer $k$.}
{$k$.}
{Is there a set $Z \subseteq S \setminus S^\infty$ 
such that $S - Z$ is consistent and 
$|Z| \leq k$?}

We use crisp equations merely for convenience 
since they can be modelled by $k+1$ copies of 
the same soft equation.
For an assignment $\alpha : V(S) \to R$,
let $\cost_S(\alpha)$ be $\infty$ if $\alpha$
does not satisfy a crisp equation and
the number of violated soft equations otherwise.
We drop the subscript $S$ when it is clear from context.
We also write $\mincost(S)$ to denote the minimum cost
of an assignment to $S$.

\paragraph*{Graphs.}

We use the following graph-theoretic terminology in what follows.
Let $G$ be an undirected graph.
We write $V(G)$ and $E(G)$ to denote the vertices and edges of $G$, respectively.
For every vertex $v \in V(G)$,
let the {\em neighbourhood of $v$ in $G$} denoted by $N_G(v)$ be the set
$\{u \in V(G) \mid \{u, v\} \in E(G)\}$.
We extend this notion to sets of vertices $S \subseteq V(G)$
in a natural way: $N_G(S)= (\bigcup_{v \in S} N_G(v)) \setminus S$.
If $U \subseteq V(G)$, then the {\em subgraph of $G$ induced
by $U$} is the graph $G'$ with
$V(G')=U$ and $E(G')=\{ vw \mid v,w \in U \; {\rm and} \; vw \in E(G)\}$.
We denote this graph by $G[U]$.
If $Z$ is a subset of edges in $G$, we write $G-Z$ to 
denote the graph~$G'$ with $V(G') = V(G)$ and $E(G') = E(G) \setminus Z$.
For a set of vertices $X \subseteq V(G)$, we denote by 
$\delta_G(X)$ the set of edges in $G$ with exactly one endpoint in $X$.
For $X,Y \subseteq V(G)$, 
an \emph{$(X,Y)$-cut} is a subset of edges $Z$
such that $G - Z$ does not contain a path with 
one endpoint in $X$ and another in $Y$.
When $X,Y$ are singleton sets $X=\{x\}$ and
$Y=\{y\}$, we simplify the notation and write $xy$-cut instead
of $(X,Y)$-cut. 
An edge $e$ in a graph $G$ is called a bridge
if $G - e$ has more connected components than $G$.

\paragraph*{Biased Graphs.}

Biased graphs are combinatorial
objects of importance especially to matroid theory~\cite{Zaslavsky:jctb89}.
To introduce biased graphs,
recall that
a \emph{theta graph} is a graph consisting of two distinct vertices
connected by three internally vertex-disjoint paths. A theta graph is illustrated in Figure~\ref{fig:theta}.
Note that a theta graph contains three cycles,
each formed by the union of two out of the three paths.
A \emph{biased graph} is a pair $(G, \BB)$ where $G$ is an undirected graph
and $\BB$ is a collection of cycles in $G$, called \emph{balanced cycles},
with the following property: for any theta subgraph of $G$,
if two of its cycles are balanced, then so is the third cycle.
For example, the set of even-length cycles in any graph
satisfies this property: if two even cycles $C$ and $C'$
are part of a theta graph, then 
the third cycle $C \triangle C'$ is also even.
In contrast, the set of odd-length cycles does not necessarily
satisfy this property: if two odd cycles $C$ and $C'$
forming a theta graph share one edge, then the third cycle
$C \triangle C'$ has even length.
Note that this is equivalent to the form given in the introduction:
for any unbalanced cycle $C$, and any path $P$ whose endpoints $u, v$ are in $C$ and internal vertices (if any) are disjoint from $C$, at least one of the two new cycles in $C+P$ is also unbalanced.

An important example of biased graphs are \emph{group-labelled graphs},
also called \emph{$\Gamma$-gain graphs} in the literature.
Let $\Gamma$ be a group and $G$ an undirected graph. 
We define an oriented edge in $G$ as a triple $(e,u,v)$ where $e \in E(G)$ 
and $u$ and $v$ are the endpoints of $e$. A group labelling of $G$ with group $\Gamma$
is then a function $\gamma$ from the oriented edges of $G$ to $\Gamma$
satisfying $\gamma(e,u,v)=\gamma(e,v,u)^{-1}$. 
A cycle $C$ in $G$ is \emph{identity} if the product of the labels
of the oriented edges of $C$ (in the orientation and direction of their occurrence on $C$) 
equals the identity element of $\Gamma$.
Note that this does not depend on the starting point or direction of traversal of $C$, and so is a well-defined notion.
We have the following result.

\begin{lemma}%
\label{lem:minlin-fin-to-rbgce}
  Let $G$ be a graph with edges labelled
  by elements of a finite group $\Gamma$.
  Let $\BB$ be the set of identity cycles in $G$.
  Then $(G, \BB)$ is a biased graph and $\BB$ has a polynomial-time
  membership oracle.
\end{lemma}
\begin{proof}
Zaslavsky~\cite[Proposition~5.1]{Zaslavsky:jctb89} has proven
that $(G, \BB)$ is a biased graph.
A membership oracle for a group-labelled
graph only requires the ability to test whether
a product $\gamma_1 \cdot \ldots \cdot \gamma_\ell$
is identity in $\Gamma$, and this is trivial
when $\Gamma$ is finite
\end{proof}

Naturally, the set of balanced cycles in a biased graph
may be exponentially large in the size of the graph,
so it is essential that $\BB$ is provided via an efficient oracle. 

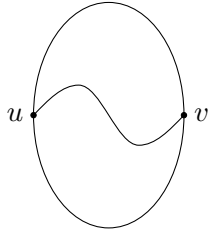
\begin{figure}[bt]
\centering
\begin{tikzpicture}
  \draw (0,0) ellipse (1cm and 1.5cm);
  \filldraw[black] (-1,0) circle (1pt) node[anchor=east]{$u$};
  \filldraw[black] (1,0) circle (1pt) node[anchor=west]{$v$};

  \draw (-1,0) sin (-0.4,0.4) cos (0,0) sin (0.4,-0.4) cos (1,0);
\end{tikzpicture}
\caption{The structure of a theta graph.}
\label{fig:theta}
\end{figure}

\paragraph*{Parameterized Complexity.}
In parameterized 
algorithmics~\cite{book/DowneyF99,book/FlumG06,book/Niedermeier06}
the runtime of an algorithm is studied with respect to
the input size~$n$ and a parameter $p \in \NN$.
The basic idea is to find a parameter that describes the structure of
the instance such that the combinatorial explosion can be confined to
this parameter.
In this respect, the most favourable complexity class is \FPT
(\emph{fixed-parameter tractable}),
which contains all problems that can be decided by an algorithm
running in $f(p)\cdot n^{O(1)}$ time, where $f$ is a computable
function.
Problems that can be solved within such a time bound are said to be 
\emph{fixed-parameter tractable} (FPT).

A {\em parameterized problem} is, formally speaking, a subset of $\Sigma^* \times {\mathbb N}$
where $\Sigma$ is the finite input alphabet. Reductions between parameterized problems need to take
the parameter into account. To this end, we will use {\em parameterized reductions} (or FPT-reductions).
Let $L_1$ and $L_2$ denote parameterized problems with $L_1 \subseteq \Sigma_1^* \times {\mathbb N}$
and $L_2 \subseteq \Sigma_2^* \times {\mathbb N}$. 
A parameterized reduction from $L_1$ to $L_2$ is a
mapping $P: \Sigma_1^* \times {\mathbb N} \rightarrow \Sigma_2^* \times {\mathbb N}$
such that
(1) $(x, k) \in  L_1$ if and only if $P((x, k)) \in L_2$, (2) the mapping can be computed
by an FPT-algorithm with respect to the parameter $k$, and (3) there is a computable function $g : {\mathbb N} \rightarrow {\mathbb N}$ 
such that for all $(x,k) \in L_1$ if $(x', k') = P((x, k))$, then $k' \leq g(k)$.
One may, for instance, note that if $L_1$ is in FPT
and $L_2$ FPT-reduces to $L_1$, then $L_2$ is 
in FPT, too.

\paragraph*{Parameterized Approximation.}
Let $c \geq 1$ be a constant. 
Formally, a factor-$c$ {approximation
algorithm} takes an instance $(I,k)$ of a minimization problem $\Pi$, 
accepts if $(I, k)$ is a yes-instance and 
rejects if $(I, c \cdot k)$ is a no-instance.
The running time of an \FPT-approximation 
algorithm is bounded by 
$f(k) \cdot \norm{I}^{O(1)}$ where $f: \NN \rightarrow \NN$
is some computable function. 
We use the following simple observation
several times in the sequel. Assume $\Pi_1$ and $\Pi_2$
to be minimization problems and $\Pi_1$ is \FPT-approximable
within factor $c$. If there is an \FPT-reduction
from $\Pi_2$ to $\Pi_1$
that does not change the parameter, then $\Pi_2$
is \FPT-approximable within factor $c$, too.

\paragraph*{MinSat.}

A Boolean relation $R$ of arity $r$ is a subset of tuples in $\{0,1\}^r$.
A set of relations is called a \emph{constraint language}.
Let $\bA$ be a Boolean constraint language.
An instance of a \emph{satisfiability problem over $\bA$},
denoted by $\Sat(\bA)$, consists of
a set of variables $V$ and a (multi)set of constraints $\cC$, each constraint
of the form $R(x_1,\dots,x_r)$, where $R \in \bA$ is a relation of arity $r$ and
$x_1,\dots,x_r$ are (not necessarily distinct) 
variables from $V$.
An assignment $\alpha : V \to A$ \emph{satisfies} a constraint
$R(x_1,\dots,x_r)$ if $(\alpha(x_1),\dots,\alpha(x_r)) \in R$,
otherwise we say that it \emph{violates} the constraint.
The \emph{cost} of $\alpha$ in an instance $(V, \cC)$ of $\CSP(\bA)$
is the number of constraints violated by $\alpha$ in $\cC$.

The optimization problem we are interested in
is known as the \textsc{Minimum-Cost Satisfiability 
over $\bA$}, or $\MinSat(\bA)$ for short.

\pbDefP{$\MinSat(\bA)$}
{An instance $I = (V, \cC)$ of $\CSP(\bA)$
with a subset $\cC^\infty \subseteq \cC$ of crisp constraints,
and an integer $k$.}
{$k$.}
{Does $I$ admit an assignment of cost at most $k$?}

The constraints in $\cC \setminus \cC^\infty$ 
are called \emph{soft}.
As with $\minlin{r}{R}$,
we use crisp constraints merely for convenience 
as they can be modelled by $k+1$ copies of 
the same constraint.
The cost of an assignment that violates a crisp constraint
is defined to be infinite.

Kim~et~al.~\cite{Kim:etal:sicomp2025} classified the complexity
of $\MinSat(\bA)$ parameterized by $k$ for every Boolean
constraint language $\bA$.
We will use the positive part of their classification.
Every Boolean relation $R \subseteq \{0,1\}^r$
is definable by a propositional formula on $r$ variables:
formally, a tuple $(b_1,\dots,b_r) \in \{0,1\}^r$ 
belongs to $R$
if and only if
$\phi_R(b_1,\dots,b_r)$ is true.
A relation $R$ is \emph{bijunctive} if
it is definable by a conjunction of $1$- or $2$-clauses.
Given two literals $l_1,l_2$, we write $(l_1 \rightarrow l_2)$ 
(implication) as a short-hand for the 2-clause $(\neg l_1 \vee l_2)$,
and we write $(l_1 \leftrightarrow l_2)$ (equivalence) as a short-hand
for the conjunction of two clauses $(l_1 \rightarrow l_2) \land (l_2 \rightarrow l_1)$. 

Every bijunctive relation $R$ admits a unique (up to permuting conjuncts) \emph{complete} 
bijunctive formula $\phi_R$ that defines it.
To derive $\phi_R$ from $R$, for every pair of indices $1 \leq i < j \leq r$
consider the set $R_{i,j} = \{ (b_i, b_j) : (b_1,\dots,b_r) \in R\}$.
\begin{itemize}
  \item If $(0,0) \notin R_{i,j}$, then add $(x_i \lor x_j)$ to $\phi_R$.
  \item If $(0,1) \notin R_{i,j}$, then add $(\lnot x_i \lor x_j)$ to $\phi_R$.
  \item If $(1,0) \notin R_{i,j}$, then add $(x_i \lor \lnot x_j)$ to $\phi_R$.
  \item If $(1,1) \notin R_{i,j}$, then add $(\lnot x_i \lor \lnot x_j)$ to $\phi_R$.
\end{itemize}
The \emph{Gaifman graph} $G_R$ of
a bijunctive relation $R$ is defined as follows.
Let $\phi_R$ be the complete formula on
variables 
$x_1,\dots,x_r$ that defines $R$;
the vertex set of $G_R$ is $\{1,\dots,r\}$ and 
the edge set consists of pairs $ij$ 
for $1\leq i<j\leq r$ such that
$\phi_R$ contains a clause on variables $x_i$ and $x_j$.
We say that a graph $G$ is \emph{$2K_2$-free} if there is no
induced subgraph of $G$ isomorphic to $2K_2$, i.e.\ the graph consisting
of two vertex-disjoint edges. 
Equivalently, $G$ is $2K_2$-free if and only if
for every pair of edges $u_1 v_1$ and $u_2 v_2$ in $G$,
there is also an edge with an endpoint in $\{u_1,v_1\}$
and in $\{u_2,v_2\}$.

\begin{theorem}[Theorem~1.2 in \cite{Kim:etal:sicomp2025}] \label{thm:2k2-free}
  Let $\bA$ be a finite, bijunctive Boolean constraint language.
  If for every relation $R \in \bA$ the Gaifman graph $G_R$ is
  $2K_2$-free, then $\MinSat(\bA)$ is in \FPT.
\end{theorem}

\section{Balanced Subgraph Covering}
\label{sec:balanced-subgraph-coverings}

Let $(G, \BB)$ be a biased graph and let $H$ be a subgraph of $G$.
The \emph{cost of $H$} is defined as 
the number of edges in $G$ incident with $V(H)$ but not present in $H$, i.e.
\[
  \cost(H) = |\delta_G(V(H))| + |E(G[V(H)]) \setminus E(H)|.
\]

\medskip

We show the following \emph{balanced subgraph covering} result, which we recall from the introduction.

\thmBiasSample*

We also provide a derandomization with essentially the same performance; see Theorem~\ref{thm:derandom_cover}.

If the procedure is successful for a given $H$, we say that \emph{$S$ covers $H$}. 
As an example, consider the bipartization case where $(G,\BB)$ is a biased graph
such that $\BB$ contains all cycles of even length. Then a subgraph $H$ as in  Theorem~\ref{ithm:balanced-shadow-sampling} is any subgraph such that
(1) $G[V(H)]$ is nearly bipartite and (2) $\delta_G(V(H))$ is small, which is captured by $\cost(V(H)) \leq k$. For any such $H$, if $S$ covers $H$, then $F$ contains an optimal cleaning set for $H$.
This fits the applications of Theorem~\ref{ithm:balanced-shadow-sampling}
in the algorithm for $\minlin{2}{\ZZ_{p^d}}$; see Section\ref{sec:the-algorithm}.

The rest of the section is structured as follows. Section~\ref{sec:balance:important}
covers definitions and properties of balanced subgraphs and related notions. 
Section~\ref{sec:balance:procedure} contains the proof of Theorem~\ref{ithm:balanced-shadow-sampling}. Section~\ref{sec:balance:derand} gives the derandomization result.
Finally, in Section~\ref{sec:balance:impsep} we make the connection to the ``shadow removal'' 
procedure of Marx and Razgon~\cite{marx2014fixed} explicit and prove Corollary~\ref{icor:sample-important} announced in the introduction.

\subsection{Important Balanced Subgraphs and Vertex Subsets}
\label{sec:balance:important}

To state the result we are using, we need additional terminology.
Let $H$ and $H'$ be balanced subgraphs of $G$ with respect to $\BB$. 
We say that $H'$ \emph{dominates} $H$ if 
$V(H) \subseteq V(H')$ and $\cost(H) \geq \cost(H')$,
and that $H'$ \emph{strictly dominates} $H$ if additionally at least one of 
these two conditions is strict. 
If two balanced subgraphs $H$ and $H'$ mutually dominate each other,
we call them \emph{domination equivalent}.
A balanced subgraph $H$ of $G$ is \emph{important} if 
no balanced subgraph in $G$ strictly dominates $H$.
Dabrowski et al.\ show the following~\cite{minlin25talg}.

\begin{theorem}[Theorem~5~in~\cite{minlin25talg}] 
  \label{thm:important-subgraphs}
  Let $(G, \BB)$ be a biased graph.
  There is an algorithm that takes as input $(G, \BB)$,
  a vertex $v \in V(G)$ and an integer $k$, and 
  in $\bigoh^*(4^k)$ time computes a family $\cH$ 
  of connected important balanced subgraphs of $G$ containing $v$
  such that the following holds:
  for every connected balanced subgraph $H$ of $G$
  containing $v$ with cost at most $k$,
  there is $H' \in \cH$ that dominates $H$.
\end{theorem}

As an illustrative example, we relate this to the notion of important separators due to Marx~\cite{Marx:tcs2006}.
Let $G$ be a graph and $s, t \in V(G)$ be vertices. Define a biased graph from $G$ by attaching $W \geq k+1$ internally disjoint cycles to $t$, and let every cycle except these be balanced. 
As Dabrowski et al.~\cite{minlin25talg} note, this defines a biased graph where a balanced subgraph $H$ (of cost at most $k$) is precisely a subgraph $H$ such that $t \notin V(H)$. 
Now apply Theorem~\ref{thm:important-subgraphs} to $(G,\BB)$ with starting vertex $s$.
Since there are no unbalanced cycles except those on $t$, the important connected balanced subgraphs are induced subgraphs $H=G[X]$ with $s \in X$, $t \notin X$ and $G[X]$ connected, such that $|\delta(X)|$ is minimized and $X$ is maximized.

Marx~\cite{Marx:tcs2006} defined important subgraphs as follows (for the case of edge cuts).
Let $G$ be a graph and $X, Y \subseteq V(G)$ disjoint subsets. 
An $(X,Y)$-separator is an edge set $S \subseteq E(G)$ such
that $X$ and $Y$ are separated in $G-S$, i.e.\ no component 
of $G-S$ contains vertices of both $X$ and $Y$. It is
\emph{inclusion-wise minimal} if no strict subset of $S$ is an
$(X,Y)$-separator. For an $(X,Y)$-separator $S$, let $R(X,S)$
denote the vertices reachable from $X$ in $G-S$. 
An \emph{important $(X,Y)$-separator} is an inclusion-wise minimal
$(X,Y)$-separator $S$ such that there is no $(X,Y)$-separator $S'$
such that $R(X,S) \subset R(X,S')$ and $|S'| \leq |S|$.
By Marx~\cite{Marx:tcs2006} and Chen et al.~\cite{chen2009improved}
we know that there are at most $4^k$ important separators of cost at most $k$
and they can be enumerated efficiently. 

This reduces back to the application of Theorem~\ref{thm:important-subgraphs}: indeed, for undirected edge cuts we can contract $X$ and $Y$ into single vertices $s$ and $t$, 
and as a result the important $(X,Y)$-separators correspond directly to the cuts $\delta(V(H))$ for the connected important balanced subgraphs $H$ computed by Theorem~\ref{thm:important-subgraphs}.
This connection is developed more in Section~\ref{sec:balance:impsep}.

We stress one difference between the result of 
Marx~\cite{Marx:tcs2006} about important separators
and Theorem~\ref{thm:important-subgraphs}:
the former enumerates \emph{all} important separators
of cost at most $k$, while the latter only computes
a dominating family.
There is a good reason for this:
the number of important balanced subgraphs
of cost at most $k$ is not necessarily bounded by a function of $k$.
For example, if $G$ is a cycle of length $n$ and $\BB$ is empty,
then $(G, \BB)$ admits $n$ important balanced subgraphs of cost $1$,
each obtained by deleting one edge from $G$.
Observe that in this example all important balanced subgraphs
of cost $1$ share the same vertex set.
Indeed, the vertex set of an important balanced subgraph
seems to be a more robust object than the subgraph itself.
This leads us to the following definitions.

\smallskip

An edge subset $F \subseteq E(G)$ is a \emph{cleaning set} if
$G - F$ is balanced.
The \emph{cost of a vertex set $X \subseteq V(G)$} is defined as
\[ 
  \cost(X) = |\delta_G(X)| + 
  \min \{ |F_X| : F_X \text{ is a cleaning set for } G[X] \}.
\]
Note that $\cost(X)$ is equal to the minimum cost of
a balanced subgraph $H$ of $G$ with $V(H) = X$.
For a pair of vertex sets $X, Y \subseteq V(G)$, we say that
$X$ \emph{dominates} $Y$ if $X \supseteq Y$ and $\cost(X) \leq \cost(Y)$,
and that $X$ \emph{strictly dominates} $Y$ if
at least one of these two conditions is strict.
A vertex set $X$ is an \emph{important subset} if 
it is not strictly dominated by another vertex set.

\begin{proposition}
  \label{prop:really_important}
  Every important subset is the vertex set
  of an important balanced subgraph.
\end{proposition}
\begin{proof}
  Let $X$ be an important subset.
  By definition, there exists a balanced subgraph
  $H$ of $G$ with $V(H) = X$ and of the same cost as $X$.
  To see that $H$ is important,
  suppose a balanced subgraph $H'$
  strictly dominates $H$.
  Then its vertex set $V(H')$ strictly dominates $X$.
\end{proof}

With this terminology, we derive the following consequence of Theorem~\ref{thm:important-subgraphs}.

\begin{lemma}
  \label{lem:important-subsets}
  Let $(G, \BB)$ be a biased graph,
  $v \in V(G)$ be a vertex and $k$ be an integer.
  Let $\cX_v$ be the family of important subsets
  of $V(G)$ containing $v$ and of cost at most $k$.
  Then, $|\cX_v| \leq 4^k$, and
  there exists an algorithm that
  takes as input $(G, \BB, v, k)$,
  and in $\bigoh^*(4^k)$ time computes the family $\cX_v$
  together with a minimum-cardinality
  cleaning set $F_X \subseteq E(G[X])$ for each $X \in \cX_v$.
\end{lemma}
\begin{proof}
  Let $\cH$ be the family of important balanced subgraphs
  computed by Theorem~\ref{thm:important-subgraphs}.
  We claim that $\{ V(H) : H \in \cH \}$ contains $\cX_v$.
  Indeed, let $X \in \cX_v$ be an important subset
  with $v \in X$ and $\cost(X) \leq k$.
  By Proposition~\ref{prop:really_important},
  there exists an important balanced subgraph
  $H_X$ of $G$ with $V(H_X) = X$.  
  By Theorem~\ref{thm:important-subgraphs},
  the family $\cH$ contains a balanced subgraph $H'$
  that dominates $H_X$.
  Since $X$ is important, $V(H')$ does not 
  strictly dominate $X$, implying that $V(H') = X$.
  This shows that $|\cX_v| \leq |\cH| \leq 4^k$.
  Furthermore, $\cX_v$ is precisely $\{ V(H) : H \in \cH \}$
  because no two members of the latter set strictly dominate each other.

  Finally, for each $X \in \cX_v$ we have 
  a subgraph $H \in \cH$ with $V(H) = X$,
  and $F_X = E(G[X]) \setminus E(H)$
  is a cleaning set for $G[X]$.
  Since $H$ is important, $F_X$ is of 
  minimum cardinality among all cleaning sets for $G[X]$.
\end{proof}

\subsection{Sampling Procedure} \label{sec:balance:procedure}

We now provide the sampling procedure and prove Theorem~\ref{ithm:balanced-shadow-sampling}.
The proof is inspired by the proof of Marx and Razgon~\cite{marx2014fixed} 
for their \emph{shadow removal/shadow covering} routine, i.e.\ \emph{random sampling of important separators}, except with success probability $1/2^{O(k^2)}$, matching Chitnis et al.~\cite{chitnis2015directed} (since we are only concerned with the edge deletion case).

Let $(G,\BB)$ be a biased graph and $k \in \NN$ the cost bound.
We may assume that every connected component of $G$ is unbalanced, as any balanced components can just be included in the sample set $S$ unmodified with no presence in $F$.
We describe the steps of the sampling procedure in
Figure~\ref{fig:sampling}.

Towards proving Theorem~\ref{ithm:balanced-shadow-sampling},
we show a result on the structure of important not-necessarily-connected balanced subgraphs.

\begin{lemma} \label{lem:important-components}
  Let $(G, \BB)$ be a biased graph, where no connected component of $G$ is balanced. 
  Let $H$ be an important balanced subgraph of $G$
  of cost at most $k$.
  Then $H$ contains at most $k$ connected components,
  every subset $K \subseteq V(H)$ inducing 
  a connected component of $H$ is an important vertex subset of $(G, \BB)$,
  and its neighbourhood $N_G(K)$ is disjoint from $V(H)$.
\end{lemma}
\begin{proof}
  Let $\cK$ be the partition of $V(H)$ into 
  maximal subsets inducing connected components of $H$.
  We first show that for every $K \in \cK$, $N_G(K)$ is disjoint from $V(H)$.
  Assume otherwise, and let $uv \in E(G)$ connect two distinct components
  $K_u, K_v \in \cK$ in $H$. The graph $H'=H+uv$ is also balanced, and has
  strictly lower cost than $H$. Indeed, since $uv$ is a bridge in $H'$, 
  it cannot participate in an unbalanced cycle. It follows that
  \[
    \cost(H)=\sum_{K \in \cK} \cost(H[K]) = \sum_{K \in \cK} (|F_K|+|\delta_G(K)|)
  \]
  where for each $K \in \cK$, $F_K$ is a minimum cleaning set for $G[K]$.
  Furthermore, the graph $G[V(H)]-\bigcup_K F_K$ is domination equivalent to $H$.  
  Now assume towards a contradiction that some $K \in \cK$ is not important,
  i.e.\ that there exists a vertex subset $K'$ that strictly dominates $K$,
  and let $F_{K'}$ be a minimum cleaning set for $G[K']$. 
  Let $F'=(F \setminus (F_K \cup \delta_G(K))) \cup (F_{K'} \cup \delta_G(K'))$
  and define $H'$ with $V(H')=(V(H) \setminus K) \cup K'$ 
  and $E(H')=E(G[V(H')]) \setminus F'$. Note that $H'$ is balanced,
  since every cycle in $H'$ is either contained in $K'$ or exists in $H$.
  Furthermore $\cost(V(H')) \leq |F'|=|F|-\cost(K)+\cost(K') \leq |F|$,
  and either $|F'|<|F|$ or $K \subset K'$.
  Thus the vertex set $V(H')$ strictly dominates $V(H)$,
  contradicting the assumption that $H$ is important by Proposition~\ref{prop:really_important}.

  Finally, to argue that $H$ has at most $k$ components,
  for every $K \in \cK$ either $F_K \neq \emptyset$ or $K$ is a strict
  subset of some connected component of $G$ because no component in its entirety is balanced, 
  in which case $\delta_G(K) \neq \emptyset$. 
  Thus $\cost(H) = \sum_K \cost(H[K])\geq |\cK|$ by the above. 
\end{proof}

\begin{figure}
\fbox{\parbox{\textwidth}{

Algorithm: $\RandomCover(G, \BB, k)$
\begin{enumerate}
\item Compute a family $\cX = \bigcup_{v \in V(G)} \cX_v$ of important connected subsets
  by taking the union of the families $\cX_v$ computed 
  by applying Theorem~\ref{thm:important-subgraphs}
  for every $v \in V(G)$ using the cost bound $k$.
  For each $X \in \cX$, let $F_X$ be the cleaning set for $G[X]$
  computed by the algorithm.
\item Sample a subfamily $\cX' \subseteq \cX$ 
  by including each element $X$ independently at random
  with probability $p = 1/4^k$.
\item Let $U' \subseteq V(G)$ be the union of all vertex sets in $\cX'$.
\item A set $X \in \cX$ is \emph{chosen} if $X \subseteq U'$ and $N_G(X) \cap U' = \emptyset$.
  Let $S \subseteq V(G)$ be the union of chosen sets, and
  $F \subseteq E(G)$ be the union of the edge sets $F_X \cup \delta(X)$ for all chosen sets $X$.
\item Add to $S$ the vertex set of any balanced connected component of $G$
\item Return $(S, F)$
\end{enumerate}
}}
\caption{Algorithm for computing balanced subgraph covers}
\label{fig:sampling}
\end{figure}

We are now ready to prove Theorem~\ref{ithm:balanced-shadow-sampling}.

\begin{proof}[Proof of Theorem~\ref{ithm:balanced-shadow-sampling}]
  Let $(G, \BB)$ be a biased graph, $k \in \NN$,
  and let $(S, F)$ be the output of the sampling procedure.
  Observe that $F$ contains $\delta_G(X)$ for every chosen set $X$,
  hence $\delta_G(S) \subseteq F$.
  Moreover, for every chosen set $X$,
  the subgraph $G[X] - F_X$ is balanced,
  hence $(G-F)[S]$ is balanced as well.

  We argue that we can dispose of balanced connected components of $G$.
  If $G$ contains some connected component $K$ that is balanced,
  note that for every $v \in K$ we have $\cX_v=\{K\}$ with 
  cleaning set $F_K=\emptyset$. Thus $F$ contains no edge of $K$,
  and $K \subseteq S$ by design. Now let $H$ be a balanced subgraph of $G$,
  let $G'$ be the union of unbalanced components of $G$, and 
  let $H'=H[V(G')]$. Then $V(H) \setminus V(H') \subseteq S$
  and $F$ contains zero edges with an endpoint in $V(H) \setminus V(H')$;
  thus the theorem holds for $H$ as a subgraph of $G$ if and only if
  it holds for $H'$ as a subgraph of $G'$, and we assume in what follows
  that every connected component of $G$ is unbalanced.
  
  Now, let $H$ be a balanced subgraph of $G$ of cost at most $k$.
  We may assume without loss of generality that $H$ is important:
  indeed, if $H'$ is an important balanced subgraph
  that dominates $H$, and $S$ covers $H'$, then $S$ also covers $H$.
  On a high-level, we will show that
  there exist two disjoint subfamilies $\cA_H, \cB_H \subseteq \cX$,
  each of size bounded by a function of $k$,
  such that sampling every member of $\cA_H$ and
  no member of $\cB_H$ guarantees that $S$ covers $H$.
  The probability of this event is $p^{|\cA_H|} \cdot (1-p)^{|\cB_H|}$,
  which is also bounded by a function of $k$.

  To this end, let $\cA_H$ be the family of vertex subsets inducing 
  connected components of $H$.
  By Lemma~\ref{lem:important-components}, since $G$ contains no balanced component
  by assumption, we have $|\cA_H| \leq k$ and 
  $\cA_H \subseteq \cX$ because each member of $\cA_H$ is important.
  Define $\cB_H$ to be the family of all sets $Y \in \cX$ that intersect $N_G(V(H))$.
  By Lemma~\ref{lem:important-components}, the families $\cA_H$ and $\cB_H$ are disjoint.
  Moreover, $|N_G(H)| \leq k$ since $H$ is of cost at most $k$,
  and, by Lemma~\ref{lem:important-subsets}, 
  there are at most $4^k$ sets in $\cX$ containing a vertex $v \in N_G(H)$,
  hence $|\cB_H| \leq 4^k \cdot k$.
  Clearly, if every set in $\cA_H$ is chosen, then $V(H) \subseteq S$
  and $F$ contains at most $k$ edges incident to $V(H)$, namely 
  the edges $F_X$ and $\delta(X)$ for each $X \in \cA_H$.
  The probability that each set in $\cA_H$ is sampled is $$p^{|\cA_H|} \geq 1/4^{k^2}.$$
  Moreover, if no set in $\cB_H$ is sampled,
  then for every $K \in \cA_H$ we have $N_G(K) \cap S = \emptyset$,
  and $K$ is chosen.
  The probability of the latter event is 
  $$(1-p)^{|\cB_H|} \geq ((1-1/4^k)^{4^k})^{k} \geq (1/4)^k$$
  since $(1-1/n)^n$ is increasing with $n$ and we have $k \geq 1$. 
  Thus, $(S,F)$ covers $H$ with probability at least
  $1/2^{\bigoh(k^2)}$.
\end{proof}

\subsection{Derandomization}
\label{sec:balance:derand}

For the derandomization, we note that the only part of the sampling
procedure that uses randomness is Step~2 in
$\RandomCover(G, \BB, k)$ (Fig.~\ref{fig:sampling})
where we sample a set $\cX' \subseteq \cX$ of important connected subsets.
Furthermore, by the proof of Theorem~\ref{ithm:balanced-shadow-sampling},
for any important balanced subgraph $H$ of cost at most $k$ there are
two disjoint families $\cA_H, \cB_H \subseteq \cX$ with $|\cA_H| \leq k$
and $|\cB_H| \leq k4^k$ such that the output $(S,F)$ of $\RandomCover(G, \BB, k)$
covers $H$ if and only if $\cA_H \subseteq \cX'$ and
$\cB_H \cap \cX' = \emptyset$.

We are looking to replace the sampling $\cX' \subseteq \cX$
by the iteration of $\cX'$ over all members of a set family $\cF \subseteq 2^{\cX}$ 
with the guarantee that one member of $\cF$ meets the requirements.
Thus we are, in derandomization terms, looking for an ``extremely
lopsided'' universal set. Fortunately, the literature does contain the
construction of such a family for our parameters, known as an 
\emph{$(n,(r,s))$-cover-free family} (with suitable parameters $n$, $r$, $s$).
We recall the definition from Bshouty and Gabizon~\cite{BshoutyG17cff}, rephrased
in an equivalent form.

\begin{definition}[Bshouty and Gabizon~\cite{BshoutyG17cff}]
  Let $r, s, n \in \NN$ with $r, s < n$.
  An \emph{$(n,(r,s))$-cover free family} ($(n,(r,s))$-CFF) is a set
  $\cF \subseteq \{0,1\}^n$
  such that for every disjoint pair of sets $A, B \subseteq [n]$
  with $|A|=r$ and $|B|=s$, there is a member $F \in \cF$
  such that $F_i=1$ for every $i \in A$ and $F_i=0$ for every $i \in B$.
\end{definition}

Bshouty and Gabizon showed a near-optimal, near-linear time
construction of an $(n,(r,s))$-CFF for any parameters $n, s, r$,
in the following sense. There is an information-theoretic lower bound
that states that there is a function $N(r,s)$ such that any
$(n,(r,s))$-CFF requires size $\Omega(N(r,s) \log n)$,
and they gave an algorithm that constructs an $(n(r,s))$-CFF
of size $N(r,s)^{1+o(1)}\log n$ in time $N(r,s)^{1+o(1)}n \log n$.
For the details of the function $N(r,s)$, see~\cite{BshoutyG17cff}.
We focus, as above, on the setting $r=k$ and $s=k4^k$. 

\begin{theorem} \label{thm:derandom_cover}
  Let $(G,\BB)$ be a biased graph and $k \in \NN$ a parameter.
  Let $n=|V(G)|$. There is a deterministic algorithm that produces
  a list of $O(4^{k^2(1+o(1))}\log n)$ pairs $(S_i,F_i)$
  such that for any balanced (not necessarily connected) subgraph $H$ 
  of cost at most $k$, there is a pair $(S_i,F_i)$ in the list
  which covers $H$. The list can be computed in time $O^*(4^{k^2(1+o(1))})$.
\end{theorem}
\begin{proof}
  As discussed above, it suffices to replace the random sampling of a
  set $\cX' \subseteq \cX'$ in $\RandomCover(G, \BB, k)$
  by iterating $\cX'$ over an $(|\cX|, (k, k4^k))$-CFF.
  For this setting, the value of $N(r,s)$ is
  \[
    \left(\frac{r+s}{r}\right)^{r+o(r)} = 4^{k^2+o(k^2)},
  \]
  hence we get a family of size $4^{k^2(1+o(1))}\log n$.
  We produce on output set $(S_i,F_i)$ for every choice
  $\cX_i \in \cF$ where $\cF$ is the $(|\cX|,(k,k4^k))$-CFF
  of Bshouty and Gabizon~\cite{BshoutyG17cff}. Note that in this
  setting we have $n \leq n_G4^k$ where $n_G=|V(G)|$,
  hence the size overhead of $\log n=O(k \log n_G)$ 
  and the running time overhead of $n=O^*(4^k)$
  are negligible compared to the total output size.
\end{proof}

\subsection{Connection to Random Sampling of Important Separators}
\label{sec:balance:impsep}

Theorem~\ref{ithm:balanced-shadow-sampling} is a direct generalization of (the edge deletion case of) the shadow covering result of Marx and Razgon~\cite{marx2014fixed}, generalized to the setting of arbitrary balanced subgraphs. 
To highlight this connection, we 
recall the simplification of the theorem that we presented
in Section~\ref{sec:intro-coverings}.

\corSampleImportant*

\begin{proof}
  Let $H$ be an important balanced subgraph of some cost at most $k$,
  and let $\cK$ be the partition of $V(K)$ according to the connected
  components of $H$. By Lemma~\ref{lem:important-components}, 
  every component $K \in \cK$ is an important subset 
  and the sets $\delta_G(K)$ over $K \in \cK$ are pairwise disjoint.
  Thus we can follow the proof of Theorem~\ref{ithm:balanced-shadow-sampling}
  and define the two sets $\cA_H=\cK$ and $\cB_H$ consisting of all
  important subsets of cost at most $k$ that intersect $N_G(V(H))$.
  Let $\cX'$ be the set of important subsets sampled in the execution
  of $\RandomCover(G, \BB, k)$ and consider the event that $\cA_H \subseteq \cX'$
  and $\cB_H \cap \cX' = \emptyset$. In this case, we get $V(H) \subseteq S$
  and $N_G(V(H)) \cap S = \emptyset$. Furthermore, for every connected
  component $K \in \cK$ the cleaning set $F_K$ for $K$ produced by
  $\RandomCover(G, \BB, k)$ has minimum cost; and since $H$ is important,
  there are precisely $|F_K|$ edges of $G[K]$ absent from $H$ for
  every $K \in \cK$. Thus both $G[S]-F$ and $H$ contain connected components
  on vertex set $K$ and with precisely $\cost(K)$ edges of $G$ with
  endpoints in $K$ absent from the respective components. Thus they
  are domination equivalent. The final statement follows immediately.
\end{proof}

Let us now take a closer look at the shadow covering result of Marx and Razgon~\cite{marx2014fixed}.
We note some more terminology; note that the following is written for vertex cuts. Let $G$ be a graph and $W \subseteq V(G)$ a set of terminals. For $S \subseteq V(G)$, the \emph{shadow of $S$} is the set of vertices not reachable from any vertex of $W$ in $G-S$. 
A set $U \subseteq V(G)$ with $W \subseteq U$ is \emph{$W$-closest} 
if there is no set $U' \subset U$ with $W \subseteq U'$ such that $|N(U')| \leq |N(U)|$. 
Note that this is the opposite direction from how important separators are defined;
indeed, Marx and Razgon show that if $U$ is a $W$-closest set, then for every vertex $v \in V(G) \setminus N[U]$, $N(U)$ contains an important $(v,W$)-vertex separator~\cite{marx2014fixed}[Lemma~3.10]. In this scenario, $V(G) \setminus N[U]$ is the shadow of $N(U)$. 

Their formulation of the result that came to be known as shadow covering is then the following. 

\begin{theorem}[{Shadow covering~\cite[Theorem~3.6]{marx2014fixed}}] \label{thm:marxrazgon}
  There is a randomized algorithm that, given a graph $G$, a set $W \subseteq V(G)$ and an integer $k$, produces a set $Z \subseteq V(G) \setminus W$ such that the following holds.
  For every $W$-closest set $U$ with $|N(U)| \leq k$ with probability at least $1/2^{O(k^3)}$ the following hold.
  \begin{enumerate}
      \item $N(U) \cap Z = \emptyset$, and
      \item $V(G) \setminus N[U] \subseteq Z$.
  \end{enumerate}
\end{theorem}

Again, let $U$ be a $W$-closest set and assume that $Z$ meets the conditions in the theorem.
Then $S:=N(U)$ acts as a separator and $V(G) \setminus N[U]$ is the shadow of $N(U)$.
The set $Z$ then contains the shadow of $S$ and is disjoint from $S$ itself. 
They also show a derandomized version, and follow-up work has generalized the notion to directed graphs and improved the success probability to $1/2^{O(k^2)}$~\cite{chitnis2015directed} and subsequently, very recently (for undirected graphs) to $1/2^{O(k \log k)}$~\cite{chu2026faster}.
We focus on the above basic form here.

To connect this to important balanced subgraphs, recall the construction discussed in Section~\ref{sec:balance:important}. Let $G$ be a graph and $W \subseteq V(G)$ a set of terminals. Attach a sufficiently large number of unbalanced cycles to every vertex of $W$ and define all other cycles of the graph to be balanced. 
Let $(G,\BB)$ be the resulting biased graph. Then a balanced subgraph $H$ of $G$ of cost at most $k$ is any subgraph $H$ (of cost at most $k$) such that $W \cap V(H)=\emptyset$. 

To state the connection precisely, let us redefine $W$-closest sets to refer to edge cuts, so that a set $U \subseteq V(G)$ with $W \subseteq U$ is $W$-closest if for any set $U' \subset U$ with $W \subseteq U'$ we have $|\delta_G(U')| > |\delta_G(U)|$. 
Then it is easy to see that $U$ is a $W$-closest set if and only if $V(G) \setminus U$ is an important vertex subset for $(G,\BB)$, i.e.\ $H=G-U$ is an important balanced subgraph (of cost $|\delta_G(U)$). Thus Corollary~\ref{icor:sample-important} implies Theorem~\ref{thm:marxrazgon} for the case of edge cuts. 

We note in passing that collections of unbalanced cycles for biased graphs are closed under union; i.e.\ if $(G,\BB_1)$ and $(G,\BB_2)$ are biased graphs then so is $(G,\BB_1 \cap \BB_2)$.
Thus, the above construction can if desired be combined with other biased graph constructions, resulting in balanced subgraphs being defined as both being disjoint from $W$ and balanced in a second biased graph sense. 

However, we briefly note a difference in application of Theorem~\ref{thm:marxrazgon} and Theorem~\ref{ithm:balanced-shadow-sampling}. Marx and Razgon use Theorem~\ref{thm:marxrazgon} to solve \textsc{Multicut} parameterized by the solution size, by applying a greedy reduction to the shadowless case; i.e.\ they show, at a certain stage of their algorithm, that the solution $S$ they are looking for can be assumed to have the property that $V(G)$ minus $S$ and the shadow of $S$ is a $W$-closest set (for a known set of terminals $W$). They then apply a reduction rule to simplify the graph by removing the shadow of $S$ from $G$, assuming that $Z$ in Theorem~\ref{thm:marxrazgon} has been guessed correctly.
This greedy aspect has been highly successful for many other FPT graph separation algorithms
(see e.g.~\cite{ChitnisHM13dmwc,chitnis2015directed,HatzelJLMPSS23}),
but as has been noted, it fails to apply in certain situations, e.g.\ in a setting of weighted graph cuts. Indeed, \textsc{Weighted Multicut} was shown to be FPT only more recently, via an application of the \emph{flow augmentation} procedure~\cite{KimPSW22weighted}. 
A particularly interesting case is \textsc{Directed Multiway Cut} which, along with a few related problems, is FPT in the unweighted setting, using shadow covering techniques, 
but is W[1]-hard in the weighted case~\cite{ChitnisHM13dmwc,HatzelJLMPSS23}.

In our algorithm for $\minlin{2}{\ZZ_{p^d}}$, our main application builds on Theorem~\ref{ithm:balanced-shadow-sampling} rather than Corollary~\ref{icor:sample-important}, and our main interest is the cleaning set $F$ rather than the exact balanced subgraph $G_B$ itself. 
In particular, we use no greedy ``shadow bypass'' step, and our algorithm could in principle be extended to the weighted case. We omit this, but see discussion in Section~\ref{sec:discussion}.

\section{Algorithm for $\minlinz{2}{m}$}
\label{sec:upperbounds}

Recall that $\omega(m)$ is the number of distinct prime factors of $m$.
Our goal in this section is to prove the
upper bound from Theorem~\ref{thm:main}
as summarized below.

\begin{restatable}{theorem}{thmmainalgo}  \label{thm:algorithm-main}
  $\minlin{2}{\ZZ_m}$ is FPT-approximable within $\omega(m)$
  for every integer $m \geq 2$.
\end{restatable}

Let $p_1^{n_1} \cdots p_\ell^{n_\ell}$ be the prime factorization of $m$ and
note that $\omega(m) = \ell$.
We know that $\ZZ_m$ is isomorphic to the direct sum
$\bigoplus_{i=1}^{\ell} \ZZ_{p_i^{n_i}}$ by 
Proposition~\ref{prop:crt}.
The following result implies that we can focus 
on the prime power case
when analyzing \FPT-approximability of $\minlin{2}{\ZZ_m}$.

\begin{proposition}[Proposition~4~in~\cite{Dabrowski:etal:esa2025}]
  \label{prop:product-approx}
  Suppose that the ring $R$ is 
  isomorphic to a direct sum $\bigoplus _{i=1}^{\ell} R_{i}$.
  If $\minlin{2}{R_{i}}$ is \FPT-approximable within a factor $c_i$
  for all $i \in [\ell]$, then $\minlin{2}{R}$ is \FPT-approximable
  within the factor $\sum_{i=1}^\ell c_i$.
\end{proposition}

Thus, it suffices to show that
$\minlin{2}{\ZZ_{p^d}}$ is FPT-approximable
within the factor $1$, or equivalently,
to prove Theorem~\ref{thm:zpd_is_fpt} restated below.

\thmZpd*

We prove the theorem in three steps. 
First, we show that we only need to prove
the theorem for \emph{special instances}
(Section~\ref{ssec:special}).
Then, given a special instance $(I,k)$
of $\minlin{2}{\ZZ_{p^d}}$, we construct 
an instance $(\bool{I},k)$ of
\MinSat{} (Section~\ref{sec:mincsp-reduction}). 
The instance $(\bool{I},k)$ is a relaxation of $(I,k)$: 
if $(I, k)$ is a yes-instance, then $(\bool{I}, k)$
is a yes-instance, 
but the converse does not necessarily hold.
To overcome this,
in the third step (Section~\ref{sec:the-algorithm}),
we augment $\bool{I}$ with additional constraints 
using the machinery for balanced subgraph covering from 
Section~\ref{sec:balanced-subgraph-coverings}.
This ensures that the resulting instance is yes/no-equivalent to $(I,k)$ with sufficient probability. 
The full algorithm is depicted in Figure~\ref{fig:mainalgorithm}.

\subsection{Reduction to Special Instances}
\label{ssec:special}

An instance of $\minlin{2}{R}$ is \emph{simple}
if all binary equations are of the form
$a \cdot u = v$ and all unary
equations are crisp and of the form 
$v = b$ for some $a,b \in R$. 
By using iterative compression, homogenization and branching in FPT time, 
Dabrowski~et~al.\ have shown the following.

\begin{proposition}[See Lemma~5~in~\cite{Dabrowski:etal:esa2025}]
  \label{prop:simple}
  If $\minlin{2}{\ZZ_{p^d}}$ restricted to simple instances
  is in \FPT, then $\minlin{2}{\ZZ_{p^d}}$ is in \FPT.
\end{proposition}
We remark that the reduction underlying Proposition~\ref{prop:simple}
is a Turing reduction: it uses iterative compression,
which makes $\bigoh^*(2^{O(k)})$ calls to the algorithm for $\minlin{2}{\ZZ_{p^d}}$
on simple instances.
To simplify the proofs and the presentation, 
it is more convenient for us to work with
\emph{special} instances,
which only allow the following equations
\begin{itemize}
   \item a single crisp equation $v = 1$,
   \item binary equations $r \cdot u = v$
     for units $r$ in $\ZZ_{p^d}$, and
   \item binary equations $p \cdot u = v$.
\end{itemize}

\begin{lemma}
\label{lem:special-instances}
  If $\minlin{2}{\ZZ_{p^d}}$ restricted to special instances
  is in \FPT, then $\minlin{2}{\ZZ_{p^d}}$ is in \FPT.
\end{lemma}
\begin{proof}
  Let $(I, k)$ be an instance of $\minlin{2}{\ZZ_{p^d}}$.
  By Proposition~\ref{prop:simple}, we may assume
  that it is simple.
  To construct a special instance $(I', k)$
  with the same parameter,
  we let $V(I') = V(I) \cup \{s\}$,
  where $s$ is a new variable,
  and add a crisp equation $s = 1$.
  Then, for every crisp unary equation
  $v = b$ in $I$, we add a crisp
  binary equation $v = b \cdot s$.
  Now, consider a soft binary equation $C$ of 
  the form $a \cdot u = v$ and
  observe that $a = r \cdot p^\ell$ for some unit $r$
  and $0 \leq \ell < d$.
  Add auxiliary variables $z_1,\dots,z_{\ell}$ to $I'$
  with crisp equations
  \[
    p \cdot u = z_1, \
    p \cdot z_1 = z_2, \
    \dots, \
    p \cdot z_{\ell-1} = z_\ell
  \]
  and a soft equation
  \[ 
    v = r \cdot z_\ell.
  \]
  Note that these equations together imply $C$.
  Moreover, if the last equation is deleted,
  then there is no equation implies on $u$ and $v$.
  Hence, $(I, k)$ is a yes-instance
  if and only if $(I', k)$ is a yes-instance.
\end{proof}

The computations in Lemma~\ref{lem:special-instances}  
can obviously be performed in polynomial time, 
so at the cost of an additional
$\bigoh^*(2^{O(k)})$ factor in the running time, 
we may reduce the $\minlin{2}{\ZZ_{p^d}}$ 
to solving special instances.

\subsection{Relaxation}
\label{sec:mincsp-reduction}

In this section, we show how special instances
of $\minlin{2}{\ZZ_{p^d}}$
can be relaxed with the aid of the \MinSat{} problem.
We begin by describing how certain subsets of domain elements can be represented in a 
convenient way (Section~\ref{ssec:cosets}), and
then use this representation as the basis for our transformation into
\MinSat{} (Section~\ref{ssec:boolean-transformation}).

\subsubsection{Cosets and Base-$p$ Representations}
\label{ssec:cosets}

An \emph{ideal} of a ring $(R, +, \cdot)$ is a subset $I \subseteq R$ of the elements
closed under addition and absorbing multiplication, i.e.\
$(I,+)$ is a subgroup of $(R,+)$ and 
for every $r \in R$ and every 
$x \in I$, the product $r \cdot x$ is in $I$. 
The \emph{cosets} of an ideal $I$ are the sets
$r + I = \{r + x \mid x \in I\}$ for every $r \in R$.
Note that if $r \in I$, then $r + I = I$,
and a coset $r + I$ is called \emph{proper} if $r \notin I$.
For example, in $\ZZ_{27}$ the elements divisible
by $3$ form an ideal, sometimes denoted by $(3)$,
and the proper cosets are $1 + (3)$ and $2 + (3)$.
Note that the cosets of an ideal form a partition of the ring elements.

We introduce a way of describing certain cosets of
the elements in $\ZZ_{p^d}$.
Every $a \in \ZZ_{p^d}$ admits a base-$p$ representation, which is
a tuple $(a_{d-1}, \dots, a_0)$ 
where each $a_j \in \ZZ_p$ is called a \emph{digit}:
formally, $a = \sum_{j=0}^{d-1} a_j \cdot p^j$.
The \emph{order} of $a$, denoted by $\ord(a)$, is the number
of trailing zeros in its base-$p$ representation;
equivalently, it is the largest integer $i$ such that
$a$ is divisible by $p^i$.
The digits $a_{d-1}, \dots, a_{\ord(i)}$ are called \emph{significant},
ordered so that more significant digits have larger subscript.
To extract $j$ significant digits of $a$,
we use the notation $\sind{a}{j}$ for $a \in \ZZ_{p^d}$ and $j \geq 1$
which equals $(a / p^{\ord(a)}) \bmod p^j$.

In the course of the algorithm we encounter variables
whose exact value is unknown but the constraints 
determine the order and some but not all significant digits.
For example, consider an equation $3v = 18 \mod 27$.
It implies that $v$ may take values $6$, $15$ or $24$.
The base-$3$ for this set of values may be described as 
$(*, 2, 0)$, meaning that the order must be $1$, 
the least significant digit must be $2$, and 
the next significant digit denoted by $*$ may be arbitrary.
More generally, a coset 
denoted by $\coset{s}{\ell}{i}$ corresponds to a tuple
\[
  \underbrace{* \ \cdots \ *}_{j} \
  \underbrace{s_\ell \ \cdots \ s_1}_{\ell} \
  \underbrace{0 \ \cdots \ 0}_{i},
\]
where $i + \ell + j = d$, 
$\ell \geq 1$, $s_1 \neq 0$ and
$s = (s_\ell,\ldots,s_1)$ is called the \emph{suffix} of the coset.
Formally, $\coset{s}{\ell}{i} = \{a \in \ZZ_{p^d} : \ord(a) = i \text{ and } \sind{a}{\ell} = s \}$.

\subsubsection{Transformation into \MinSat}
\label{ssec:boolean-transformation}

Let $(I, k)$ be a special instance of $\minlin{2}{\ZZ_{p^d}}$.
We construct an instance $\bool{I}$ of $\MinSat(\bA_{p^d})$ 
for a finite Boolean, bijunctive and $2K_2$-free constraint language $\bA_{p^d}$.
The language $\bA_{p^d}$ is finite and only depends on $\ZZ_{p^d}$.
We do not explicitly list the relations but rather define them by conjunctions of $1$- and $2$-clauses
and showing that their Gaifman graphs are $2K_2$-free.

For the constructed instance $\bool{I}$,
the following holds:
if $(I, k)$ is a yes-instance,
then $(\bool{I},k)$ is a yes-instance.
We stress that the converse is not necessarily true:
see the example in \eqref{eq:triangle_on_a_stick} in the introduction
and the discussion below.
Thus, $(\bool{I},k)$ is a relaxation of $(I, k)$.
The instance $\bool{I}$ will be constructed in polynomial time,
and since $\bA_{p^d}$ is finite and is $2K_2$-free,
we can decide $(\bool{I}, k)$ is FPT time by Theorem~\ref{thm:2k2-free}.

Let us start with a high-level overview
of the construction.
For every variable $v \in V(I)$,
we will introduce a set $\bool{V}(v)$ of Boolean
variables in $\bool{I}$ and add a set
of crisp constraints $\bool{K}(v)$ on these variables.
The Gaifman graph of $\bool{K}(v)$ will form a clique.
Then, for every equation $C$ on variables $u,v$ in $I$,
we will introduce a Boolean constraint
$\bool{C}$ on variable set $\bool{V}(u) \cup \bool{V}(v)$,
and show that it $2K_2$-free.
We present the steps in detail below.

\paragraph{Variables in $I$.}

For every variable $v \in V(I)$, create a set $\bool{V}(v)$ of
Boolean variables $\bvar{v}{s}{\ell}{i}$
for all integers $i \geq 0$ and $\ell \geq 1$ 
such that $i + \ell \leq d$,
and all units $s \in \ZZ_{p^\ell}$.
The variable $\bvar{v}{s}{\ell}{i}$
indicates whether the value assigned to $v$ belongs 
to $\coset{s}{\ell}{i}$.

Boolean variables $\bvar{v}{s}{\ell}{i}$ with $\ell + i = d$ 
are called \emph{singleton variables} as
$\coset{s}{\ell}{i}$ consists of a single element,
while variables with $\ell + i < d$ are called \emph{coset variables}.
For conciseness,
we sometimes write $\bvarsingle{v}{a}$ instead of $\bvar{v}{s}{\ell}{d-\ell}$
for the singletons, where $a = s \cdot p^{d-\ell}$.
Note that there is no Boolean variable indicating that $v$
is assigned value $0$: instead, we treat setting 
all Boolean variables in $\bool{V}(v)$ to false 
as setting $v$ to $0$.

Every pair of cosets we defined are either
disjoint or one is a superset of the other, i.e. the cosets
form a laminar set family.
To encode this, we create a set of crisp constraints 
$\bool{K}(v)$ for every pair of variables in $\bool{V}(v)$.
\begin{itemize}
  \item If $\coset{a}{\ell}{i}$ and $\coset{a'}{\ell'}{i'}$ are disjoint, 
  then we add to $\bool{K}(v)$ a negative clause
  $(\lnot \bvar{v}{a}{\ell}{i} \lor \lnot \bvar{v}{a'}{\ell'}{i'})$.
  \item If $\coset{a}{\ell}{i} \subsetneq \coset{a'}{\ell'}{i'}$, 
  then we add to $\bool{K}(v)$ an implication
  $(\bvar{v}{a}{\ell}{i} \to \bvar{v}{a'}{\ell'}{i'})$.
  Observe that for $\coset{a}{\ell}{i} \subsetneq \coset{a'}{\ell'}{i'}$
  to hold, we must have $i = i'$, $\ell > \ell'$ and
  $a \bmod p^{\ell - \ell'} = a'$.
\end{itemize}
The Gaifman graph of $\bool{K}(v)$ is complete, i.e.\ a clique.

\begin{example}
  Consider $\ZZ_9$, i.e. $p=3$ and $d=2$.
  For a variable $v \in V(I)$, the set $\bool{V}(v)$ contains 
  the following Boolean variables:
  \begin{itemize}
    \item singleton $\bvarsingle{v}{a}$ for $a \in \{1,2,3,4,5,6,7,8\}$, and
    \item coset variables $\bvar{v}{1}{1}{0}$ and $\bvar{v}{2}{1}{0}$
    corresponding to $\{1,4,7\}$ and $\{2,5,8\}$, respectively.
  \end{itemize}
  The constraints in $\bool{K}(v)$ are
  \begin{itemize}
    \item $(\bvarsingle{v}{a} \to \bvar{v}{1}{1}{0})$ for $a \in \{1,4,7\}$,
    \item $(\bvarsingle{v}{b} \to \bvar{v}{2}{1}{0})$ for $b \in \{2,5,8\}$, and
    \item negative clauses between all other pairs of variables.
  \end{itemize}
\end{example}

\paragraph{Constraints in $I$.}

The instance $I$ is special so we only have to consider three
cases.

\begin{enumerate}

\item
For the crisp equation $(v = 1)$ in $I$, 
we add a crisp constraint 
$\bvarsingle{v}{1}$ to $\bool{I}$.

\item
For every binary equation $C = (r \cdot u = v)$ in $I$
where $r$ is a unit,
we create a constraint $\bool{C}$ in $\bool{I}$ on 
variables $\bool{V}(u) \cup \bool{V}(v)$
defined as a conjunction all constraints in $\bool{K}(u)$ and $\bool{K}(v)$, and the constraints
\[ 
  (\bvar{u}{a}{\ell}{i} \leftrightarrow \bvar{v}{b}{\ell}{i})
\]
for all $i$ and $\ell$ such that 
$i + \ell \leq d$ and
$b = a \cdot r \bmod p^{\ell}$.
Observe that the Gaifman graph of $\bool{C}$ is $2K_2$-free 
because $\bool{K}(u)$ and $\bool{K}(v)$ form cliques.

\item
For every binary equation $C = (v = p \cdot u)$ in $I$,
we create a constraint $\bool{C}$ in $\bool{I}$ on 
variables $\bool{V}(u) \cup \bool{V}(v)$
defined as a conjunction of all constraints in $\bool{K}(u)$ and $\bool{K}(v)$, the constraints
\[
  (\lnot \bvar{v}{a}{\ell}{0}) 
\]
and
\[
  (\bvar{u}{a}{\ell}{i} \leftrightarrow \bvar{v}{a}{\ell}{i+1})
\]
for all $i$ and $\ell$ such that $i + \ell + 1 \leq d$ and 
all units $a \in \ZZ_{p^\ell}$.
Observe that the Gaifman graph of $\bool{C}$ is $2K_2$-free
because $\bool{K}(u)$ and $\bool{K}(v)$ form cliques.
\end{enumerate}

\begin{example}
  Consider the equation $C = (2u = v)$ over $\ZZ_9$.
  The constraints in $\bool{C}$, apart from 
  those in $\bool{K}(u)$ and $\bool{K}(v)$, are
  \begin{itemize}
    \item $(\bvarsingle{u}{1} \leftrightarrow \bvarsingle{v}{2})$,
    \quad $(\bvarsingle{u}{2} \leftrightarrow \bvarsingle{v}{4})$,
    \quad $(\bvarsingle{u}{3} \leftrightarrow \bvarsingle{v}{6})$,
    \quad $(\bvarsingle{u}{4} \leftrightarrow \bvarsingle{v}{8})$,
    \item[] $(\bvarsingle{u}{5} \leftrightarrow \bvarsingle{v}{1})$,
    \quad $(\bvarsingle{u}{6} \leftrightarrow \bvarsingle{v}{3})$,
    \quad $(\bvarsingle{u}{7} \leftrightarrow \bvarsingle{v}{5})$,
    \quad $(\bvarsingle{u}{8} \leftrightarrow \bvarsingle{v}{7})$,
    \item $(\bvar{u}{1}{1}{0} \leftrightarrow \bvar{v}{2}{1}{0})$, and
    \item $(\bvar{u}{2}{1}{0} \leftrightarrow \bvar{v}{1}{1}{0})$.
  \end{itemize}
\end{example}

\medskip

\noindent
We will now take a closer look at the connection
between the assignments to $I$ and the Boolean
assignments to $\bool{I}$.
Given an assignment $\alpha : V(I) \to \ZZ_{p^d}$, we define 
$\bool{\alpha} : V(\bool{I}) \to \{0, 1\}$ as follows.
\begin{itemize}
  \item For $v \in V(I)$ such that $\alpha(v) = 0$,
    let $\bool{\alpha}(\mathbf{x}) = 0$ for all $\mathbf{x} \in \bool{V}(v)$.
  \item For $v \in V(I)$ such that $\alpha(v) = r \neq 0$,
    let $\bool{\alpha}(\bvarsingle{v}{r}) = 1$, and extend $\bool{\alpha}$
    to the rest of $\bool{V}(v)$ so that all constraints in $\bool{K}(v)$
    are satisfied.
    Note that the extension is unique:
    only the variables with an implication 
    from $\bvarsingle{v}{r}$ in $\bool{K}(v)$ are set to $1$,
    and all the remaining variables are set to $0$.
\end{itemize}
Observe that $\bool{\alpha}$ obtained this way 
sets exactly one singleton variable in $\bool{V}(v)$ to $1$.
However, $\bool{K}(v)$ also admits satisfying assignments
that maps all singleton variables in $\bool{K}(v)$ to $0$
and maps a coset variable to $1$.
We call such Boolean assignments \emph{ambiguous on $v$}.
For example, consider $\ZZ_4$, i.e.\ $p=2$ and $d=2$,
and an assignment $\beta$ such that 
$\beta(\bvarsingle{v}{1}) = \beta(\bvarsingle{v}{2}) = \beta(\bvarsingle{v}{3}) = 0$ and
$\beta(\bvar{v}{1}{1}{0}) = 1$.
Note that $\coset{1}{1}{0} = \{1,3\}$, so,
intuitively, both assignments 
$v \mapsto 1$ and $v \mapsto 3$ 
are compatible with $\beta$.
If a Boolean assignment is ambiguous on some variable,
we call it \emph{ambiguous}, and \emph{unambiguous} otherwise.

\begin{lemma} \label{lem:unambiguous_Boolean_to_modpn}
  Let $(I, k)$ be a special instance of $\minlin{2}{\ZZ_{p^d}}$.
  If $\bool{I}$ admits an unambiguous assignment of 
  cost $k$, then $(I, k)$ is a yes-instance.
\end{lemma}
\begin{proof}
  Let $\beta$ be an unambiguous Boolean assignment to $\bool{I}$.
  Define $\alpha : V(I) \to \ZZ_{p^d}$ as
  \[
    \alpha(v) = \begin{cases}
      a & \text{if } \beta(\bvarsingle{v}{a}) = 1 \text{ for some } a \in \ZZ_{p^d} \\ 
      0 & \text{otherwise}.
    \end{cases}
  \]
  Note that $\alpha$ is well defined because $\beta$ has finite cost,
  so it satisfies all crisp constraints which disallow setting
  two singleton Boolean variables for the same $v$ to $1$.
  Moreover, observe that since $\beta$ is unambiguous,
  for any $v \in V(I)$,
  if there is a variable in $\bool{V}(v)$
  set to $1$, then there must be a singleton variable
  set to $1$ as well.
  
  Consider an equation $C$ in $I$.
  We claim that if $\beta$ satisfies $\bool{C}$,
  then $\alpha$ satisfies $C$.
  Note that this implies that the cost of $\alpha$ in $I$
  is at most $k$, and thus is sufficient to prove the lemma.
  Since $I$ is special, it is sufficient
  to consider the following three cases.

 \begin{enumerate}

\item
  $C$ is the crisp equation $(v = 1)$. Then 
  $\beta(\bvarsingle{v}{1}) = 1$ because $\beta$ has finite cost 
  and $\alpha(v) = 1$.

  \item
  $C$ is of the form $r \cdot u = v$ for some unit $r$.
  By the definition of $\bool{C}$, the variables
  in $\bool{V}(u)$ and $\bool{V}(v)$ are in one-to-one
  correspondence formed by the $\leftrightarrow$ constraints.
  We claim that if $\alpha(u) = 0$ then $\alpha(v) = 0$:
  indeed, if $\beta$ assigns $0$ to all variables in $\bool{V}(u)$,
  so it also assigns all variables in $\bool{V}(v)$ to $0$.
  Otherwise, $u$ takes a nonzero value under $\alpha$,
  which can be expressed as 
  $\alpha(u) = a \cdot p^i$ for some unit $a \in \ZZ_{p^{d-i}}$.
  Then $\beta(\bvar{u}{a}{d-i}{i}) = 1$ 
  and $\beta(\bvar{v}{b}{d-i}{i}) = 1$
  for $b = r \cdot a \bmod p^{d-i}$ by the constraints in $\bool{C}$.
  Hence, $\alpha(v) = b \cdot p^i = r \cdot a \cdot p^i = r \cdot \alpha(u)$.

\item
  $C$ is of the form $p \cdot u = v$.
  By the definition of $\bool{C}$, for every unit $a \in \ZZ_{p^d}$,
  we have $\beta(\bvarsingle{v}{a}) = 0$;
  moreover, for every nonzero and nonunit value in $\ZZ_{p^d}$,
  which can be expressed as $a \cdot p^i$ for $0 < i < d$ and unit $a \in \ZZ_{p^{d-i}}$,
  we have $\beta(\bvarsingle{u}{a \cdot p^{i-1}}) = \beta(\bvarsingle{v}{a \cdot p^i})$.
  In particular, if $\alpha(u) = 0$, we have $\alpha(v) = 0$
  because $\beta$ is not allowed to set any variable to $1$ in $\bool{V}(v)$.
  Moreover, if $\alpha(u) = a \cdot p^{i-1}$ for
  some $0 < i < d$ and unit $a \in \ZZ_{p^{d-i}}$,
  then $\alpha(v) = a \cdot p^{i}$.
  Finally,
  if $\alpha(u) = a \cdot p^{d-1}$, then $\beta$ is not allowed
  to set any variable in $\bool{V}(u)$ to $1$,
  so $\alpha(v) = 0$.
  In all three cases, $\alpha$ satisfies $C$. \qedhere
  \end{enumerate}
\end{proof}

\subsection{Disambiguation}
\label{sec:the-algorithm}

Let $(I, k)$ be a special instance of $\minlinz{2}{p^d}$.
Lemma~\ref{lem:unambiguous_Boolean_to_modpn} shows that
  if $\bool{I}$ has an unambiguous assignment of 
  cost $k$, then $(I, k)$ is a yes-instance.
Thus, we need a method for disambiguating 
assignments to $\bool{I}$ and
this is where balanced subgraph covers come into play.
In order to apply the balanced subgraph cover machinery,
we define auxiliary graphs $(G_i, \BB_i)$ for all $0 \leq i \leq d-2$ as follows.
Let $\Gamma_i$ be the group with elements being the units in $\ZZ_{p^{d-i}}$
and the group operation being multiplication modulo $p^{d-i}$.
The vertex set of $G_i$ consists of vertices 
$v^0, \dots, v^i$ for every variable $v$ in $I$.
We refer to them as the \emph{copies of $v$ (in $G_i$)}.
The edge set of $G_i$ and the labelling 
$\gamma_i : E(G) \times V(G_i)^2 \to \Gamma_i$ are defined as follows.
\begin{itemize}
  \item 
  For every equation in $I$ of the form $r \cdot u = v$, 
  where $r$ is a unit in $\ZZ_{p^d}$,
  add an edge $e = u^jv^j$ for all $0 \leq j \leq i$
  with $\gamma_i(e, u^j, v^j) = r \bmod p^{d-i}$.
  \item 
  For every equation in $I$ of the form $p \cdot u = v$,
  add an edge $e=u^jv^{j+1}$ for all $0 \leq j \leq i-1$
  with $\gamma_i(e, u^j, v^{j+1}) = 1$.
\end{itemize}
By Lemma~\ref{lem:minlin-fin-to-rbgce},
$(G_i, \BB_i)$ is a biased graphs for all $i$
and each $\BB_i$ admits a polynomial-time oracle.

\smallskip

\begin{figure}[tb]
\fbox{\parbox{\textwidth}{

Algorithm: $\SolveMinLin(I, k)$

\begin{enumerate}

  \item \label{alg:make_special}
  Convert $I$ into special form according to Lemma~\ref{lem:special-instances}.

  \item \label{alg:encode_bool}
  Initialize the \MinSat{} 
  instance $\enhance{I} \gets \bool{I}.$
  
  \item \label{alg:forloop_1}
  For all $i \in \range{0}{d-2}$, do the following.
  \begin{enumerate}
    \item \label{alg:aux_graph}
    Construct auxiliary graphs $G_i \gets G_i(I)$.
    \item \label{alg:random_cover}
    Compute $(S'_i, F'_i) \gets \RandomCover(G_i, \BB_i, 2k)$.
  \end{enumerate}
  
  \item \label{alg:forloop_2}
  For all $i \in \range{0}{d-2}$, do the following.
  \begin{enumerate}
  
    \item \label{alg:compute_si}
      Let $S_i = \bigcap_{j=i}^{d-2} \{ v \in V(I) : v^i \in S'_j \}$.

    \item \label{alg:compute_ti}
      Let $T_i = \bigcup_{j=i}^{d-2} \{ v \in S_i : v^i \in V(F'_j) \}$.      
    
    \item \label{alg:reduce_domain}
      For all $v \notin S_i$, add crisp constraints
      $(\lnot \bvar{v}{a}{1}{i})$ to $\enhance{I}$ 
      for every unit $a \in \ZZ_p$.

    \item \label{alg:add_implications}
      For all $v \in T_i$, pick a unit 
      $s_v \in \ZZ_{p^{d-i}}$ uniformly at random,
      and let $g_v = s_v \cdot p^{i}$. \\
      For all $\ell \geq 1$ such that $i + \ell \leq d$,
      add the following crisp constraints to $\enhance{I}$:
      \begin{itemize}
        \item $(\bvar{v}{a}{\ell}{i} \to \bvarsingle{v}{g_v})$
          for all units $a \in \ZZ_{p^\ell}$ such that $a = s_v \bmod p^\ell$, and
        \item $(\lnot \bvar{v}{b}{\ell}{i})$ for all units 
          $b \in \ZZ_{p^\ell}$ 
          such that $b \neq s_v \bmod p^\ell$.
      \end{itemize}
  \end{enumerate}
  \item \label{alg:final_boolean}
  Solve $(\enhance{I}, k)$ with the algorithm underlying Theorem~\ref{thm:2k2-free} and return the same answer.
\end{enumerate}
}}
\caption{Algorithm for $\minlinz{2}{p^d}$ on simple instances.}
\label{fig:mainalgorithm}
\end{figure}

We outline the steps of the algorithm in Figure~\ref{fig:mainalgorithm}.
We devote the next two sections to proving correctness of the algorithm 
via the following lemmas, and afterwards show how to derandomize it.

In Section~\ref{sec:algo-complete} we prove the following lemma.

\begin{restatable}{lemma}{lemZpdcomplete}\label{lem:Zpdcomplete}
  Let $(I, k)$ be an instance of $\minlinz{2}{p^d}$.
  If $(I,k)$ is a yes-instance,
  then the instance $(I^+, k)$ of \MinSat{} computed by \emph{\SolveMinLin} 
  is a yes-instance with probability at least $1/2^{\bigoh(k^2)}$.
\end{restatable}

In Section~\ref{sec:algo-sound} we prove the following lemma.

\begin{restatable}{lemma}{lemZpdsound}\label{lem:Zpdsound}
  Let $(I, k)$ be an instance of $\minlinz{2}{p^d}$.
  If the instance $(\enhance{I}, k)$ of \MinSat{} computed by \emph{\SolveMinLin} is 
  a yes-instance, then $(I, k)$ is a yes-instance.
\end{restatable}
\setcounter{theoremspecial}{\value{theorem}}

Here we combine all the results to prove 
Theorem~\ref{thm:algorithm-main} restated below.

\thmmainalgo*

\begin{proof}[Proof of Theorem~\ref{thm:algorithm-main}]
  Arbitrarily choose a prime power $p^d$.
  We show that \SolveMinLin
  correctly solves $\minlin{2}{\ZZ_{p^d}}$ in \FPT time with sufficient probability. 
  This is sufficient
  for proving the full result by Proposition~\ref{prop:product-approx}.

  By Proposition~\ref{prop:simple},
  it suffices to solve the problem on simple instances.
  Let $(I,k)$ denote an simple instance of 
  $\minlin{2}{\ZZ_{p^d}}$ and let $(\enhance{I}, k)$ be the
  instance that is computed in lines 1-4 of the algorithm.
  By Lemma~\ref{lem:Zpdcomplete},
  if $(I, k)$ is a yes-instance,
  then $(\enhance{I}, k)$ is a yes-instance with probability at least $1/2^{\bigoh(k^2)}$,
  and \SolveMinLin accepts it.
  If, on the other hand, $(I,k)$ is a no-instance, then $(\enhance{I}, k)$ is a no-instance
  by Lemma~\ref{lem:Zpdsound}, so \SolveMinLin correctly rejects it.
  
  We conclude the proof by verifying that the algorithm runs in \FPT time.
  
\medskip

\noindent  
Step 1. Transforming the input $(I,k)$ into special form takes polynomial time (by Section~\ref{ssec:special})
since we assume that $I$ is a simple instance.

\medskip

\noindent
Step 2. Computing $\bool{I}$ takes polynomial time by Section~\ref{sec:mincsp-reduction}.

\medskip

\noindent
Step 3. The dominant factor comes from procedure \RandomCover.
It runs in $\bigoh^*(4^{2k})$ time (by Theorem~\ref{ithm:balanced-shadow-sampling}) and is performed a fixed number of times (since $d$ is a fixed constant).

  \medskip

\noindent
Step 4. Takes polynomial time.

\medskip

\noindent
Step 5. $(I^+,k)$ is an instance of $\MinSat(\bA_{p^d})$ for a finite, bijunctive, and $2K_2$-free Boolean constraint
language $\bA_{p^d}$ that only depends on the choice of ring
$\ZZ_{p^d}$ (as pointed out in Section~\ref{sec:mincsp-reduction}). Theorem~\ref{thm:2k2-free} implies that $(I^+,k)$
can be solved in \FPT time.
\end{proof}

\subsubsection{Completeness}
\label{sec:algo-complete}

Let $(I, k)$ be an instance of $\minlinz{2}{p^d}$.
In this section we prove Lemma~\ref{lem:Zpdcomplete},
i.e.\ if $(I,k)$ is a yes-instance,
then the instance $(I^+, k)$ of \MinSat{} computed by \SolveMinLin
is a yes-instance with probability at least $1/2^{\bigoh(k^2)}$.

We first describe the intuition behind the graphs $(G_i, \BB_i)$.
To translate between the edges of the graphs
$G_i$ and the corresponding equations in $I$,
let us define $\eqn_i : E(G_i) \to \cC(I)$ that maps 
every edge of $G_i$ to the corresponding equation in $I$,
and we let $\ed_i$ be its inverse, i.e.\
$\ed_i$ maps an equation in $I$ to the set of corresponding edges in
$G_i$, i.e.\ for an equation $c$ in $I$, we have
$\ed_i(c) = \{ e \in E(G_i) : \eqn_i(e) = c \}$.
The definitions are naturally extended to sets of edges and equations.

Suppose $\alpha : V(I) \to \ZZ_{p^d}$ is an assignment and 
$Z$ is the set of constraints that $\alpha$ violates in $I$.
Recall that $\sind{a}{i}$ is the number represented by $i$ least significant digits of $a$.
On a high level, each graph $G_i$ is concerned with the variables in $I$
whose values under $\alpha$ have order at most $i$, and captures the constraints
implied on the $d-i$ least significant digits of their values.
To formalise this, say that $v^j$ in $G_i$ is activated if $\ord(\alpha(v)) \setminus j$,
and let $H_i=H_i(\alpha)$ be the subgraph of $G_i - \ed_i(Z)$ induced by the activated vertices.
Then one can show that $H_i$ is in fact a balanced subgraph of $G_i$,
i.e.\ there is a labelling of the vertices in $H_i$ with elements in $\Gamma_i$ consistent with the edge labels.
To obtain such a labelling, assign a vertex $v^j$ in $H_i$ to
$\sind{\alpha(v)}{d-i}$.
For instance, the graph $H_0$ is induced by 
the vertices $v^0$ such that $\alpha(v)$ is a unit, and 
the label of $v^0$ is $\alpha(v)$ itself.
In the graph $H_1$, the label of $v^j$ is 
the $(d-1)$-suffix of $\alpha(v)$:
specifically, if $j=0$, then the label is $\alpha(v)$ without the most 
significant digit, i.e.\ $\alpha(v) \bmod p^{d-1}$, and if 
$j=1$, then the order of $v$ is $1$ and the label is obtained
by removing the trailing zero from $\alpha(v)$, i.e.\ it is $\alpha(v) / p$.

\smallskip

To prove \Cref{lem:Zpdcomplete}, we need some auxiliary results.
The first lemma shows that the graphs $G_i(I)$ have the intended semantics, i.e.\ 
for an assignment $\alpha : V(I) \rightarrow \ZZ_{p^d}$ and a variable
$v \in V(I)$, the copies of $v$ in $G_i(I)$ represent
the part $\sind{\alpha(v)}{d-i}$ of the assignment $\alpha$, 
and the label $\gamma_i(e,u',v')$ of an edge
$e=u'v'$ in $G_i(I)$ coming from an equation $c=\eqn_i(e)$
capture what happens to the parts $\sind{\alpha(u)}{d-i}$ and
$\sind{\alpha(v)}{d-i}$ of the assignment under $c$, i.e.\ 
$\sind{\alpha(v)}{d-i}=\gamma_i(e,u',v') \cdot \sind{\alpha(u)}{d-i}$.

\begin{lemma}\label{lem:egdesGi}
  Let $G=G_i(I)$ for some $i\in\range{0}{d-2}$, $e=u'v' \in E(G)$, where $u'$
  and $v'$ are copies of
  the variables $u$ and $v$, respectively. Let $\alpha :
  \{u,v\}\rightarrow \ZZ_{p^d}$ be any assignment satisfying the
  equation $c=\eqn_i(e)$. Then, 
  $\gamma_i(e,u',v') \cdot \sind{\alpha(u)}{d-i}=\sind{\alpha(v)}{d-i}$.
\end{lemma}
\begin{proof}
  We consider the following cases based on the type of the edge $e$.
  If $e$ is of the form $u^jv^j$ for some $j$, then, without
  loss of generality, $c$ is
  of the form $r \cdot u=v$ for some unit $r$ and $\gamma_i(e,u',v')=r
  \bmod p^{d-i}$. Because $\alpha$ satisfies $c$, it holds that $r\cdot\alpha(u)=\alpha(v)$ and therefore
  $\gamma_i(e,u',v')\cdot \sind{\alpha(u)}{d-i}=\sind{\alpha(v)}{d-i}$.

  If $e$ is of the form $u^jv^{j+1}$ for some $j$, then $c$ is
  of the form $p \cdot u=v$ and $\gamma_i(e,u',v')=1$. Because $\alpha$
  satisfies $c$, it holds that $p\cdot \alpha(u)=\alpha(v)$ and therefore
  $\gamma_i(e,u',v')\cdot \sind{\alpha(u)}{d-i}=\sind{\alpha(v)}{d-i}$.
\end{proof}
The next two lemmas shows that the equations $\eqn_i(C)$ corresponding to
the edges of any unbalanced cycle $C$ in the biased graph
$(G_i(I),\BB_i(I))$ cannot be simultaneously satisfied if at least one
variable $v$ in $V(\eqn_i(C))$ that occurs on $C$ as $v^j$ is assigned
to a value of order $j$. In other words, it shows the
main role of the graphs $G_i(I)$, i.e.\ taking care of obstructions
(to satisfying assignments)
resulting from such unbalanced cycles.
\begin{lemma}\label{lem:cycles_Gi}
  Let $I$ be an instance of $\linz{2}{p^d}$, let $G=G_i(I)$ and
  $\BB=\BB_i(I)$ for some $i \in \range{0}{d-2}$ and let $C$ be a
  cycle in $G$. If $C \notin \BB$, then the corresponding
  system $\eqn_i(C)$ of equations in $I$ is not satisfiable using any
  assignment that sets at least one variable $v \in V(\eqn_i(C))$ that
  occurs on $C$ as $v^j$ to a value of order $j$.
\end{lemma}
\begin{proof}
  Assume for a contradiction that the equations $\eqn_i(C)$ are
  satisfiable by some assignment $\alpha : V(\eqn_i(C)) \rightarrow
  \ZZ_{p^d}$ such that $\alpha(v)=r \cdot p^{j}$ for some unit $r$ and
  some $0 \leq j \leq i$. Let $(e_0,\dotsc,e_\ell)$ be the edges on
  $C$ in cyclic order such that $e_x$ is an edge between $v_x$ and
  $v_{x+1 \bmod \ell+1}$ for every $x \in \range{0}{\ell}$. Because of
  Lemma~\ref{lem:egdesGi}, we have that
  $\alpha$ must satisfy
  \[ 
    \sind{\alpha(v)}{d-i}=\left( \prod_{x=0}^{\ell}\gamma_i(e_x,v_x,v_{x+1 \bmod
    \ell+1}) \right) \cdot \sind{\alpha(v)}{d-i}. 
\]
Note that since $\sind{\alpha(v)}{d-i}= r/p^j \bmod p^{d-i}$ is a unit, it
holds that $\sind{\alpha(v)}{d-i}=b\cdot \sind{\alpha(v)}{d-i}$ if and only if
$b=1$. However, because $C \notin \BB$, it holds that
$(\prod_{x=0}^{\ell}\gamma_i(e_x,v_x,v_{x+1 \bmod
  \ell+1}))\neq 1$ and the above equation cannot be satisfied.
\end{proof}
\begin{lemma}\label{lem:HAlphaBal}
  Let $I$ be an instance of $\linz{2}{p^d}$, $Z\subseteq C(I)$
  such that $I-Z$ is satisfiable using the assignment $\alpha : V(I)
  \rightarrow \ZZ_{p^d}$, $G=G_i(I)$, $\BB=\BB_i(I)$, and let
  $H=H_i(\alpha,Z)$ be the subgraph of $G\setminus \ed_i(Z)$ induced by the vertices
  in $\SB v^j : \ord(\alpha(v))=j\SE$.
  Then, $H$ is balanced (in $(G,\BB)$) and has cost at most
  $2|Z|$.
\end{lemma}
\begin{proof}
  First note that $\ed_i(Z)$ contains at most two edges in
  $\delta_G(V(H))\cup E(G[V(H)])$ for every equation $z \in Z$. This is
  because $H$ contains at most one copy of every vertex
  corresponding to a variable of $I$ and therefore the number of edges
  in $\ed_i(z)$ incident to a vertex in $H$ is at most two for every
  $z \in Z$. Therefore, if 
  $$\delta_G(V(H))\cup E(G[V(H)])\setminus E(H) \subseteq
  \ed_i(Z),$$ then $H$ has cost at most $2|Z|$. Since
  $E(G[V(H)])\setminus E(H) \subseteq \ed_i(Z)$ holds due to the
  definition of $H$, it remains to show that $\delta_{G}(V(H))\subseteq \ed_i(Z)$.
  Assume for a
  contradiction that
  $e=u^xv^y \in \delta_{G}(V(H))\setminus \ed(Z)$ with $u^x \in V(H)$ and $v^y \notin
  V(H)$. Let $c=\eqn_i(e)$, which because $c \notin Z$ must be
  satisfied by $\alpha$. Moreover,
  because $u^x \in V(H)$, it holds that $\alpha(u)=r\cdot p^x$
  for some unit $r$. We now distinguish the following cases.
  
  If $x=y$, then, without loss of generality,
  $c=\eqn_i(e)$ is an equation of the form $r' \cdot u=v$
  for some unit $r'$. Because $\alpha$ satisfies $c$,
  we obtain that $\alpha(v)=r'\cdot r\cdot
  p^x=r'\cdot r\cdot p^y$ and therefore $v^y \in V(H)$, a contradiction.

  Otherwise, $y=x+1$ or $y=x-1$. In the former case,
  $c=\eqn_i(e)$ is an equation of the form $p \cdot u=v$.
  Because $\alpha$ satisfies $c$, we obtain that $\alpha(v)=p\cdot r\cdot
  p^x=r\cdot p^y$ and therefore $v^y \in V(H)$, a contradiction.
  In the latter case,
  $c=\eqn_i(e)$ is an equation of the form $p \cdot v=u$. Because
  $\alpha$ satisfies $c$, we see that $\alpha(v)=r\cdot
  p^{x-1}=r\cdot p^y$ and consequently $v^y \in V(H)$, a contradiction.

  We conclude that the cost of $H$ is at most $2|Z|$ and it only remains to show
  that $H$ is balanced in $(G, \BB)$.
  Assume for a contradiction that $H$ is not balanced and let $C$ be
  a cycle of $H$ that is not in $\BB$. Because of
  Lemma~\ref{lem:cycles_Gi}, we obtain that the corresponding system of
  equations $\eqn_i(C)$ is not satisfiable, contradicting our
  assumption that $I-Z$ is satisfiable.
\end{proof}

We are now ready to show \Cref{lem:Zpdcomplete}.

\begin{proof}[Proof of \Cref{lem:Zpdcomplete}]
  Let $Z \subseteq C(I)$ with $|Z|\leq k$ such that $I-Z$ has a
  satisfying assignment $\alpha : V(I) \rightarrow \ZZ_{p^d}$ and let
  $\beta=\bool{\alpha}$. Note that $\beta$ satisfies $\bool{(I-Z)}$.

  For $i \in \{0,\dotsc,d-2\}$, let $G_i=G_i(I)$ be the auxiliary graph constructed in step 2.(a)
  and let $(S_i',F_i')$ be the pair given by $\RandomCover(G_i, \BB_i, i
  \cdot k)$ in step~\ref{alg:random_cover}. Moreover, let $H_i$ be the subgraph of
  $G_i\setminus \ed_i(Z)$ induced by the vertices
  in $\SB v^j \SM \ord(\alpha(v))=j\SE$.
  Because of Lemma~\ref{lem:HAlphaBal}, $H_i$ is balanced in $(G_i,\BB_i)$
  and has cost at most $2|Z|\leq 2k$.

  Applying Theorem~\ref{ithm:balanced-shadow-sampling} for $H=H_i$, $S=S_i'$,
  $F=F_i'$ and $k$, we obtain that with probability at least
  $1/2^{\bigoh(k^2)}$, $H_i$
  is covered by $S_i'$, i.e. $V(H_i)\subseteq S_i'$ and moreover
  $F_i'$ contains at most $k$ edges with at least one endpoint in $V(H_i)$.

  Let $S_i = \bigcap_{j=i}^{d-2} \SB v \in V(I) \SM v^i \in S'_j \SE$ as
  defined in step~\ref{alg:compute_si}. Since $\SB v^i \SM v^i \in V(H_i)\SE
  \subseteq \SB v^i \SM v^i \in V(H_j)\SE$ for every $j\geq i$ and
  because $S_i'$ covers $H_i$ for every $i$, we obtain that $\SB v^i
  \SM v^i \in V(H_i)\SE \subseteq S_i$. Therefore, if $v \notin S_i$,
  then $\ord(\alpha(v))\neq i$, which shows that the constraints
  introduced in step~\ref{alg:reduce_domain} are satisfied by $\beta$.

  Let $T_i = \bigcup_{j=i}^{d-2} \SB v \in S_i \SM v^i \in V(F'_j) \SE$ as
  defined in step~\ref{alg:compute_ti}. First note that for every $v \in T_i$ such
  that $\ord(\alpha(v))\neq i$, it holds that
  all constraints introduced in step~\ref{alg:add_implications} are satisfied by $\beta$.

  We now show that there are at most $2k(d-i-2)$ variables $v$ in
  $T_i$ with $\ord(\alpha(v))=i$. To see this, first recall that for
  every $j\geq i$, it holds that $F_j'$ contains
  at most $k$ edges with at least one endpoint in
  $V(H_j)$. Therefore, $|V(F_j')\cap V(H_j)|\leq 2k$, which implies that $|T_i\cap
  V(H_i)|$ is at most $2k(d-i-2)$, as required. Consequently, the
  constraints introduced in step~\ref{alg:add_implications}, are satisfied by $\beta$ with probability at
  least $$(p^d)^{2k(d-i-2)}=2^{\bigoh(k)}$$ for every variable $v$ in $T_i$ with
  $\ord(\alpha(v))=i$. Taking everything together we obtain that
  $\beta$ satisfies all constraints in \bool{(I-Z)} with probability
  at least $1/2^{\bigoh(k^2)}$.
\end{proof}

\subsubsection{Soundness}
\label{sec:algo-sound}

Let $(I, k)$ be an instance of $\minlinz{2}{p^d}$.
In this section we prove Lemma~\ref{lem:Zpdsound},
i.e.\ if the instance $(\enhance{I}, k)$ of \MinSat{} computed by 
\SolveMinLin is 
a yes-instance, then $(I, k)$ is a yes-instance.

First, note that since $\bool{I}$ is a subinstance if $\enhance{I}$,
if $\enhance{I}$ admits an unambiguous assignment of cost $k$,
then $(I, k)$ is a yes-instance by Lemma~\ref{lem:unambiguous_Boolean_to_modpn}.
Thus, to prove soundness,
it suffices to show that one 
can disambiguate any Boolean assignment to $\bool{I}$
without increasing its cost.
Towards this goal, we introduce some notation to extract information
from an ambiguous Boolean assignment.

Let $\beta$ be a Boolean assignment to $\enhance{I}$ of finite cost,
i.e.\ satisfying all crisp constraints.
Consider a variable $v$.
Intuitively, the clauses in $\bool{K}(v)$ force $\beta$ to choose
a coset for $v$.
More formally, we have the following.

\begin{proposition} \label{proposition:order_and_suffix}
  Let $\beta$ be a Boolean assignment satisfying $\bool{K}(v)$
  that assigns at least one variable in $\bool{V}(v)$ to $1$.
  Then there is a unique $\coset{a}{\ell}{i}$ with maximum
  $\ell$ such that $\beta$ assigns $1$ to the variable
  $\bvar{v}{a}{\ell}{i}$ and all variables implied by it in $\bool{K}(v)$,
  and $0$ to all other variables.
\end{proposition}

With this in mind, \emph{order} of a variable $v \in V(I)$ under $\beta$ as 
$\ord_\beta(v) = 0$ if $\beta$ assigns $0$ to all Boolean variables in $\bool{V}(\beta)$,
and
$\ord_\beta(v) = i$ if $\beta(\bvar{v}{a}{\ell}{i}) = 1$ for some $i$.
Note that since $\beta$ has finite cost, it satisfies $\bool{K}(v)$, so
there is exactly one such $i$.
Moreover, if $v$ is nonzero under $\beta$,
then $\beta$ also fixes some least significant digits.
To extract them, for a nonzero $v$,
let $\bvar{v}{a}{\ell}{i}$ be the Boolean variable in $\bool{V}(v)$
assigned $1$ under $\beta$ with maximum $\ell$,
and let $\suff(v) = a$ and $\wgt(v) = \ell$.
The next observations follow from the definitions of 
Boolean constraints in $\bool{I}$ corresponding to the equations of $I$.

\begin{proposition}
  \label{prop:string_mult}
  Let $\beta$ be an assignment that 
  satisfies all crisp constraints in $\bool{I}$.
  Suppose $\beta$ satisfies $\bool{C}$ for some equation $C$ in $I$.
  
  \begin{enumerate}
    \item If $C$ is of the form $v = r \cdot u$ for some unit $r$, then
      $\ord_\beta(v) = \ord_\beta(u)$, $\wgt_\beta(v) = \wgt_\beta(u)$ and
      $\suff_\beta(v) = r \cdot \suff_\beta(u) \bmod p^{\wgt_\beta(u)}$. 
    \item If $C$ is of the form $v = p \cdot u$ and $u$ or $v$ is ambiguous, then
      $\ord_\beta(v) = \ord_\beta(u) + 1$,
      $\wgt_\beta(v) = \wgt_\beta(u)$ and
      $\suff_\beta(v) = \suff_\beta(u)$.  
  \end{enumerate}
\end{proposition}

We refer to $\wgt_\beta(v)$ as the \emph{weight} of $v$ under $\beta$,
and say that $v$ is \emph{$\ell$-ambiguous} if
it is ambiguous and has weight $\ell$ under $\beta$.
We are ready to prove
Lemma~\ref{lem:Zpdsound}.

\begin{proof}[Proof of Lemma~\ref{lem:Zpdsound}]
  Let $\beta$ be a Boolean assignment to $\enhance{I}$ that satisfies all crisp constraints.
  Let $\ell$ be the minimum weight of an ambiguous variable under $\beta$.
  It suffices to show that we can modify $\beta$ by setting 
  more Boolean variables to $1$ while increasing 
  the minimum weight of ambiguous variables under $\beta$
  and without affecting the cost of $\beta$:
  indeed, the process can be repeated until there are no ambiguous variables left.
  
  Let $Z = \{C \in C(I) : \beta \text{ does not satisfy } \bool{C}\}$ be 
  the set of equations in $I$ that $\beta$ violates. 
  We will construct a new Boolean assignment $\beta'$ 
  that extends $\beta$ (i.e.\ assigns $1$ to a superset of Boolean variables)
  so that no variable is $\ell$-ambiguous under $\beta'$,
  and then argue that $\beta'$ satisfies $\bool{C}$ for all 
  equation $C \notin Z$.
  For that, we consider the graph 
  \[ 
    G' = (G_{d-\ell-1} - F'_{d-\ell-1} - \ed_{d-\ell-1}(Z))[S'],
  \]
  where $S' \subseteq S'_{d-\ell-1}$ contains
  vertices $v^i \in S'_{d-\ell-1}$ such that $v \in S_i$,
  as defined in step~\ref{alg:compute_si}.
  
  Since $G'$ is a subgraph of 
  $(G_{d-\ell-1} - F'_{d-\ell-1})[S'_{d-\ell-1}]$, 
  it admits a consistent labelling 
  $\lambda : V(G') \to \Gamma_{d-\ell-1}$.
  We will use $\lambda$ to obtain $\beta'$ from $\beta$.
  To do this, we observe that a copy of every $\ell$-ambiguous variable occurs in $G'$.

  \begin{claimspecial} \label{claim:found_a_copy}
    If $v$ is $\ell$-ambiguous under $\beta$,
    then $v^i \in V(G')$ where $i = \ord_\beta(v)$.
  \end{claimspecial}
  \begin{claimproof}
    If $v$ is $\ell$-ambiguous and has order $i$,
    then $i + \ell < d$, so $i \leq d - \ell - 1$
    and $v^i$ occurs as a vertex in $G_{d-\ell-1}$.
    Moreover, crisp constraints in step~\ref{alg:reduce_domain}
    forbid variables outside $S_i$ from having order $i$,
    so $v \in S_i$.
    Thus, $v^i \in S'$ and $v^i \in V(G')$.    
  \end{claimproof}

  Let $U$ be a set of $\ell$-ambiguous variables whose copies
  induce a connected component of $G'$.
  Pick an arbitrary variable $v \in U$ and suppose $v^i \in V(G')$.
  Note that $\suff_\beta(v)$ is a unit in $\ZZ_{p^\ell}$,
  and hence in $\ZZ_{p^{\ell+1}}$.
  Thus, $\suff_\beta(v)$ belongs to the group $\Gamma_{d-\ell-1}$,
  and there exists $c \in \Gamma_{d - \ell - 1}$ such that
  $\lambda(v^i) \cdot c = \suff_\beta(v) \bmod p^\ell$.
  For each $u \in U$, set
  $\beta'(\bvar{u}{\lambda(u^{i}) \cdot c}{\ell+1}{i}) = 1$
  where $\ord_\beta(u) = i$.
  Note $\suff_{\beta'}(u) = \lambda(u^i) \cdot c$ 
  and $\ord_{\beta'}(u) = \ord_{\beta}(u)$,
  so the weight of $u$ increases by $1$ in $\beta'$ compared to $\beta$,
  while the order is unchanged.
  This completes the construction of $\beta'$.
  By Claim~\ref{claim:found_a_copy},
  no variable is $\ell$-ambiguous under $\beta'$, and
  it remains to show that $\beta'$ satisfies $\bool{C}$ 
  for all $C \notin Z$.

  \smallskip
  
  Consider an equation $C \in C(I) \setminus Z$ on variables $u$ and $v$.
  Clearly, if $\beta$ and $\beta'$ agree on $u$ and $v$,
  then $\bool{C}$ is satisfied by $\beta'$.
  We claim that this case covers all equations $C$
  such that $\ed_{d-\ell-1}(C)$ intersects $F'_{d-\ell-1}$.
  Note that $\beta$ and $\beta'$ disagree on $\bool{K}(v)$
  if and only if $v$ is $\ell$-ambiguous under $\beta$.

  \begin{claimspecial} \label{claim:terminals}
    If $\ed_{d-\ell-1}(C)$ intersects $F'_{d-\ell-1}$,
    then $\beta$ and $\beta'$ agree on $\bool{C}$.
  \end{claimspecial}
  \begin{claimproof}
    Let $u$ and $v$ be the variables of $C$, and 
    suppose $u^i v^j \in \ed_{d-\ell-1}(C) \cap F'_{d-\ell-1}$.    
    We claim that $u$ and $v$ are not $\ell$-ambiguous under $\beta$.
    Indeed, if $u \notin S_i$, then $u^i \notin V(G')$ by Claim~\ref{claim:found_a_copy}.
    Otherwise, we have $u \in S_i$ and $u \in T_i$
    by the construction in step~\ref{alg:compute_ti};
    then $\bool{I}$ contains crisp constraints 
    introduced in step~\ref{alg:add_implications}
    implying that $u$ has order $i$, then it is unambiguous.
    By analogous reasoning, neither $v$ is $\ell$-ambiguous.
  \end{claimproof}
  
  In the remaining cases, we have $\ed_{d-\ell-1}(C) \cap F'_{d-\ell-1} = \emptyset$,
  while $\beta$ and $\beta'$ disagree on at least one of $u$ and $v$.

  \smallskip

  \emph{Case 1:} $\beta(u) \neq \beta'(u)$ and $\beta(v) \neq \beta'(v)$,
  i.e.\ both $u$ and $v$ are $\ell$-ambiguous under $\beta$.
  Let $i = \ord_\beta(u)$ and $j = \ord_\beta(v)$.
  The edge $u^i v^j$ is present in $G'$ since 
  $u^i, v^j \in V(G')$ by Claim~\ref{claim:found_a_copy},
  $u^i v^j$ is not in $F'_{d-\ell-1}$ by Claim~\ref{claim:terminals} and
  not in $\ed_{d-\ell-1}(Z)$ since $C \notin Z$.
  Thus, $u^i$ and $v^j$ are in the same connected component of $G'$,
  which implies that $\suff_{\beta'}(u) = \lambda(u^i) \cdot c$ and
  $\suff_{\beta'}(v) = \lambda(v^i) \cdot c$ for the same $c \in \Gamma_{d-\ell-1}$. 
  We consider the two possible subcases.

\begin{enumerate}
  \item
  $C$ is of the form $r \cdot u = v$ for some unit $r$.
  Then we have 
  $i = j$ and $r \cdot \lambda(u^i) = \lambda(v^j) \bmod p^{\ell+1}$.
  Thus, $r \cdot \suff_{\beta'}(u) = r \cdot \lambda(u^i) \cdot c = 
  \lambda(v^j) \cdot c = \suff_{\beta'}(v) \mod p^{\ell+1}$,
  and $\beta'$ satisfies $\bool{C}$ by Proposition~\ref{prop:string_mult}.

  \item
  $C$ is of the form $p \cdot u = v$. We see that
  $j = i + 1$, $\lambda(u^i) = \lambda(v^j) \bmod p^{\ell+1}$
  and 
  $\suff_{\beta'}(u) = \suff_{\beta'}(v) \bmod p^{\ell+1}$.
  By Proposition~\ref{prop:string_mult},
  this implies that $\beta'$ satisfies $\bool{C}$.
\end{enumerate}

  \smallskip

  \emph{Case 2:} $\beta(u) \neq \beta'(u)$ and $\beta(v) = \beta'(v)$,
  i.e.\ $u$ is $\ell$-ambiguous under $\beta$ and $v$ is not.
  By Proposition~\ref{prop:string_mult}, this is only possible if
  $C$ is of the form $p \cdot u = v$ and
  $\ord_\beta(u) + 1 + \ell = d$.
  Thus, $\beta'$ assigns one of the singleton variables in
  $\bool{V}(u)$ to $1$. 
  By construction, $\bool{C}$ does not contain constraints involving 
  singleton variables from $\bool{V}(u)$.
  Since $\beta'$ agrees with $\beta$ on all other Boolean
  variables in $\bool{V}(u)$,
  it satisfies $\bool{C}$.
\end{proof}

\subsubsection{Derandomization}

There are two randomized ingredients in the algorithm in Figure~\ref{fig:mainalgorithm}:
the first is the procedure \RandomCover{} being called in step~\ref{alg:random_cover};
the second is the choice of $s_v$ for each $v \in T_i$ uniformly at random in step~\ref{alg:add_implications}.
The first one was derandomized in Theorem~\ref{thm:derandom_cover}.
For the second one, we may combine perfect hash families with exhaustive enumeration.
Perfect hash families are a standard tool in derandomizing FPT algorithms.

\begin{definition}
  An $(n, \kappa)$-perfect hash is family $\cH$
  of functions $h : [n] \to [\kappa]$ such that,
  for all subsets $S \subseteq [n]$ of size $\kappa$
  there is a function $h \in \cH$ that is
  injective on $S$, i.e.\ 
  $h(i) \neq h(j)$ for all distinct $i,j \in S$.
\end{definition}

We will need the following result.

\begin{theorem}[Naor et al.~\cite{naor1995splitters}]
  For any $n, \kappa \geq 1$, one can construct
  an $(n, \kappa)$-perfect hash family of size
  $e^\kappa \kappa^{O(\log \kappa)} \cdot \log n$ in time
  $e^\kappa \kappa^{O(\log \kappa)} \cdot  n \log n$
\end{theorem}

Now let $D$ be a set.
There are $|D|^\kappa$ functions from $[\kappa]$ to $D$.
By composing an $(n, \kappa)$-perfect hash
with all functions $f : [\kappa] \to D$,
we obtain the following.

\begin{lemma} \label{lem:nkd_function}
  For any $n, \kappa \geq 1$ and any set $D$ of size $d$,
  one can construct a family $\cF$ of functions
  $f : [n] \to D$ such that for all
  subsets $S \subseteq [n]$ with $|S| = \kappa$
  and every function $g : S \to D$,
  there is a function $f \in \cF$
  that agrees with $g$ on $S$, i.e.\ 
  $f(x) = g(x)$ for all $x \in S$.
  Moreover, the size of $\cF$ is
  $(ed)^\kappa \kappa^{O(\log \kappa)} \cdot \log n$
  and it can be constructed in time
  $(ed)^\kappa \kappa^{O(\log \kappa)} \cdot  n \log n$.
\end{lemma}

Thus, to derandomize step~\ref{alg:add_implications},
we may use the lemma above to
construct a family $\cF$ with parameters
$n = |T_i|$, $\kappa = 2k(d-i-2)$ and 
$D$ being the set of units in $\ZZ_{p^{d-i}}$, for each iteration $i$.
We replace guessing $s_v$ for each $v \in T_i$
by enumerating all functions $f \in \cF$ 
and setting $s_v = f(v)$ to produce the
implication constraints in step~\ref{alg:add_implications}.
The domain size is bounded as $|D| < |\ZZ_{p^{d-i}}|$,
hence for each iteration Lemma~\ref{lem:nkd_function}
introduces a multiplicative overhead of at most
$p^{O(kd^2)} \log n$. 

We summarize as follows. We take the opportunity to also chase down a rough bound on the running time and dependency on $p$ and $d$.

\begin{theorem} \label{thm:alg-derand}
  There is a deterministic FPT algorithm for $\minlin{2}{\ZZ_{p^d}}$
  for any prime $p$ and $d \geq 1$ running in time $O^*(2^{(k+p^d)^{O(1)}})$.
\end{theorem}
\begin{proof}
  As discussed above, we modify $\SolveMinLin(I, k)$, given in Figure~\ref{fig:mainalgorithm},
  so that Step~\ref{alg:forloop_1} uses Theorem~\ref{thm:derandom_cover} instead of a call to $\RandomCover$
  and Step~\ref{alg:add_implications} uses Lemma~\ref{lem:nkd_function} instead of random guessing. 
  In both cases, the algorithm iterates exhaustively over the respective family of objects
  and accepts if and only if at least one resulting branch accepts the instance.  

  For correctness, we only need to verify completeness. 
  Let $\alpha : V(I) \to \ZZ_{p^d}$ be an optimal assignment to $(I, k)$,
  and let $\bool{\alpha}$ be the corresponding Boolean assignment.
  As in the proof of Lemma~\ref{lem:Zpdcomplete}, for each $i \in \{0,\ldots,d-2\}$
  there is a balanced subgraph $H_i$ of $G_i$ of cost at most $2k$.
  Thus, the output of Theorem~\ref{thm:derandom_cover} on $(G_i,\BB_i,2k)$
  contains one pair $(S_i,F_i)$  which covers $H_i$. 
  Moreover, as in Lemma~\ref{lem:Zpdcomplete}
  at most $2k(d-i-2)$ variables $v$ in $T_i$ have $\ord(\alpha(v)) = i$.
  The implications introduced in step~\ref{alg:add_implications}
  are satisfied by $\bool{\alpha}$ for all variables $v$ in $T_i$
  with $\ord(\alpha(v)) \neq i$ because the left-hand-side
  of the implications is false under $\bool{\alpha}$,
  and by Lemma~\ref{lem:nkd_function} there is $f \in \cF$ such that
  for all $v \in T_i$ with $\ord(\alpha(v)) = i$, 
  $\alpha(v) = f(v) \cdot p^i$, so these implications
  are satisfied by $\bool{\alpha}$ for at least one $f \in \cF$.
  Thus there is a branch where the instance is accepted.

  For running time, all operations are of lower-order term except
  the product of iteration over $O(d)$ outputs from Theorem~\ref{thm:derandom_cover} with parameter~$2k$,
  $O(d)$ outputs from Lemma~\ref{lem:nkd_function},
  and the call to Theorem~\ref{thm:2k2-free} in the end. 
  The former two parts together introduce a total overhead of
  $$2^{O(k^2d + kd^3 \log p)}\log^{O(d)} n.$$ 
  When we treat $d$ as a constant,
  this is already in our time budget. For taking $p^d$ as an additional parameter,
  we note the well-known fact that $\log^t n \leq n+t^{O(t)}$.
  Finally, Theorem~\ref{thm:2k2-free} has a deterministic form which runs in $O^*(2^{(k+r)^{O(1)}})$ time, 
  where $r$ is the maximum arity of a constraint in the Boolean formula $\enhance{I}$;
  see Theorem~1.6 in~\cite{Kim:etal:sicomp2025}. We have $r \leq dp^d$
  (see the representation of cosets in Section~\ref{ssec:cosets}).
  Thus the total time is $O^*(2^{(k+p^d)^{O(1)}})$ as promised.
\end{proof}

\section{Hardness of FPT Approximation}
\label{sec:hardness}

The aim of this section is to prove that, under the ETH, $\minlin{2}{\ZZ_m}$ is not \FPT{}-approximable within $\omega(m)-\epsilon$ for any $\epsilon > 0$ (Corollary~\ref{cor:modular-hardness}). 
From now on, we will formulate approximation problems as {\em gap problems}
so that we can view them as decision problems.
Hence, given a minimization problem $P$, the equivalent \textsc{Gap$_c$} problem distinguishes between two cases, for an input $k$ and an instance $x$ of problem $P$:

\smallskip

\begin{tabular}{ll}
(YES) & ${\displaystyle {\text{OPT}}_{P}(x)\leq k}$.\\
(NO) & ${\displaystyle {\text{OPT}}_{P}(x)>c \cdot k}$. 
\end{tabular}

\smallskip

\noindent
If the optimum value is between the thresholds, then there is no requirement on the output. Note that $c$-approximating the solution to an instance of $P$ is at least as hard as solving the $c$-gap version of $P$.

Our starting point is the connection between
$\minlin{2}{\ZZ_m}$ and graph separation problems already observed in
\cite{Dabrowski:etal:arxiv2024}. In particular, Dabrowski et
al.~\cite{Dabrowski:etal:arxiv2024} showed that $\GAP{2-\eps}$ $\minlin{2}{\ZZ_m}$ is
not in \FPT for any $\eps>0$, unless the \ETH\ is false; note that this
is equivalent to saying that $\minlin{2}{\ZZ_m}$ is
not FPT-approximable within $2-\eps$ for any $\eps>0$, unless the \ETH
is false. Crucially, Dabrowski et al.~\cite{Dabrowski:etal:arxiv2024}
employed the well-known \textsc{Pair Min Cut}
problem as an intermediate problem, where given a graph $G$, two vertices $s$ and $t$ of $G$, 
a set $\mathcal{P}$ of disjoint edge pairs (or bundles), and a
parameter $k$, the aim is to decide whether there is an $st$-mincut $Z
\subseteq E(G)$, which is the union of at most $k$ bundles in
$\mathcal{P}$. They showed that $\GAP{2-\eps}$ \textsc{Pair Min Cut} is
not in \FPT for any $\eps>0$, unless the \ETH is false, even for
\emph{split} instances, i.e.\ where $G-\{s,t\}$ is the disjoint union
of two graphs $G_1$ and $G_2$ and every bundle contains at most one
edge from each $G_i$.

As a first step to obtain our lower bounds, we introduce a natural generalization of
{\sc Paired Min Cut} for split instances, which we coin
\pSplitng. Recall from \Cref{sec:intro-hardness} that an instance of
the \pSplitng problem is
a graph $G$ with two terminals $s, t$ such that $G-\{s,t\}$ is the disjoint union of 
$p$ graphs $G_1,\dotsc,G_p$, together with
a set of disjoint edge sets $\cB \subseteq \binom{E(G)}{1} \cup \binom{E(G)}{p}$,
where each set contains at most one edge from each $G_1,\dots,G_p$, and
an integer $k$, which also serves as the parameter for the problem.
The members of $\cB$ are referred to as {\em bundles}.
The question is whether 
there is an $st$-cut which is the union of at most $k$ bundles in
$\cB$. 

Clearly, $2$-{\sc Split Min Cut} is the same problem 
as \textsc{Pair Min Cut} for split instances. 
Unfortunately, Dabrowski et al.'s hardness result for
$\GAP{2-\eps}$ $2$-{\sc Split Min Cut}
is not easily generalizable to a hardness result for
$\GAP{p-\eps}$ $p$-{\sc Split Min Cut} with $p > 2$.
Their reduction starts from
a sparse version of
the $\GAP{\eps}$ {\sc Max-2-CSP} problem\footnote{See formal definition just before Theorem~\ref{thm:Bafna-hardness}}, and
the non-approximability of this problem
follows from Theorem~4.1~of~\cite{guruswami2024almost} combined with
a result concerning parallel repetition in projection games~\protect\cite{Rao:sicomp2011}.
They use a straightforward probabilistic argument for showing that the same non-approximability result holds for
instances with bipartite primal graphs, and this
is then transferred to the $2$-{\sc Split Min Cut} problem.
Generalizing their probabilistic argument to $p$-partite
graphs with $p > 2$ does not seem possible so we 
use a different approach based on a recent non-approximability result for
a {\em dense} version of 
the $\GAP{\eps}$ {\sc Max-2-CSP} problem~\cite{Bafna:etal:stoc2025}.
The main
technical contribution of our lower bound then lies in the definition
of an tailor-made intermediate problem (as well as reductions to and from this
problem), which we coin \gapsven (\gapsvens).
\svens\ is a restricted variant of the well-known
\textsc{Multicoloured Subgraph Isomorphism} problem (see, e.g.~\cite{DBLP:journals/toc/Marx10}), where the aim is
to find a subgraph isomorphism from a $t$-partite graph $H$ with each
part having parameter many ($k$) vertices and whose edges can be
partitioned into $\Omega(k^2)$ $t$-cliques.
To show the hardness of this problem, we
provide a reduction from the above mentioned dense version of the
\textsc{Max-2-CSP} problem using the \textsc{Multicoloured Densest Subgraph}
problem as an intermediate step.
A crucial ingredient for the reduction is
to show that for sufficiently large $k$ one can construct a
$t$-partite graph $H$ with $k$ vertices in each part whose edges can
be partitioned into at least $\Omega(k^2)$ $t$-cliques. We show this
result with the help of a classical result by Turán
stating the maximum number of edges in a $K_r$-free graph with $n$
vertices. Using an delicate reduction from \gapsvens, we are then able
to show that \pSplit\  is not in \FPT.

We continue by proving (in Theorem~\ref{thm:incomparable-annihilators}
and Corollary~\ref{cor:modular-hardness}) that $\GAP{\omega(m)-\eps}$
$\omega(m)$-{\sc Split Min Cut} fpt-reduces to $\GAP{\omega(m)-\eps}$
$\minlin{2}{\ZZ_m}$, and consequently $\minlin{2}{\ZZ_m}$ is not
\FPT{}-approximable within $\omega(m)-\epsilon$ for any $\epsilon >
0$.
The full reduction chain in our hardness proof is outlined in
Figure~\ref{fig:hardness-structure}. We start by introducing \svens\
and showing hardness for the problem in \Cref{sec:sven-hard}. We then
show hardness for
$\GAP{p-\eps}$ \pSplit in
Section~\ref{sec:graph-hard} and provide the reduction from $\GAP{\omega(m)-\eps}$
$\omega(m)$-{\sc Split Min Cut} to $\GAP{\omega(m)-\eps}$
$\minlin{2}{\ZZ_m}$ in Section~\ref{sec:equation-hard}.

\begin{figure}
\begin{center}
$\GAP{\eps}$ {\sc Max-2-CSP} with complete primal graph is not in \FPT for any $\epsilon > 0$
\end{center}
\vspace{-2mm}
\hspace{7cm}$\big\downarrow$ {\footnotesize(Lemma~\ref{lem:gapmultidenshard})} 
\vspace{-2mm}
\begin{center}
$\GAP{\eps}$ {\sc Multicoloured Densest Subgraph} is not in \FPT for any $\epsilon > 0$
\end{center}
\vspace{-2mm}
\hspace{7cm}$\big\downarrow$ {\footnotesize(Lemma~\ref{lem:gapsvenhard})} 
\vspace{-2mm}
\begin{center}
\gapsven\\ is not in \FPT for any $t \geq 2$ and $\epsilon > 0$
\end{center}
\vspace{-2mm}
\hspace{7cm}$\big\downarrow$ {\footnotesize(Lemma~\ref{lem:gap-p-cut})} 
\vspace{-2mm}
\begin{center}
$\GAP{p-\eps}$ $p$-{\sc Split Min Cut} is not in \FPT for any $p \geq 2$ and $\epsilon > 0$
\end{center}
\vspace{-2mm}
\hspace{7cm}$\big\downarrow$ {\footnotesize(Theorem~\ref{thm:incomparable-annihilators}, Corollary~\ref{cor:modular-hardness})} 
\vspace{-2mm}
\begin{center}
$\GAP{\omega(m)-\eps}$ $\minlin{2}{\ZZ_m}$ is not in \FPT{} for any $\epsilon > 0$.
\end{center}
\caption{Reduction chain in hardness proof}
\label{fig:hardness-structure}
\end{figure}

\subsection{Definition and Hardness of \svens}
\label{sec:sven-hard}

The goal of this section is to introduce our variant of the
\textsc{Multicoloured Subgraph Isomorphism} problem, i.e.\ the \svens\
problem, and show its hardness. We start by showing hardness for the $\GAP{\eps}$ {\sc Multicoloured Densest
Subgraph} problem using a reduction from a dense variant of $\GAP{\eps}$ {\sc Max-2-CSP}.

\pbDefGapnonparam{$\GAP{\eps}$ {\sc Max-2-CSP}}
{A set of variables $V$, a value domain $D$, and a set of constraints $C$ where each constraint
is of the form $R(x,y)$ with $x,y \in V$ and $R \subseteq D^2$.}
{There is a function $f:V \rightarrow D$ such that for every $R(x,y) \in C$, it holds that
$(f(x),f(y)) \in R$.}
{For every function $f:V \rightarrow D$, there are at most $\epsilon \cdot |C|$ constraints $R(x,y) \in C$ such that $(f(x),f(y)) \in R$.}

Let $I=(V,D,C)$ denote an arbitrary instance of $\GAP{\eps}$ {\sc Max-2-CSP}.
The {\em primal graph} $G$ of $I$ has $V(G)=V$ and $xy \in E(G)$ if and only if $R(x,y) \in C$
or $R(y,x) \in C$.

\begin{theorem}[Follows from Theorem 1.4 in Bafna et al.~\cite{Bafna:etal:stoc2025}] \label{thm:Bafna-hardness}
Assuming ETH, $\GAP{\eps}$ {\sc Max-2-CSP} is not in \FPT for any $0 < \epsilon < 1$, even
if the primal graph of instances is required to be complete.
\end{theorem}

The {\em graph density} of a simple graph $G$ is defined to be the ratio of the number of edges $|E(G)|$ with respect to the maximal number of edges, i.e. $\frac{|E(G)|}{{|V(G)| \choose 2}}$.
The previous result thus shows hardness for the case when the primal graph has maximum
density. Our first reduction transforms a $\GAP{\eps}$ {\sc Max-2-CSP} instance into a graph representation known
as the {\em microstructure graph}~\cite{Jegou:aaai93}: vertices represent variable-value 
assignments and edges describe compatibility between pairs.
The construction preserves (in a certain sense) the high density that we can assume the given {\sc Max-2-CSP} instance to have, and this allows us to produce a gap
that only depends on $\epsilon$ and the chromatic number of the input graph.

\pbDefGap{$\GAP{\eps}$-Multicoloured Densest Subgraph}
{A $k$-partite undirected graph $G$ with all parts of equal size and an integer $k$.}
{$k$.}
{$G$ contains a multicoloured $k$-clique.}
{Every multicoloured subgraph of $G$ induced on $k$ vertices 
contains at most $\eps \binom{k}{2}$ edges.}

\begin{lemma}
  \label{lem:gapmultidenshard}
  Assuming \ETH,
  \textsc{$\GAP{\eps}$-Multicoloured Densest Subgraph}
  is not in \FPT for any $0 < \eps < 1$.
\end{lemma}
 \begin{proof}
  We present an FPT-reduction from $\GAP{\eps}$ \textsc{Max-2-CSP} that runs in 
  polynomial time and preserves the parameter exactly. Let 
  $\cI=(V,D,\cC)$ be an arbitrary instance of $\GAP{\eps}$ \textsc{Max-2-CSP} such that its underlying
  primal graph is complete. We can, without loss of generality, assume that
  for every pair of distinct variables $x,y \in V$, there is exactly
  one constraint involving these two variables. Note that if both $R(x,y)$ and $S(y,x)$
  are in $C$, then we can replace these two constraints by $U(x,y)$ where $U=R \cap \{(b,a) \in D^2 \; | \; (a,b) \in S\}$.
  This implies that $\cC$
  contains exactly $\binom{|V|}{2}$ constraints.

   We construct the instance $\cI'=(G,k)$ as follows. The graph $G$ has
   one vertex $u_{v,d}$ for every $v \in V$ and $d \in D$ and the
   partition of $V(G)$ is given by $V_v=\SB u_{v,d}\SM d \in
   D\SE$. Moreover, we set $k=|V|$ and $G$ has an edge between
   $u_{v,d}$ and $u_{v',d'}$ whenever $v \neq v'$ and the assignment
   setting $v$ to $d$ and $v'$ to $d'$ does not falsify any constraint
   in $\cC$. Clearly, the reduction can be performed in polynomial-time. Towards showing the equivalence of the
   two instances, assume first that $\cI$ is a yes-instance, i.e. there
   is an assignment $\tau : V \rightarrow D$ that satisfies all
   constraints in $\cC$. Then, it is straightforward to verify that
   $\SB u_{v,\tau(v)} \SM v \in V\SE$ is a $k$-clique in $G$. 
   Now suppose that $\cI$ is a no-instance, i.e. every
   assignment of the variables in $V$ satisfies at most $\eps|\cC|=\eps\binom{k}{2}$
   constraints. But then every multicoloured set $S$ of vertices of $G$
   can contain at most $\eps\binom{k}{2}$ edges.
 \end{proof}

\newcommand{\PPP}{\mathcal{P}}

Our next step is to introduce the \gapsven (\gapsvens)
problem.
As
mentioned earlier, the aim of this problem is
to find a subgraph isomorphism to a $t$-partite graph $H$ with each
part having parameter many ($k$) vertices and whose edges can be
partitioned into $\Omega(k^2)$ $t$-cliques.
Before doing so, we first need to show that
such a graph $H$ does indeed exist and can be computed efficiently.
We will achieve this using a classical result by
Turán stated in \Cref{thm:turan} below. Before stating the theorem, we
need to introduce some notions from extremal graph theory.
For
positive integers $r \leq n$, the {\em Turán graph} $T(n,r)$ is the unique complete $r$-partite
$n$-vertex graph where each part consists of 
$\lfloor n/r \rfloor$ or $\lceil n/r \rceil$ vertices, i.e.
the vertices are partitioned as evenly as possible.
We let $t(n,r)$ denote the number of edges
in $T(n,r)$ and note that 
$$t(n,r) = (1+\frac{1}{r}+o(1)) \cdot \frac{n^2}{2}$$ for fixed $r$.
The following is a central result.

\begin{theorem}[\cite{Turan:mfl41}. Also see Chapter~7 in \cite{Diestel:GT}.]\label{thm:turan}
  Let $r \leq n$ be positive integers. Then, any $K_{r+1}$-free $n$-vertex graph has at most $t(n,r)$ edges. 
\end{theorem}

Stated differently,
$T(n,r)$ has maximum density
among graphs
which contains no $(r+1)$-clique. We use this result for analyzing how a clique
can be partitioned into smaller multicoloured cliques.
Let $H$ be a graph without isolated vertices. An $H$-{\em packing} of
a graph $G$ is a set $\{G_1,\ldots,G_n\}$ of edge-disjoint subgraphs of
$G$ that are isomorphic to $H$.

\begin{lemma} \label{lem:tcliques}
  Let $k$ and $t$ be integers with $k \geq t$.
  There is a colouring $\chi \colon [k] \to [t]$ such that the edges
  of $K_k$ can be partitioned into $\Omega(k^2/(t^3e^t))$ multicoloured
  (w.r.t. $\chi$) $t$-cliques.
\end{lemma}
\begin{proof}
  Theorem~\ref{thm:turan} tells us that
  $T(n,r)$ is the densest $n$-vertex graph which contains no $(r+1)$-clique.
  Thus, let $C_1$, \ldots, $C_p$ be a maximal edge-disjoint packing of $t$-cliques in $K_k$.
  Then the number of edges not covered by the packing is at most
  $$|E(T(k,t-1))| = \left(1-\frac{1}{t-1}+o(1) \right) \cdot \frac{k^2}{2},$$
  and the packing contains at least $p = \Omega(k^2/t^3)$ cliques.
  Now select a random $t$-partition of $[k]$. The probability that a
  given $t$-clique is multicoloured is 
  $$\frac{t!}{t^t}=e^{-(1-o(1))t}$$ by Stirling's approximation and asymptotically $\Omega(e^{-t})$. 
  Therefore there exists a $t$-partition in which
  $\Omega(k^2/(t^3e^t))$ $t$-cliques are multicoloured.
\end{proof}

The following lemma now shows the existence of the graph $H$
that is required for the definition of \gapsvens and whose edges can be can be partitioned into
$\Omega(k^2)$ $t$-cliques.

\begin{lemma} \label{lem:constHkt}
  Let $t$ be an integer. There is an integer $k_0(t)$, depending only on $t$, such
  that for every $k\geq k_0(t)$ there is an algorithm
  that in time $\bigoh(k^k)$ constructs a $t$-partite graph $H$
  with partition $\{U_1,\dotsc,U_t\}$ and $|U_i|=k$
  together with a partition $\PPP_H$ of its edge set into at least $k^2/(t^2e^t)$
  $t$-cliques. Moreover, $H$ is the disjoint union of $t$ isomorphic
  graphs.
\end{lemma}
\begin{proof}
  Consider a colouring $\tau \colon [k] \to [t]$ and 
  an edge-disjoint packing ${\cal K}_k=\{K_k[C_1], \dotsc, K_k[C_p]\}$ of
  $t$-cliques multicoloured by $\tau$ inside $K_k$. By Lemma~\ref{lem:tcliques}
  there is an integer $k_0(t)$
  such that if $k \geq k_0(t)$, then $p \geq k^2/(t^3e^t)$. 
  We note that ${\cal K}_k$ can be computed in $\bigoh(k^k)$ time by exhaustive search.
  
  Let $H'$ denote the graph
  consisting of the edge disjoint union of the graphs in ${\cal K}_k$.
  The graph $H'$ is $t$-partite as is witnessed by $\tau$.
  Construct graph $H$ by taking the disjoint union of $t$ copies of $H'$.
  Note that a $t$-partition $\{U_1,\dotsc,U_t\}$ of $H$
  with $|U_i|=k$ for every $i \in [t]$ can be obtained by letting
  every $U_i$ contain one copy of every part of $H'$.

  Finally, let $\PPP_H$ be the partition of the edges of $H$ into $pt$
  $t$-cliques consisting of $H[C_i^j]$ for every copy $C_i^j$
  of $C_i$ in $H$ (where $i\in [p]$ and $j \in [t]$). Then,
  $|\PPP_H|\geq k^2/(t^2e^t)$.
\end{proof}

In the following, let $t$ be an integer and let $k_0(t)$ be the integer whose existence is
guaranteed by \Cref{lem:constHkt}. Moreover, for every $k \geq k_0(t)$, let
$\Htk$ be the graph $H$ from \Cref{lem:constHkt}. Having shown the
existence of the graph $\Htk$, we are now ready to define
\gapsvens.
Please also refer to Figure~\ref{fig:gapsven} for an illustration of a problem
instance for \gapsvens.

\pbDefGap{\gapsven (\gapsvens)}{%
  An integer $k\geq k_0(t)$, a $t \times k$-partite graph $G=(V,E)$ with partition
  $\PPP_G=\{V_{1,1}, \dotsc,V_{t,k}\}$ into sets of equal size and
  a homomorphism $\chi:G \rightarrow H$, where $H=\Htk$, that assigns
  every vertex $v\in V_{i,j}$ to the $j$-th vertex in $U_i$ of $H$.}
{$k$.}
{There exists $S \subseteq V$ (with $|S|=|\chi(S)|=kt$) such that $\chi$ restricted to $S$ is an isomorphism between $G[S]$ and $H$.}
{For every $S \subseteq V$ with $|S|=|\chi(S)|=kt$, it holds that $G[S]$ contains at most $\eps |E(H)|$ edges.}

\begin{figure}
  \centering
  \begin{tikzpicture}[label distance=0.3cm]
    \tikzstyle{mp}=[ellipse,draw]
    \tikzstyle{mn}=[circle,draw, inner sep=2pt]
    \tikzstyle{md}=[]
    \tikzstyle{me}=[draw, line width=1pt,-latex]

    \draw
    node[mp, minimum width=0.8cm, minimum height=3cm, label=north:$V_{i,1}$] (V11) {}
    (V11.north) +(0cm,-0.4cm) node[mn, label=left:$1$] (V11v1) {}
    (V11.center) node[md] (V11d) {\textbf{\vdots}}
    (V11.south) +(0cm,0.4cm) node[mn, label=left:$n$] (V11vn) {}

    (V11.east) +(1.5cm,0cm) node[md] (V111kd) {\textbf{\dots}}
    
    (V11.east) +(3cm, 0cm) node[mp, anchor=west,minimum width=0.8cm, minimum height=3cm, label=north:$V_{i,k}$] (V1k) {}
    (V1k.north) +(0cm,-0.4cm) node[mn, label=right:$1$] (V1kv1) {}
    (V1k.center) node[md] (V1kd) {\textbf{\vdots}}
    (V1k.south) +(0cm,0.4cm) node[mn,label=right:$n$] (V1kvn) {}
    ;
    \draw
    (V111kd) +(0cm,-3cm) node[anchor=north,mp, minimum width=3cm,
    minimum height=0.8cm,label=south:$U_i$] (U1) {}
    (U1.west) +(0.4cm,0cm) node[mn, label=below:$1$] (u11) {}
    (U1.center) node[md] (u1d) {\textbf{\dots}}
    (U1.east) +(-0.4cm,0cm) node[mn, label=below:$k$] (u1k) {}
    ;
    \draw
    (V11) edge[me] (u11)
    (V1k) edge[me] (u1k)
    ;

    \draw
    (V1k.east) +(+2.5cm,0cm) node[anchor=west,mp, minimum width=0.8cm, minimum height=3cm, label=north:$V_{j,1}$] (Vj1) {}
    (Vj1.north) +(0cm,-0.4cm) node[mn, label=left:$1$] (Vj1v1) {}
    (Vj1.center) node[md] (V11d) {\textbf{\vdots}}
    (Vj1.south) +(0cm,0.4cm) node[mn, label=left:$n$] (Vj1vn) {}

    (Vj1.east) +(1.5cm,0cm) node[md] (Vj11kd) {\textbf{\dots}}
    
    (Vj1.east) +(3cm, 0cm) node[mp, anchor=west,minimum width=0.8cm, minimum height=3cm, label=north:$V_{i,k}$] (Vjk) {}
    (Vjk.north) +(0cm,-0.4cm) node[mn, label=right:$1$] (Vjkv1) {}
    (Vjk.center) node[md] (V1kd) {\textbf{\vdots}}
    (Vjk.south) +(0cm,0.4cm) node[mn,label=right:$n$] (Vjkvn) {}
    ;
    \draw
    (Vj11kd) +(0cm,-3cm) node[anchor=north,mp, minimum width=3cm,
    minimum height=0.8cm,label=south:$U_j$] (Uj) {}
    (Uj.west) +(0.4cm,0cm) node[mn, label=below:$1$] (uj1) {}
    (Uj.center) node[md] (u1d) {\textbf{\dots}}
    (Uj.east) +(-0.4cm,0cm) node[mn, label=below:$k$] (ujk) {}
    ;
    \draw
    (Vj1) edge[me] (uj1)
    (Vjk) edge[me] (ujk)
    ;
  \end{tikzpicture}
  \caption{An illustration of the parts $U_i$ and $U_j$ of $H$ and the parts
    $V_{i,1},\dotsc,V_{i,k}$ and $V_{j,1},\dotsc,V_{j,k}$ of $G$ and their relation with respect to the
    homomorphism $\chi$ from $G$ to $H$, which is represented by
    directed edges, for an instance of \gapsvens{}.}
  \label{fig:gapsven}
\end{figure}
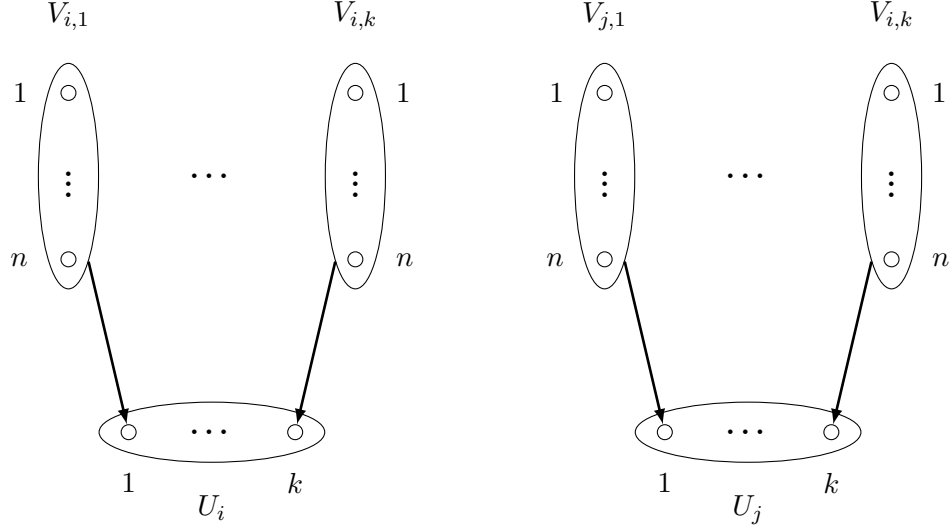

We are now ready to show hardness for \gapsvens{}.
\begin{lemma}\label{lem:gapsvenhard}
  Assuming \ETH,
  \gapsvens{} is not in \FPT for any $0 < \eps < 1$ and $t \geq 2$.
\end{lemma}
\begin{proof}
  We provide a reduction from \textsc{$\GAP{\eps}$ Multicoloured
    Densest Subgraph}.  Let $t$ be an integer and let $k_0(t)$ be the
  integer (depending solely on $t$), whose existence is guaranteed by
  \Cref{lem:constHkt}. Let $\cI=(D,k)$ be an instance of \textsc{$\GAP{\eps}$ Multicoloured
    Densest Subgraph} with $D$ having a $k$-partition represented by $\lambda \colon V(D)
  \to [k]$. We assume without loss of generality that $k\geq k_0(t)$.
  Let $H=\Htk$ be the $t$-partite graph with partition
  $\{U_1,\dotsc,U_t\}$ that consists of $t$ copies of the graph $H'$ from
  \Cref{lem:constHkt} with the partition $\PPP_H$ of its edge set into
  at least $k^2/(t^2e^t)$ $t$-cliques, which can be constructed in
  time $\bigoh(k^k)$. Without loss of
  generality we assume that $[k]$ is the vertex set of $H'$.
  We construct an equivalent instance $\cI'=(k,G,\PPP_G,\chi)$ of
  \gapsvens as follows.

 \begin{enumerate}

\item
  Let 
  $G'$ be the $k$-partite graph (with $k$-partition given by
  $\lambda$) obtained from $D$ after removing every edge $uv$ such that
  $\lambda(u)\lambda(v) \notin E(H')$. Note that $\lambda$ is a
  homomorphism from $G'$ to $H'$. Let $G$ denote the $t \cdot k$-partite graph obtained by taking the
  disjoint union of $t$ copies of $G'$.

\item
We construct the homomorphism $\chi:G \rightarrow H$ as follows.
Consider the homomorphism $\lambda$ from $G'$ to $H'$.
Recall that $G$ is the disjoint union of $t$ copies of $G'$ while
$H$ is the disjoint union of $t$ copies of $H'$.
We let $\chi$ homomorphically map the $i$:th copy of $G'$ in $G$
into the $i$:th copy of $H'$ in $H$ using $\lambda$.

\item
Finally, we construct the partition $\PPP_G$ as follows. That is, we let $V_{i,j}$ be
the set of vertices in $G$ that $\chi$ maps to the $j$:th vertex of $U_i$.
This is indeed a partitioning since $\chi$ maps every vertex in $V(G)$ to a vertex in $V(H)$.
Note that $\lambda$ maps all elements in the same colour class to the same element
in $H'$ and that the elements in distinct colour classes are assigned
distinct elements in $H'$. Since $G$ and $H$ consists of $t$ copies of
$G'$ and $H'$, respectively, it follows that $\chi$ maps exactly the
vertices of one copy of a colour class to any vertex in $H$. Therefore,
$V_{i,j}$ contains exactly all vertices from one copy of one colour
class and since all colour classes in $D$ have the same size, it
follows that all sets $V_{i,j}$ have the same size.
\end{enumerate}

\noindent

We continue by verifying the correctness of the reduction.

\medskip

\noindent
\textbf{$\cI$ is a yes-instance.}
  This implies that $D$ contains a multicoloured
  $k$-clique $D[S]$. 
  Consider the graph $G'[S]$ and recall that
  $G'[S]$ is obtained from $D[S]$ by removing every edge $uv$ such that
  $\lambda(u)\lambda(v) \notin E(H')$. That is, $G'[S]$ is isomorphic to $H'$. 
  Let $S^+$ be the set $S$ extended by all $t$ copies
  of every vertex in $S$. Then, $G[S^+]$ is
  isomorphic to $H$. The set $S^+$ contains $kt$ vertices and $\cI'$ is a
  yes-instance.

\medskip

\noindent
\textbf{$\cI$ is a no-instance.}
We show that  $\cI'$ is a no-instance by proving the contrapositive:
if $\cI'$ is not a no-instance, i.e.\ there is an $S \subseteq V$ with
$|S|=|\chi(S)|=kt$ such that
$G[S]$ contains at least $\eps |E(H)|$ edges for some $0 < \epsilon <
1$, then $\cI$ is not a no-instance, i.e.
$D$ contains a multicoloured subgraph on $k$ vertices that contains
at least $\eps' k^2$ edges, where $\eps'=\eps \cdot \binom{t}{2}/(t^3e^t)$.
Recall that $H$ is the edge-disjoint union of the $t$-cliques in
$\PPP_H$. Because $|\PPP_H|\geq k^2/(t^2e^t)$, it follows that $H$
has at least $k^2\binom{t}{2}/(t^2e^t)$.

Because $S$ is multicoloured and $G[S]$ contains 
at least $\eps |E(H)|$ edges, 
there is a copy $G^c$ of $G'$ in $G$ such that
such that $G^c[S']$, where $S'=S\cap V(G^c)$, contains at least
$\eps|E(H)|/t$ edges.
Therefore, $G^c[S']$ (and therefore also $D[S']$) contains at least
$\eps \cdot k^2\binom{t}{2}/(t^3e^t)=\eps' k^2$ edges, which shows that $\cI$ is
not a no instance and completes the proof of the theorem.
\end{proof}

\subsection{Hardness of \pSplit}
\label{sec:graph-hard}

Using the hardness for \svens\ given in \Cref{lem:gapsvenhard}, we are now ready to
show hardness for the \pSplit{} problem given below.

\pbDefGap{\pSplit}
{A graph $G$ with special vertices $s, t$ such that $G-\{s,t\}$ is the disjoint union of 
$p$ graphs $G_1,\dotsc,G_p$, 
a set of pairwise disjoint edge sets $\cB \subseteq \binom{E(G)}{1}
\cup \binom{E(G)}{p}$, each set containing at most one edge from each
graph $G_i$,
and an integer $k$.}
{$k$.}
{There is an $st$-cut $Z \subseteq E(G)$ which is the union of $k$ bundles in $\cB$.}
{Every $st$-cut $Z \subseteq E(G)$ intersects at least $(p-\eps)k$ bundles in $\cB$.}

\begin{lemma} \label{lem:gap-p-cut}
  Assuming \ETH,
  $\GAP{(p-\eps)}$ $p$-{\sc Split Min Cut} is not in \FPT for any $0 < \eps < 1$
  and $p \geq 2$.
\end{lemma}
\begin{proof}
  Let $\cI=(k,G,\PPP_G,\chi)$ be an instance of \gapsvens{}.
  Moreover, let $H=\Htk$ be the $t$-partite graph with partition
  $\{U_1,\dotsc,U_t\}$ from
  \Cref{lem:constHkt} with the partition $\PPP_H$ of its edge set into
  at least $k^2/(t^2e^t)$ $t$-cliques, which can be constructed in
  time $\bigoh(k^k)$. For convenience, let
  $U_i=\{u_{i,1},\dotsc,u_{i,k}\}$ and let
  $\PPP_G=\{V_{1,1},\dotsc,V_{t,k}\}$ with $V_{i,j}=\{v_{i,j}^1,\dotsc,v_{i,j}^n\}$.
  
  We construct an instance 
  $\cI' = (F, \hat{s}, \hat{t}, \{F_1,\dotsc, F_p\}, \cB, k')$ of 
  \pSplit{}
  with $p=t$ and $k' = |\PPP_H|\leq k^t$ (note that $|\PPP_H|\leq k^t$
  and therefore $k'$ is bounded by a function of $k$) such that
  \begin{itemize}
  \item if $\cI$ is a yes-instance, then
    $\cI'$ is a yes-instance, and
  \item if $F$ admits an $\hat{s}\hat{t}$-cut
    that intersects at most $(p - \eps)k'$ bundles in $\cB$
    for any $0 < \eps < 1$,
    then $G$ contains a multicoloured induced subgraph
    with at least $\eps |E(H)|$ edges.
  \end{itemize}
  The hardness then follows by Lemma~\ref{lem:gapsvenhard}. We now present the instance
  $\cI'$.
  
\smallskip
\noindent
\textbf{Construction of graph.}
  Let the graph $F$ have vertices $\hat{s}$, $\hat{t}$ and
  $f_{i,j}^d$
  for every $i \in [t]$, $j \in [k]$, and
  $0 \leq d \leq n$.
  Identify $f_{i,j}^0$ with $\hat{s}$ and $f_{i,j}^n$ with $\hat{t}$
  for all $i$ and $j$.
  We will now add a number of paths to the graph: this will describe the full set
  of edges but it will also introduce more auxiliary vertices.

  For every $i \in [t]$, $j \in [k]$, and $d \in [n]$, we do the following.
  Connect $f_{i,j}^{d-1}$ to $f_{i,j}^d$
  with a set of internally disjoint paths of length $n^{t-1}$,
  for every $t$-clique $C \in \PPP_H$ that contains $u_{i,j}$.
  We refer to these paths as $P_{i,j,d,C}$.
  The edges of $P_{i,j,d,C}$ are named $e_{i,j,C,x}$ for $x \in
  [n]^t$ such that $x[i]=d$ and occur on those paths
  in an arbitrary order.
  Finally, define $F_i$ as the union of 
  the paths $P_{i,j,d,C}$ for all $j,d, C$.
  Clearly, $F-\{\hat{s},\hat{t}\}$ is the disjoint union of 
$t$ vertex-disjoint graphs so the graph satisfies the requirements.

\smallskip

\noindent
\textbf{Construction of bundles.}
For $C=\{u_{1,j_1},\dotsc,u_{t,j_t}\} \in \PPP_H$ and $x \in [n]^t$,
  let $S(C,x)$ be the set containing the $x[i]$'th vertex
  $v_{i,j_i}^{x[i]}$ in $V_{i,j_i}$ for every $i \in [t]$. Note that $\chi(S(C,x))=C$.
  We let the set $\cB$ contains one bundle (of size $t$)
  $\SB e_{i,j_i,C,x}\SM i \in [t]\SE$
  for every $x \in [n]^t$ and every
  $C=\{u_{1,j_1},\dotsc,u_{t,j_t}\} \in \PPP_H$ such that $S(C,x)$
  forms a $t$-clique in $G$. 
  Each bundle contains one edge from each part $F_1,\dots,F_t$.
  Additionally, $\cB$ contains one bundle
  (of size $1$) for every edge of $F$ that does not occur in any of
  the bundles of size $t$ introduced above.

  \medskip

\noindent
We make a couple of observations before presenting the correctness proof.
  Note that a set $Z$ of edges of $F$ is an $\hat{s}\hat{t}$-cut in $F$ if and only if
  for every $i \in [t]$ and $j \in [k]$, there
  is a $d\in [n]$ such that $Z$ contains exactly one edge from every path
  $P_{i,j,d,C}$, for every $C=\{u_{1,j_1},\dotsc,u_{t,j_t}\}\in
  \PPP_H$ with $u_{i,j} \in C$. Note also that if $Z$ is a minimal
  $st$-cut, then $Z$ does not contain any other edges and $|Z|=
  t|\PPP_H|$, since for every $C \in \PPP_H$, exactly $t$ paths need
  to be cut.
  We are now ready to give the correctness proof.
  
  \medskip

\noindent
\textbf{$\cI$ is a yes-instance.}
  We show that this implies that $\cI'$ is a
  yes-instance. Arbitrarily choose $S \subseteq V(G)$
  (with $|S|=|\chi(S)|=kt$) such that the function $\chi_S$ obtained
  by projecting $\chi$ to $S$ is an isomorphism between
  $G[S]$ and $H$. Let $C=\{u_{1,j_1},\dotsc,u_{t,j_t}\} \in \PPP_H$,
  then because $\chi_S$ is an isomorphism between $G[S]$ and $H$, it holds
  that $\chi_S^{-1}(C)$ forms a $t$-clique in $G[S]$.
  
  Let $y \in [n]^{t\times k}$ be the vector such
  that $\chi_{S}^{-1}(u_{i,j})$ is the $y[i,j]$-th vertex
  $v_{i,j}^{y[i,j]}$ in
  $V_{i,j}$. Moreover, for $C=\{u_{1,j_1},\dotsc,u_{t,j_t}\} \in
  \PPP_H$, let $y_C$ be the vector $y$ restricted to the coordinates
  in $\{(i,j_i) \; | \; i \in [t]\}$. Then, the set
  $\cP(C)=\SB e_{i,j_i,C,y_C}\SM i \in [t]\SE$ is a bundle in $\cB$ for every
  $C=\{u_{1,j_1},\dotsc,u_{t,j_t}\} \in \PPP_H$ and we claim that the
  set $B=\SB
  \cP(C) \SM C \in \PPP_H\SE$ is a solution for $\cI'$. Clearly,
  $|B|\leq k'=|\PPP_H|$. Moreover, $\bigcup_{b \in B}b$ is an $st$-cut
  in $F$ since it contains the edge $e_{i,j_i,C,y_C}$ for every $i \in
  [t]$, $j\in [k]$, and $C \in \PPP_H$ with $u_{i,j}\in C$. This shows
  that $\cI'$ is a yes-instance.

\medskip

\noindent
\textbf{$\cI$ is a no-instance.}
We show that  $\cI'$ is a no-instance by proving the contrapositive: if $\cI'$ is not a no-instance,
  i.e. $F$ admits an $st$-cut 
  that intersects at most $(t - \eps)k'$ bundles in $\cB$
  for any $0 < \eps < 1$,
  then $\cI$ is not a no-instance, i.e.
  $G$ contains a multicoloured vertex set $S$ such that $G[S]$
  has at least $(\eps/t^4) |E(H)|$ edges.

  Now, suppose that there is an $\hat{s}\hat{t}$-cut $Z$ in $F$ 
  that intersects at most $(t - \eps)|\PPP_H|$ bundles in $\cB$.
  We assume without loss of generality that $Z$ is smallest possible.
  Pick $i \in [t]$ and $j \in [k]$.  There is a $d_{i,j}
  \in [n]$ such that $Z$ contains exactly one edge from every path
  $P_{i,j,d,C}$ for every $C=\{u_{1,j_1},\dotsc,u_{t,j_t}\}\in
  \PPP_H$ with $u_{i,j} \in C$. 
  Let $S=\{v_{i,j}^{d_{i,j}} \; | \; i \in
  [t] \land j \in [k]\}$ and let $\chi_S$ be the projection of $\chi$
  to $S$. Note that $\chi_S$ is bijective and again because $Z$ is
  minimal, $|Z|=t|\PPP_H|$.
  
  We claim that the $G[S]$
  contains at least $\eps / t^4 |E(H)|$ edges.
  Since $Z$ intersects at most $(t - \eps)|\PPP_H|$ bundles in $\cB$,
  there must be at least $\eps/t |\PPP_H|$ bundles
  which $Z$ intersects in at least two edges.
  Let $B=\SB e_{i,j_i,C,x}\SM i \in [t]\SE$ be such a bundle,
  where $C=\{u_{1,j_1},\dotsc,u_{t,j_t}\} \in
  \PPP_H$  and $x \in [n]^t$ such that $S(C,x)$ is a $t$-clique of
  $G$, and let
  $e_{i,j_i,C,x}, e_{i',j_{i'},C,x} \in Z \cap B$.
  Because $S(C,x)$ is a $t$-clique of $G$, we obtain that the vertices
  $v_{i,j_i}^{x[i]}$ and $v_{i',j_{i'}}^{x[i']}$ are adjacent in $G$.
  Therefore, $\chi^{-1}_S(C)$ contains at least one edge
  $e_C=\{v_{i,j_i}^{x[i]}, v_{i',j_{i'}}^{x[i']}\}$. Moreover, since
  $\chi_S$ is bijective and $\PPP_H$ partitions the edges of $H$, it
  follows that $e_C$ and $e_{C'}$ are distinct for every two distinct
  cliques $C$ and $C'$ in $\PPP_H$.
  Since $Z$ intersects at most $t$ bundles belonging to any clique in
  $\PPP_H$, $Z$ must intersect at least $\eps |\PPP_H| / t^2$ bundles
  corresponding to distinct cliques with at least two edges.
  Therefore, $G[S]$ contain at least 
  $$\frac{\eps \cdot |\PPP_H|}{t^2} \geq \frac{\eps}{
  t^4} \cdot |E(H)|$$ edges since $|E(H)|=\binom{t}{2} \cdot |\PPP_H|$.
\end{proof}

\subsection{Hardness of Linear Equations}
\label{sec:equation-hard}

We finally turn our attention to equations over $\ZZ_m$.
 We consider a more general class of rings since it makes
the proof conceptually easier.
Let $R$ be a finite commutative ring with additive identity 0 and
multiplicative identity 1. 
An element
$r \in R$ is {\em idempotent} if $r^2=r$, and two idempotents $r,s \in R$ are
{\em orthogonal} if $rs=0$.
The elements $0$ and $1$ are always idempotent so
we say that an idempotent $r \in R$ is {\em non-trivial} if $r \not\in \{0,1\}$.
Idempotents are intimately connected with various decompositions of rings.
Let
$R_1,\dots,R_p$ denote nontrivial rings (i.e. they contain at least two elements each) where $0_{R_i},1_{R_i}$ denote the additive and multiplicative identity, respectively, in ring $R_i$.
Assume $R$ is isomorphic to the direct sum $R_1 \oplus \dots \oplus R_p$
and consider the set
\[S=\{(0_{R_1},\dots,0_{R_{i-1}},1_{R_i},0_{R_{i+1}},\dots,0_{R_p}) \; | \; 1 \leq i \leq p\}\]
We see that $S$ (and thus $R$) contains $p$ orthogonal non-trivial idempotents.
In particular, this shows that $\ZZ_m$ contains a set of orthogonal non-trivial idempotents of size $\omega(m)$ since (by Proposition~\ref{prop:crt})
$\ZZ_m$ is isomorphic to the direct sum
$\bigoplus_{i=1}^{\ell} \ZZ_{p_i^{n_i}}$ where $p_1^{n_1} \cdots p_\ell^{n_\ell}$ is the prime factorization of $m$
and $\ell=\omega(m)$. We exemplify with the ring $\ZZ_{30}$ where
Proposition~\ref{prop:crt} gives us an isomorphism $f:\ZZ_{30} \rightarrow \ZZ_2 \oplus \ZZ_3 \oplus \ZZ_5$
defined by $x \mapsto (x \mod 2, x \mod 3, x \mod 5)$.
It is easy to verify that $f(15)=(1,0,0)$, $f(10)=(0,1,0)$, and $f(6)=(0,0,1)$ so $\{6,10,15\}$ is
a set of orthogonal non-trivial idempotents in $\ZZ_{30}$ of size 3.

Given a ring $R$ with a set $\{q_1,\dots,q_p\}$ of orthogonal idempotents,
the reduction starts from an instance $I$ of \pSplit with graph $G$ and bundle set $\cB$.
The parts $G_i$, $i \in [p]$, of the graph $G$ are associated with
idempotent $q_i$ and 
transformed into sets of equations.
Let us consider the set of equations connected to $q_i$.
It is initially not satisfiable but if a suitable set of equations is removed
from it, then
it has a solution that only uses values $0$ and $q_i$. This is exploited for
expressing the bundles in $\cB$: the bundles cut across the parts of $G$
and the orthogonality of $q_1,\dots,q_p$ gives us a mechanism for synchronizing
the removal of equations in the resulting instance of $\GAP{p-\eps}$ $\minlin{2}{R}$.

\begin{theorem} \label{thm:incomparable-annihilators}
Let $R$ denote a finite commutative ring that contains a set $Q=\{q_1,\dots,q_p\}$ of orthogonal
non-trivial idempotents.
  Assuming \ETH, $\GAP{(p-\eps)}$-$\minlin{2}{R}$ is not in \FPT{} for any $\epsilon > 0$.
\end{theorem}
\begin{proof}
  We provide an fpt-reduction from \pSplit. The reduction preserves solution cost exactly
  which gives us the result due to Lemma~\ref{lem:gap-p-cut}.
  Let $\cI=(G, s, t, G_1,\dots,G_p,k, \cB)$ be an instance of \pSplit.
  We construct an instance $\cI'=(S, k)$ of \minlin{2}{R} as follows.
Introduce four sets of variables.

\begin{itemize}
\item
$V_G$ contains one variable for every vertex in $V(G) \setminus \{s,t\}$ with the same name.

\item
$V_s=\{s_1, \dots, s_p\}$.

\item
$V_t=\{t_1, \dots, t_p\}$.

\item
$V_{\cal P}$ contains two variables $x_B,y_B$ for every bundle $B \in
\cB$.
\end{itemize}

\medskip

\noindent
Let the variable set $V$ be the union of $V_G$, $V_s$, $V_t$, and $V_{\cB}$.
We now describe the equations in $S$. 
First of all, we introduce
{\em source/sink} equations. These are
crisp equations $s_i=q_i$ and $t_i=0$, $1 \leq i \leq p$.
For every edge $su \in G$ with $u \in G_i$, we add
a crisp equation $s_i=u$. Similarly, for every edge $ut \in G$ with $u$ in
$G_i$, we add a crisp equation $u=t_i$.

For every edge $uv \in E(G)$ 
that is {\em not} present in any bundle in $\cB$, we introduce a crisp {\em fixed-edge}  equation $u=v$.
For every bundle $B = (u_1 v_1, u_2 v_2,\dots,u_p,v_p) \in \cB$, we introduce the central
gadget in the reduction: it consists of
  one soft equation $x_{B} = y_{B}$
  together with a collection of crisp equations
  \begin{align*}
    &q_1u_1 = q_1x_{B} && q_1y_{B} = q_1v_1, \\
    &q_2u_2 = q_2x_{B} && q_2y_{B} = q_2v_2 \\
    &\vdots && \vdots\\
    &q_pu_p = q_px_{B} && q_py_{B} = q_pv_p.
  \end{align*}
  The gadget construction is illustrated in Figure~\ref{fig:gadget}.

\medskip

  \noindent
  We continue by proving the correctness of the reduction. 
  
  \medskip

  \noindent
  \textbf{$\cI$ is a yes-instance.}
  Let $Z$ be an $st$-cut in $G$ which is a union of $k$ bundles in $\cB$.
  We will show that $(S, k)$ is a yes-instance.
  Let $X \subseteq \cP$ be the set of tuples such that $Z = \bigcup X$, and
  define the set of equations $X' = \{ x_t = y_t \mid t \in X \}$.
  Note that $|X'| = |X| = k$ and $X'$ only contains soft equations.
  We argue that $S - X'$ is satisfied by an assignment $\alpha : V \rightarrow R$ defined as follows.
  For $v \in V_G$, we let

  \[
    \alpha(v) = 
    \begin{cases}
      q_1 & \text{if } v \in V(G_1) \text{ and } s \text{ reaches } v \text{ in } G - Z, \\
      q_2 & \text{if } v \in V(G_2) \text{ and } s \text{ reaches } v \text{ in } G - Z, \\
      \vdots \\
      q_p & \text{if } v \in V(G_p) \text{ and } s \text{ reaches } v \text{ in } G - Z, \; \text{and} \\
      0  & \text{otherwise}.
    \end{cases}
  \]
  For $1 \leq i \leq p$, we set $\alpha(s_i)=q_i$ and $\alpha(t_i)=0$.
 For the variables in $V_{\cal B}$, we do the following. Arbitrarily choose a bundle
 $B=(u_1v_1,\dots,u_pv_p) \in \cB$ and let
  \[\alpha(x_B) = \alpha(u_1) + \alpha(u_2) + \dots + \alpha(u_p), \;\alpha(y_B) = \alpha(v_1) + \alpha(v_2) + \dots + \alpha(v_p).\]
We verify that $\alpha$ satisfies the equations in $S \setminus X'$.
We first observe the following.

\medskip

\noindent
(OBS): Assume the edge $uv$ is in $G-Z$. If $s$ reaches $u$ in $G$, then $s$ reaches $v$ in $G$, too, and $\alpha(u)=\alpha(v)=q_j$
for some $1 \leq j \leq p$. If $s$ does not reach $u$ in $G$, then $s$ does not reach $v$
in $G$ and $\alpha(u)=\alpha(v)=0$. Hence, $\alpha(u)=\alpha(v)$ in this case.

\medskip

\noindent
We proceed with the proof.
The assignment $\alpha$ obviously satisfies 
the basic source/sink equations
$s_i=q_i$ and $t_i=0$, $1 \leq i \leq p$. Assume the edge $su$ is in $G-Z$
with $u \in G_i$. Then $\alpha(s_i)=q_i=\alpha(u)$ and the corresponding
source/sink equation $s_i=u$ is satisfied. 
Assume the edge $ut$ is in $G-Z$ with $u \in G_i$. 
Since $Z$ is an $st$-cut, $s$ does not reach
$t$ in $G-Z$ so $\alpha(u)=0=\alpha(t_i)$, and the source/sink equation $u=t_i$
is satisfied.

We continue with the fixed-edge equations.
Arbitrarily choose $u=v \in S$ where edge $uv \in E(G)$  does not appear in any bundle
   in $\cB$. The edge $uv$ is left in $G-Z$ and (OBS) implies that $\alpha(u)=\alpha(v)$.

   Finally, we consider the gadget equations.
 Pick a bundle
$B = (u_1 v_1, \dots, u_pv_p) \in {\cal B}$ and consider its corresponding gadget.
  Note that $\alpha(u_i),\alpha(v_i) \in \{0,q_i\}$ for
  every $1 \leq i \leq p$ since $G_1,\dots,G_p$ is a partitioning of $G-\{s,t\}$.
  It is straightforward to verify that
  $\alpha$ satisfies the crisp equations
$q_iu_i = q_ix_{B}$ and $q_iy_{B} = q_iv_i$ for arbitrary
$1 \leq i \leq p$.
Let us exemplify by the equation $q_1u_1 = q_1x_{B}$. After applying $\alpha$, we have
\[q_1\alpha(u_1) = q_1(\alpha(u_1)+\dots+\alpha(u_p)).\]
The element $q_1$ is orthogonal to $q_2,\dots,q_p$ so the equation simplifies to $q_1\alpha(u_1)=q_1\alpha(u_1)$ and $\alpha$ satisfies $q_1u_1 = q_1x_{B}$.

Up to this point, $\alpha$ satisfies all gadget equations.
If the bundle $B$ is in $X$, then there is nothing to prove so we assume $B \not\in X$.
  Then, the edges $u_1 v_1, \dots, u_p v_p$ are present in $G - Z$ and
(OBS) implies that $\alpha(u_i)=\alpha(v_i)$, $1 \leq i \leq p$.
By the definition of $\alpha$, we see that 

\[\alpha(x_{B})=\alpha(u_1)+\dots+\alpha(u_p)=\alpha(v_1)+\dots+\alpha(v_p)
=\alpha(y_{B}).\]
We conclude that $S \setminus X'$ is satisfiable and $\cI'=(S,k)$ is a yes-instance.

\medskip

  \noindent
  \textbf{$\cI'$ is a yes-instance.}
  Suppose there exists a subset
  $Y$ of soft equations such $S - Y$ is satisfiable.
  We claim that $G$ admits an $st$-cut that intersects $|Y|$ 
  bundles in $\cB$ and $\cI$ is a yes-instance.
  By definition, the only soft equations in $S$ are the gadget equations $x_B = y_B$ for bundles $B \in \cB$.
  Define $Y' = \bigcup \{ t \mid (x_t = y_t) \in Y\}$ and 
  note that $Y'$ intersects exactly $|Y|$ tuples in $\cB$.
  We claim that $Y'$ is an $st$-cut in $G$.
  Suppose, with the aim of obtaining a contradiction, that $G_1 - Y'$ contains an $st$-path
  (the proofs for $G_2,\dots,G_p$ are analogous).
  For every edge $uv$ on the path, we first verify that the equations in $S - Y$
  imply $q_1u = q_1v$.
  If $uv$ does not appear in any bundle in $\cB$, then this is enforced by the
  crisp fixed-edge equation
  $u = v$. If $uv$ appears in some bundle $B \in \cB$, then the
  gadget for $B$ contains equations
   $q_1u = q_1x_{B}$, $x_{B} = y_{B}$, $q_1y_{B} = q_1v$
 that together imply $q_1u=q_1v$.
  Thus, the equations in $S - Y$ imply
  $q_1s_1 = q_1t_1$.
  We know that $q_1$ is a non-trivial idempotent so
  $q_1^2 = q \neq 0$ and this makes
  $q_1s_1 = q_1t_1$ incompatible
  with the crisp source/sink equations
  $s_1 = q_1$ and $t_1 = 0$, and this completes the proof.
\end{proof}

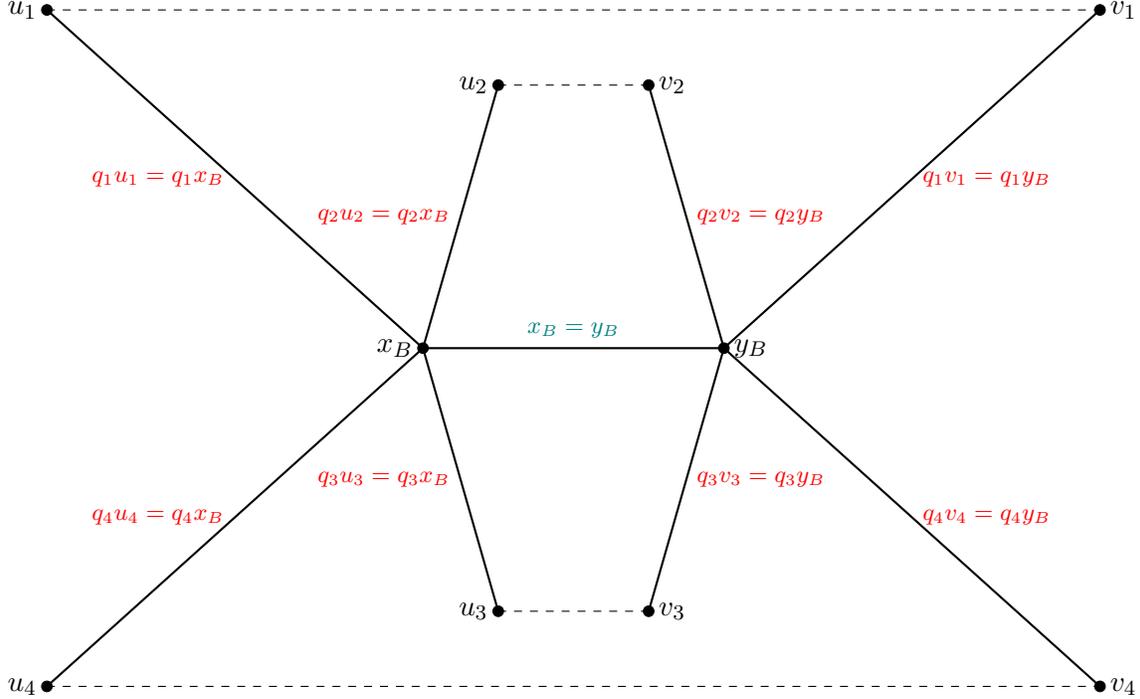
\begin{figure}
 \centering
  \begin{tikzpicture}
    \coordinate (u1) at (-5,6);
    \coordinate (v1) at (9,6);
    \coordinate (u2) at (1,5);
    \coordinate (v2) at (3,5);
    \coordinate (u3) at (1,-2);
    \coordinate (v3) at (3,-2);
    \coordinate (u4) at (-5,-3);
    \coordinate (v4) at (9,-3);
    \coordinate (xp) at (0,3/2);
    \coordinate (yp) at (4,3/2);

    \filldraw[black] (u1) circle (2pt) node[anchor=east]{$u_1$};
    \filldraw[black] (v1) circle (2pt) node[anchor=west]{$v_1$};
    \filldraw[black] (u2) circle (2pt) node[anchor=east]{$u_2$};
    \filldraw[black] (v2) circle (2pt) node[anchor=west]{$v_2$};
    \filldraw[black] (u3) circle (2pt) node[anchor=east]{$u_3$};
    \filldraw[black] (v3) circle (2pt) node[anchor=west]{$v_3$};
    \filldraw[black] (u4) circle (2pt) node[anchor=east]{$u_4$};
    \filldraw[black] (v4) circle (2pt) node[anchor=west]{$v_4$};
    \filldraw[black] (xp) circle (2pt) node[anchor=east]{$x_B$};
    \filldraw[black] (yp) circle (2pt) node[anchor=west]{$y_B$};

    \draw[dashed] (u1) -- (v1);
    \draw[dashed] (u2) -- (v2);
    \draw[dashed] (u3) -- (v3);
    \draw[dashed] (u4) -- (v4);
    \draw[thick] (u1) -- (xp) node[midway, left,  color=red] {\footnotesize{$q_1u_1 = q_1x_{B}$}};
    \draw[thick] (u2) -- (xp) node[midway, left,  color=red] {\footnotesize{$q_2u_2 = q_2x_{B}$}};
    \draw[thick] (xp) -- (yp) node[midway, above, color=teal] {\footnotesize{$x_{B}=y_{B}$}};
    \draw[thick] (yp) -- (v1) node[midway, right, color=red] {\footnotesize{$q_1v_1 = q_1y_{B}$}};
    \draw[thick] (yp) -- (v2) node[midway, right, color=red] {\footnotesize{$q_2v_2 = q_2y_{B}$}};

    \draw[thick] (u3) -- (xp) node[midway, left, color=red] {\footnotesize{$q_3u_3 = q_3x_{B}$}};
    \draw[thick] (yp) -- (v3) node[midway, right, color=red] {\footnotesize{$q_3v_3 = q_3y_{B}$}};
    \draw[thick] (u4) -- (xp) node[midway, left, color=red] {\footnotesize{$q_4u_4 = q_4x_{B}$}};
    \draw[thick] (yp) -- (v4) node[midway, right, color=red] {\footnotesize{$q_4v_4 = q_4y_{B}$}};

  \end{tikzpicture}
 \caption{System of equations obtained from a tuple of edges $\{u_1v_1,\dots,u_4v_4\} \in {\cal P}$. Red equations are crisp and the green equation is soft.
 Edges are illustrated by dashed lines.}
 \label{fig:gadget}
\end{figure}

\begin{corollary} \label{cor:modular-hardness}
  Assume that the \ETH holds.
 If $R$ is a finite and commutative ring that is the direct sum
 $R_1 \oplus \dots \oplus R_p$ of non-trivial rings, then $\minlin{2}{R}$
  is not FPT-approximable within $p-\epsilon$ for any $\epsilon > 0$.
  In particular, $\minlin{2}{\ZZ_m}$ is not FPT-approximable within $\omega(m)-\epsilon$ for any $\epsilon > 0$.
\end{corollary}
\begin{proof}
The direct sum $R_1 \oplus \dots \oplus R_p$ of non-trivial rings
contains a set of $p$ orthogonal idempotents (as explained in the beginning of Section~\ref{sec:equation-hard}) so $\minlin{2}{R}$
  is not FPT-approximable within $p-\epsilon$ for any $\epsilon > 0$
  by Theorem~\ref{thm:incomparable-annihilators}.
  The bound for $\ZZ_m$ follows immediately since it is the direct sum of $\omega(m)$
  non-trivial rings by Proposition~\ref{prop:crt}.
\end{proof}

\section{Conclusion}
\label{sec:discussion}

We prove that $\minlin{2}{\ZZ_m}$ is FPT-approximable within factor $\omega(m)$ of the optimum,
and show that FPT-approximation within $\omega(m) - \eps$
for any $\eps > 0$ is ruled out under ETH.
The core part of the algorithm is an FPT algorithm for $\minlin{2}{\ZZ_{p^d}}$ for every prime $p$ and $d \geq 1$.
For the lower bound, we define a structurally simple problem \textsc{$q$-Split Min Cut} and show that it cannot be FPT-approximated within a factor of $q-\varepsilon$ for any $\varepsilon>0$ under ETH. 

We remark that the algorithmic result is slightly stronger, and it follows that
$\minlin{2}{\ZZ_{p^d}}$ is in FPT parameterized by $k$, $p$ and $d$.
Furthermore,
the algorithm can be extended to handle the weighted variant of the problem,
where the input comes with a function $w : C \to \ZZ_+$ assigning
positive integer weights to every equation, and a weight budget $W$,
which is \emph{not} part of the parameter;
the goal is to find a solution $Z$ such that $|Z| \leq k$ and $w(Z) \leq W$.
Such variants of other FPT problems have been considered previously~\cite{Kim:etal:sicomp2025,KimPSW22weighted}
and often require additional algorithmic insights.
In particular, the FPT algorithm for \textsc{Multicut} by Marx and Razgon~\cite{marx2014fixed}, where the shadow covering procedure was introduced,
fails to handle weights due to a greedy aspect in the ``shadow removal'' step. 
We note that our algorithm, despite using a generalization of shadow covering, has no such greedy aspect.

The running time of our algorithm is $O^*(2^{k^{O(1)}})$ where the exponent degree is non-trivial to analyze but we know it is 
at least 4.  Can this be improved? The current bottleneck is the algorithm for \MinSat{} that we rely upon;
thus one improvement would be to replace this algorithm by a special-purpose FPT algorithm for our problem. Apart from this, the next bottlenecks are $2^{O(k^2)}$ from balanced subgraph covering and $2^{O(k)}$ for everything else (assuming that $p$ and $d$ are constants). 

The following is a related question: can the success probability of the balanced subgraph covering procedure be improved? The success probability of shadow covering, which is a special case, was recently improved to $1/2^{O(k \log k)}$ by Chu et al.~\cite{chu2026faster}, using a result of Korhonen and Lokshtanov~\cite{korhonen2023improved} that important separators have small \emph{blocking sets}, i.e.\ for any graph $G$ and disjoint vertex sets $X$ and $Y$ in $G$, there is a set of size $k$ that intersects all important $(X,Y)$-separators of size at most $k$.
A possible generalization to balanced important subgraphs would be the following question. Given a biased graph $(G,\BB)$ and a vertex $v \in V(G)$, is there a set $Z$ of $O(k)$ edges such that every important connected subset $S$ in $(G,\BB)$ with $v \in S$ has $\delta(S) \cap Z \neq \emptyset$? 
This is an intriguing possibility, which could improve the success probability of balanced subgraph covering to $1/2^{O(k \log k)}$.

For a much more ambitious research direction, the problems we study are instances of the problem class  $\MinCSP(\Gamma)$ for finite-domain constraint languages $\Gamma$. 
It is natural to ask for which languages $\Gamma$ the problem $\MinCSP(\Gamma)$ admits an FPT algorithm, or a constant-factor FPT-approximation.
While this is fully settled in the Boolean domain~\cite{bonnet2016fixed,Kim:etal:sicomp2025},
in the more general setting the question is very much open; only preliminary or ad-hoc cases are known (see~\cite{Kratsch:etal:compscirev2026} for an overview).
For constant-factor FPT-approximation, it would be particularly interesting to obtain matching upper and lower bounds for broader classes of problems than
those studied in this paper.

With the \textsc{MinCSP} motivation in mind,
perhaps the easiest extension of \textsc{Min-2-Lin} that our algorithm does not cover is the \emph{conservative}, a.k.a.\ \emph{list} case of $\minlin{2}{\ZZ_m}$, where in addition every variable $v$ has a list of values $D_v \subseteq \ZZ_m$ that it is allowed to take. These can be treated as unary constraints, or as a separate part of the problem formulation; this makes no difference to the problem complexity.
Is \textsc{List $\minlin{2}{\ZZ_{p^d}}$} FPT?

Finally, our results for \pSplitng{} can be used to prove optimality of constant factors achieved by FPT-approximation algorithms, as exemplified by the case of $\minlin{2}{\ZZ_m}$.
For example, they directly imply that the approximation factor of $2$ in~\cite[Theorem~10]{Dabrowski:etal:ipec2023} cannot be improved under the ETH.
Many more W[1]-hardness reductions in the literature can be rethought as reductions from \pSplitng{} for some $p \geq 2$, and it is natural to ask whether the lower bounds we obtain match the factors achieved by the best known FPT-approximation algorithms for these problems.
In this regard, one could consider e.g. 
\textsc{Directed Odd Cycle Transversal}~\cite{lokshtanov2020parameterized},
\textsc{Directed Multicut with $t$ Terminal Pairs} for $t \geq 4$~\cite{pilipczuk2018directed}, and
\textsc{Directed Symmetric Multicut}~\cite{EibenRW22ipec,OsipovPW24pointalgebra}.

\section*{Acknowledgements}

The second author was partially supported by 
the Swedish Research Council (VR)
under grant 2021-04371.
The fourth author was supported by VR
under grant 2024-00274.

\bibliographystyle{alpha}
\bibliography{references}

\end{document}